\documentclass[11pt]{article}

\usepackage{algorithm}
\usepackage[noend]{algpseudocode}
\usepackage{amsmath,amssymb,amsthm}
\usepackage{caption}
\usepackage{cite}
\usepackage[margin=0.95in]{geometry}
\usepackage{graphicx}
\usepackage{hyperref}
\usepackage{subcaption} 
\usepackage{tikz}
\usepackage{xcolor}

\newcommand{\Gtri}{\ensuremath{G_{\Delta}}}  
\newcommand{\M}{\mathcal{M}}
\newcommand{\A}{\mathcal{A}}
\newcommand{\cbound}{2+\sqrt{2}}  
\newcommand{\Os}{\Omega_{\mathcal{P}}}
\newcommand{\Zs}{Z_{\mathcal{P}}}
\newcommand{\pis}{\pi_\mathcal{P}}
\newcommand{\Ghex}{G_{hex}}
\newcommand{\mH}{\mathcal{V}}
\newcommand{\mV}{\mathcal{V}}
\newcommand{\Fout}{F_{ext}}
\newcommand{\mL}{\mathcal{L}}
\newcommand{\mC}{\mathcal{C}}
\newcommand{\mS}{\mathcal{S}}

\renewcommand{\P}{\mathcal{P}}
\newcommand{\Q}{\psi}
\newcommand{\Osc}{\Os^{c_1}}

\newcommand{\sthree}{1.732050808}  
\newcommand{\gsep}{4^{5/4}}
\newcommand{\csep}{0.0001}
\newcommand{\threecsep}{0.0003}
\newcommand{\fourcsep}{\threecsep}
\newcommand{\lgsep}{6.83}
\newcommand{\zint}{0.0125}
\newcommand{\aint}{10^{-5}}
\newcommand{\threeaint}{0.00003}
\newcommand{\gintlower}{79/81}
\newcommand{\gintlowerf}{\frac{79}{81}}
\newcommand{\gintupper}{81/79}
\newcommand{\lgint}{6.83}
\newcommand{\threea}{0.00003}
\newcommand{\ethreea}{1.00003}

\newcommand{\lint}{1}
\newcommand{\gintl}{79/81}

\newcommand{\gintu}{81/79}

\newtheorem{thm}{Theorem}[section]
\newtheorem{lem}[thm]{Lemma}
\newtheorem{cor}[thm]{Corollary}

\newtheorem{defn}[thm]{Definition}

\newtheorem{property}{Property}
\newtheorem{claim}[thm]{Claim}

\newif\ifcomment
\commenttrue  

\title{A Local Stochastic Algorithm for Separation in Heterogeneous Self-Organizing Particle Systems}
\author{Sarah Cannon\thanks{This material is based upon work supported by the National Science Foundation under Award No. DMS-1803325.} \and
Joshua J.\ Daymude\thanks{Daymude and Richa are supported in part by NSF awards CCF-1422603, CCF-1637393, and CCF-1733680.} \and
Cem Gokmen\thanks{Gokmen is supported in part by NSF award CCF-1733812. Randall is supported in part by NSF awards CCF-1526900, CCF-1637031, and CCF-1733812.} \and
Dana Randall\footnotemark[3] \and
Andr\'ea W.\ Richa\footnotemark[2]}
\date{\small{\footnotemark[1] University of California, Berkeley, \texttt{sarah.cannon@berkeley.edu}\\
\footnotemark[2] Computer Science, CIDSE, Arizona State University, \texttt{\{jdaymude,aricha\}@asu.edu}\\
\footnotemark[3] Georgia Institute of Technology, \texttt{cgokmen@gatech.edu}, \texttt{randall@cc.gatech.edu}}
\vspace{-0.38in}}

\begin{document}

\maketitle
\setcounter{footnote}{0}  

\begin{abstract}
We present and rigorously analyze the behavior of a distributed, stochastic algorithm for \emph{separation} and \emph{integration} in \emph{self-organizing particle systems}, an abstraction of programmable matter.
Such systems are composed of individual computational \emph{particles} with limited memory, strictly local communication abilities, and modest computational power.
We consider {\it heterogeneous} particle systems of two different colors and prove that these systems can collectively {\it separate} into different color classes or {\it integrate}, indifferent to color.
We accomplish both behaviors with the same fully distributed, local, stochastic algorithm.
Achieving separation or integration depends only on a single global parameter determining whether particles prefer to be next to other particles of the same color or not; this parameter is meant to represent external, environmental influences on the particle system.
The algorithm is a generalization of a previous distributed, stochastic algorithm for {\it compression} (PODC '16), 
which can be viewed as a special case of separation where all particles have the same color.
It is significantly more challenging to prove that the desired behavior is achieved in the heterogeneous setting, however, even in the bichromatic case we focus on.
This requires combining several new techniques, including the {\it cluster expansion} from statistical physics, a new variant of the {\it bridging} argument of Miracle, Pascoe and Randall (RANDOM~'11), 
the {\it high-temperature expansion} of the Ising model, and careful probabilistic arguments.
\end{abstract}

\section{Introduction} \label{sec:intro}

Across many disciplines spanning computational, physical, and social sciences, heterogeneous systems self-organize into both separated (or segregated) and integrated states.
Examples include molecules exhibiting attractive and repulsive forces, distinct types of bacteria competing for resources while collaborating towards common goals (e.g.,~\cite{Stewart2008,Wei2015}), social insects tolerating or aggressing towards those from other colonies (e.g.,~\cite{Roulston2003,Johnson2011}), and inherent human biases that influence how we form and maintain social groups (e.g.,~\cite{Turner1981,Hogg1985}).
In each of these, individuals are of different ``types'': integration occurs when the ensemble gathers together without much preference about the type of their neighbors, while separation occurs when individuals cluster with others of the same type.
Here, we investigate these fundamental behaviors of separation or integration as they apply to \emph{programmable matter}, a material that can alter its physical properties based on user input or stimuli from its environment.
Instead of studying a particular instantiation of programmable matter, of which there are many~\cite{Adleman1994,Correa2015,Thakker2014,Rubenstein2014}, we abstractly envision these systems as collections of simple, active computational \emph{particles} that individually execute local distributed algorithms to collectively achieve some emergent behavior.
We consider \emph{heterogeneous} particle systems in which particles have immutable \emph{colors}.
We seek local, distributed algorithms that, when run by each particle independently and concurrently, result in emergent, self-organizing \emph{separation} or \emph{integration} of the color classes.

This work uses the \emph{stochastic approach to self-organizing particle systems} first used for \emph{compression}, where (monochromatic) particles self-organize to gather together as tightly as possible~\cite{Cannon2016}.
Using this stochastic approach, one first defines an energy function where desired configurations have the lowest energy values.
One then designs a Markov chain whose long run behavior favors these low energy configurations.
This Markov chain is carefully designed so that all its transition probabilities can be computed locally, allowing it to be translated to a fully local distributed algorithm each particle can run independently.
The resulting collective, emergent behavior of this distributed algorithm is thus described by the long run behavior of the Markov chain.
Using this stochastic approach, we previously extended our compression algorithm~\cite{Cannon2016} to an algorithm for \emph{shortcut bridging}~\cite{AndresArroyo2018} --- or maintaining bridge structures that balance the tradeoff between bridge efficiency and cost --- and developed the theoretical basis for an experimental study in swarm robotics~\cite{Savoie2018}.
While the process of designing distributed algorithms for self-organizing particle systems via this stochastic approach is fairly well-understood, proving that such algorithms achieve their desired objectives remains quite challenging.
In particular, it is not enough to know the desired configurations have the highest long-run probability; there may be so many other, lower probability configurations that they collectively outweigh the desirable ones.
This energy/entropy trade-off has been studied in various Markov chains for the purposes of proving slow mixing, but we analyze it directly to show our algorithms achieve the desired objectives with high probability.

Here, we focus on separation and integration in heterogeneous systems.
Our inspiration comes from the classical Ising model in statistical physics~\cite{Ising1925,Vinkovic2006}, where the vertices of a graph are assigned positive and negative ``spins'' and there are rules governing the probability that adjacent vertices have the same spin.
Connected to the Ising model is classical work from stochastic processes on the Schelling model of segregation~\cite{Schelling1969,Schelling1971}, which explores how individuals' micro-motives can induce macro-level phenomena like racial segregation in residential neighborhoods.
Recent variants of this model from computer science have investigated the degree of individual bias required to induce such segregation~\cite{Bhakta2014,Immorlica2017}, and a related distributed algorithm has been developed~\cite{Omidvar2017}.
Our work differs from those on the Ising and Schelling models because of natural physical constraints on dynamic systems of self-organizing, active particles like ours.
For example, if we consider particles of one color to be vertices with positive spin and particles of another color to be vertices with negative spin, this is an Ising model but on a graph that is constantly changing as particles move.
Despite these obstacles, we are still able to apply ideas developed for rigorously analyzing the Ising and similar models to prove our distributed algorithm for separation and integration accomplishes the desired goals.

While we are interested in distributed algorithms, it is worth noting that efficient stochastic algorithms for separation can be challenging even when we have a centralized Markov chain.
Separation of a region into a constant number of equitably sized, compact districts has been widely explored lately in the context of gerrymandering, where the aim is to sample colorings of a weighted graph from an appropriately defined stationary distribution~\cite{Duchin2018,Herschlag2018}.
Heuristics for random districting have been discussed in the media, but there are still no known rigorous, efficient algorithms.

\subsection{Results} \label{subsec:results}

We present a distributed algorithm for self-organizing separation and integration that takes as input two bias parameters, $\lambda$ and $\gamma$.
Setting $\lambda > 1$ corresponds to particles favoring having more neighbors; this is known to cause compression in homogeneous systems when $\lambda$ is large enough~\cite{Cannon2016}.
For separation in the heterogeneous setting, we introduce a second parameter $\gamma$, where $\gamma > 1$ corresponds to particles favoring having more neighbors \emph{of their own color}.
We then investigate for what values of $\lambda$ and $\gamma$ our algorithm yields compression and separation.
Informally, a particle system is separated if there is a subset of particles such that $(i)$ the boundary between this subset and the rest of the system is small, $(ii)$ a large majority of particles in this subset are of the same color, say $c$, and $(iii)$ very few particles with color $c$ exist outside of this subset.
This notion of separation (defined formally in Definition~\ref{defn:sepinformal}) captures what it means for a system to have large monochromatic regions of particles.

We prove that for any $\lambda > 1$ and $\gamma > 4^{5/4} \sim 5.66$ such that $\lambda\gamma > 2(\cbound)e^{\fourcsep} \sim \lgsep$, our algorithm accomplishes separation with high probability.\footnote{We say an event $A$ occurs with high probability (w.h.p.) if $\Pr[A] \geq 1 - c^{n^\delta}$, where $0 < c < 1$ and $\delta > 0$ are constants and $n$ is the number of particles. Our w.h.p.\ results all have $\delta \in \{1/2, 1/2 - \varepsilon\}$, for arbitrarily small $\varepsilon > 0$.}
However, we prove the opposite for some values of $\gamma$ close to one; counterintuitively, this even includes some values of $\gamma > 1$, the regime where particles favor having like-colored neighbors.
Formally, we prove that for any $\lambda > \lint$ and $\gamma \in (\gintl,\gintu)$
such that $\lambda(\gamma+1) > 2(2+\sqrt{2})e^{\threea} \sim \lgint$, our algorithm fails to achieve separation (i.e., it achieves integration) with high probability.

\subsection{Proof Techniques} \label{subsec:techniques}

Because our distributed algorithm is based on a Markov chain, we can use standard tools such as detailed balance to understand its long-term behavior and prove its convergence to a unique probability distribution $\pi$ over particle system configurations.
This stationary distribution $\pi$ depends on the input parameters $\lambda$ and $\gamma$.
Our main contribution is analyzing $\pi$ for various ranges of $\lambda$ and $\gamma$, showing that a configuration drawn from distribution $\pi$ is either very likely (for large $\gamma$) or very unlikely (for $\gamma$ close to one) to be separated.

To show separation occurs when $\lambda$ and $\gamma$ are both large, we modify the proof technique of {\it bridging} introduced by Miracle, Pascoe, and Randall~\cite{Miracle2011}.
In this approach, we define a map from non-separated configurations to separated configurations and show this map has an exponential gain in weight, implying non-separated configurations are exponentially unlikely compared to separated ones.
Several new innovations are needed to adapt the ideas of~\cite{Miracle2011} to our more challenging setting due to the irregular shapes of particle system configurations, the interchangeability of color classes, the non-self-duality of the triangular lattice, and other irregularities related to boundary conditions.

To show separation does not occur when $\lambda$ is large and $\gamma$ is small (close to one), we use a probabilistic argument, a Chernoff-type bound, and a decomposition of configurations into different regions.
The key to this proof was finding a set of at most a polynomial number of events such that if separation occurs, so does one of these events.
We then show each event is exponentially unlikely, which --- by a union bound over all events in the set --- shows separation is also exponentially unlikely.

These arguments --- both for large and small $\gamma$ --- require that the particle system is compressed; i.e., that the system has perimeter $\Theta(\sqrt{n})$.
However, the arguments from~\cite{Cannon2016} showing compression occurs for homogeneous systems when $\lambda$ is large do not extend to the heterogeneous setting.
Instead, to show our separation algorithm indeed achieves compression for large enough $\gamma$, we turn to the {\it cluster expansion} from statistical physics.
The cluster expansion was first introduced in 1937 by Mayer~\cite{Mayer1937}, though a more modern treatment can be found in the textbook~\cite{Friedli2018} where it is used to derive several properties of statistical physics models including the Ising and hard-core models.
In the past year, the cluster expansion has received renewed attention in the computer science community due to the recent work of Helmuth, Perkins, and Regts which uses the cluster expansion to develop approximate counting and sampling algorithms for low-temperature statistical physics models on lattices including the Potts and hard-core models~\cite{Helmuth2019}.
Subsequent work has considered similar techniques on expander graphs~\cite{Jenssen2019} and random regular bipartite graphs~\cite{Liao2019}.
Inspired by the interpolation method of Barvinok~\cite{Barvinok2016article,Barvinok2016book}, these works give algorithms for estimating partition functions that explicitly calculate the first $\log n$ coefficients of the cluster expansion. We use the cluster expansion differently, to separate the volume and surface contributions to a partition function.

The cluster expansion is a power series representation of $\ln Z$ where $Z$ is a {\it polymer partition function}.
We relate each of our quantities of interest to a particular polymer partition function, and then use a version of the Koteck\'{y}-Preiss condition~\cite{Kotecky1986} to show that the power series in the cluster expansion is convergent for the ranges of parameters we are interested in.
We then use this convergent cluster expansion to split our polymer partition functions into a {\it volume term}, depending only on the size of the region of interest, and a {\it surface term}, depending only on its perimeter.
This separation into volume and surface terms turns out to be the key to our compression argument, both for large $\gamma$ and for $\gamma$ close to one.
While splitting partition functions into volume and surface terms is not a new idea in the statistical physics literature (for example, Section 5.7.1 of~\cite{Friedli2018} uses it to derive an explicit expression for the infinite volume pressure of the Ising model on $\mathbb{Z}^d$ with large magnetic field), we are the first to bring this approach into the computer science literature.  We are hopeful it will be useful beyond its specific applications in this paper.

\section{Background} \label{sec:background}

We begin by defining our amoebot model for programmable matter and stating a few key results.
We then extend the amoebot model to heterogeneous particle systems and formally define what it means for a system to be separated or integrated.
We conclude with the necessary terminology and results on Markov chains.

\subsection{The Amoebot Model} \label{subsec:model}

In the \emph{amoebot model}, introduced in~\cite{Derakhshandeh2014} and fully described in~\cite{Daymude2019}, programmable matter consists of individual, homogeneous computational elements called \emph{particles}.
In its geometric variant, particles are assumed to occupy nodes of the triangular lattice $\Gtri = (V,E)$ and can move along its edges (see Figure~\ref{fig:amoebotgrid}).
Each node in $V$ can be occupied by at most one particle at a time.

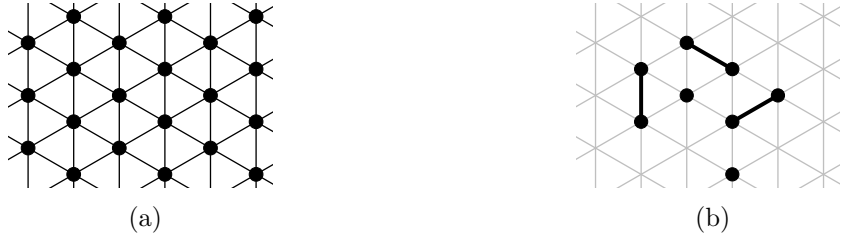
\begin{figure}
\centering
\begin{subfigure}{.45\textwidth}
	\centering
	\begin{tikzpicture}[scale=0.7]
    \clip (0.5,-0.25) rectangle (5.5,3.25);
    \foreach \i in {0,...,10} \draw[black,line width=.5pt] (\i*\sthree / 2,-5)--(\i*\sthree / 2,5);
    \foreach \i in {-10,...,10}
    {
    \draw[black,line width=.5pt] (0,\i)--(5*\sthree,\i + 5);
    \draw[black,line width=.5pt] (0,\i)--(5*\sthree,\i - 5);
    }
    \foreach \i in {0,2,...,10}
    \foreach \j in {-5,...,5}
    \draw[fill] (\i*\sthree / 2,\j) circle (0.13);
    \foreach \i in {1,3,...,10}
    \foreach \j in {-4.5,...,4.5}
    \draw[fill] (\i*\sthree / 2,\j) circle (0.13);
\end{tikzpicture}
	\caption{}
	\label{fig:amoebotgrid}
\end{subfigure}%
\begin{subfigure}{.45\textwidth}
	\centering
	\begin{tikzpicture}[scale=0.7]
    \clip (0.5,-0.25) rectangle (5.5,3.25);
    \foreach \i in {0,...,10}
    \draw[lightgray,line width=.5pt] (\i*\sthree / 2,-5)--(\i*\sthree / 2,5);
    \foreach \i in {-10,...,10}
    {
    \draw[lightgray,line width=.5pt] (0,\i)--(5*\sthree,\i + 5);
    \draw[lightgray,line width=.5pt] (0,\i)--(5*\sthree,\i - 5);
    }
    \draw[fill] (1*\sthree,1) circle (0.125);
    \draw[black,line width=1.5pt](1*\sthree,1)--(1*\sthree,2);
    \draw[fill] (1*\sthree,2) circle (0.125);
    \draw[fill] (2*\sthree,1) circle (0.125);
    \draw[black,line width=1.5pt](2*\sthree,1)--(2.5*\sthree,1.5);
    \draw[fill] (2.5*\sthree,1.5) circle (0.125);
    \draw[fill] (1.5*\sthree,2.5) circle (0.125);
    \draw[black,line width=1.5pt](1.5*\sthree,2.5)--(2*\sthree,2);
    \draw[fill] (2*\sthree,2) circle (0.125);
    \draw[fill] (1.5*\sthree,1.5) circle (0.125);
    \draw[fill] (2*\sthree,0) circle (0.125);
\end{tikzpicture}
	\caption{}
	\label{fig:amoebotparticles}
\end{subfigure}
\caption{(a) A section of the triangular lattice $\Gtri$. (b) Expanded and contracted particles (black dots) on $\Gtri$ (gray lattice). Particles with a black line between their nodes are expanded.}
\label{fig:model}
\end{figure}

Each particle occupies either a single node in $V$ (i.e., it is \emph{contracted}) or a pair of adjacent nodes in $V$ (i.e., it is \emph{expanded}), as in Figure~\ref{fig:amoebotparticles}.
Particles move via a series of \emph{expansions} and \emph{contractions}: a contracted particle can expand into an unoccupied adjacent node to become expanded, and completes its movement by contracting to once again occupy a single node.

Two particles occupying adjacent nodes are said to be \emph{neighbors}.
Each particle is \emph{anonymous}, lacking a unique identifier, but can locally identify each of its neighboring locations and can determine which of these are occupied by particles.
Each particle has a constant-size local memory that it can write to and its neighbors can read from for communication.
In particular, a particle stores whether it is contracted or expanded in its memory.
Particles do not have any access to global information such as a shared compass, coordinate system, or estimate of the size of the system.

The system progresses through \emph{atomic actions} according to the standard asynchronous model of computation from distributed computing (see, e.g.,~\cite{Lynch1996}).
A classical result under this model states that for any concurrent asynchronous execution of atomic actions, there exists a sequential ordering of actions producing the same end result, provided conflicts that arise in the concurrent execution are resolved.
In the amoebot model, an atomic action corresponds to the activation of a single particle.
Once activated, a particle can $(i)$ perform an arbitrary, bounded amount of computation involving information it reads from its local memory and its neighbors' memories, $(ii)$ write to its local memory, and $(iii)$ perform at most one expansion or contraction.
Conflicts involving simultaneous particle expansions into the same unoccupied node are assumed to be resolved arbitrarily such that at most one particle moves to some unoccupied node at any given time.
Thus, while in reality many particles may be active concurrently, it suffices when analyzing algorithms under the amoebot model to consider a sequence of activations where only one particle is active at a time.

\subsection{Terminology and Results for Homogeneous Particle Systems} \label{subsec:homosystems}

We now recall terminology and notation from our previous work on compression~\cite{Cannon2016} that we need for this work.
A particle system \emph{arrangement} is the set of vertices of the triangular lattice $\Gtri$ occupied by particles.
Two arrangements are equivalent if they are translations of each other; we define a particle system \emph{configuration} to be an equivalence class of arrangements.
An \emph{edge} of a configuration is an edge of $\Gtri$ where both endpoints are occupied by particles.
A configuration is \emph{connected} if for any two particles in the system, there is a path of such edges between them.
A configuration has a {\it hole} if there is a maximal, finite, connected component of unoccupied vertices in $\Gtri$.

As we justify with Lemma~\ref{lem:connectednoholes}, our analysis will focus on connected, hole-free configurations.
The \emph{boundary} of such a configuration $\sigma$ is the closed walk $\mathcal{P}$ on edges of $\sigma$ that encloses all particles of $\sigma$ and no unoccupied vertices of $\Gtri$.
The \emph{perimeter} $p(\sigma)$ of configuration $\sigma$ is the length of this walk, also denoted $|\P|$.
The following bounds the number of configurations with a given perimeter.

\begin{lem}[$\!\!\!$\cite{Cannon2016}, Lemma 4.3] \label{lem:nu}
For any $\nu > \cbound$, there is an integer $n_1(\nu)$ such that for all $n \geq n_1(\nu)$, the number of connected, hole-free particle system configurations with $n$ particles and perimeter $k$ is at most $\nu^k$.
\end{lem}

Let $p_{min}(n)$ be the minimum possible perimeter for a configuration of $n$ particles; it is easy to see that $p_{min}(n) = \Theta(\sqrt{n})$.
Given any $\alpha > 1$, a configuration of $n$ particles is said to be \emph{$\alpha$-compressed} if $p(\sigma) \leq \alpha \cdot p_{min}(n)$.
The following lemma establishes a concrete upper bound on $p_{min}(n)$.

\begin{lem} \label{lem:pmin}
For any $n \geq 1$, there is a connected, hole-free particle system configuration of $n$ particles with perimeter at most $2\sqrt{3}\sqrt{n}$.
That is, $p_{min}(n) \leq 2\sqrt{3}\sqrt{n}$.
\end{lem}
\begin{proof}
This lemma follows easily from noting that hexagonal configurations of $n$ particles have perimeter on the order of $2\sqrt{3}\sqrt{n}$; a proof can be found in Appendix~\ref{sec:pf-pmin}.
\end{proof}

\subsection{Heterogeneous Particle Systems} \label{subsec:heterosystems}

Generalizing previous work on the amoebot model in which all particles are homogeneous and indistinguishable, we assume that each particle $P$ has a fixed \emph{color} $c(P) \in \{c_1, \ldots, c_k\}$ that is visible to itself and its neighbors, where $k \ll n$ is a constant.
These colors can represent anything from different colonies of ants to varying equipment sets in a multi-robot team to different social demographic classes.
We extend the definition of \emph{configuration} given in Section~\ref{subsec:homosystems} to include both the vertices of $\Gtri$ occupied by particles as well as the colors of those particles.
An edge of configuration~$\sigma$ with endpoints occupied by particles $P$ and $Q$ is \emph{homogeneous} if $c(P) = c(Q)$ and \emph{heterogeneous}~otherwise.

We further extend the original model by allowing neighboring particles to exchange their positions in a \emph{swap move}.
Swap moves have no meaning in homogeneous systems as all particles are indistinguishable, but they grant heterogeneous systems flexibility in allowing particles trapped in the interior of the system to move freely.\footnote{In domains where physical swap moves are unrealistic, colors could be treated as in-memory attributes that could be exchanged by neighboring particles to simulate a swap move.}
Although these swap moves are not necessary for the correctness of our algorithm or our rigorous analysis, they enable faster convergence in practice.

In this paper, we study heterogeneous systems with $k = 2$ color classes.
As discussed in Section~\ref{sec:conclude}, our algorithm performs well in practice for larger values of $k$ and we expect our proof techniques would generalize without needing significant new ideas.
However, this generalization would be cumbersome; thus, for simplicity, we restrict our attention to systems with colors $\{c_1, c_2\}$.
For $2$-heterogeneous systems, we can formally define separation with respect to having large monochromatic regions.

\begin{defn} \label{defn:sepinformal}
For $\beta > 0$ and $\delta \in (0, 1/2)$, a $2$-heterogeneous particle system configuration $\sigma$ is said to be \emph{$(\beta, \delta)$-separated} if there is a subset of particles $R$ such that:
\begin{enumerate}
\item There are at most $\beta\sqrt{n}$ edges of $\sigma$ with exactly one endpoint in $R$;
\item The density of particles of color $c_1$ in $R$ is at least $1 - \delta$; and
\item The density of particles of color $c_1$ not in $R$ is at most $\delta$.
\end{enumerate}
\end{defn}

\noindent Unpacking this definition, $\beta$ controls how small a boundary there is between the monochromatic region $R$ and the rest of the system, with smaller $\beta$ requiring smaller boundaries.
The $\delta$ parameter expresses the tolerance for having particles of the wrong color within the monochromatic region $R$: small values of $\delta$ require stricter separation of the color classes, while larger values of $\delta$ allow for more integrated configurations.
Notably, $R$ does not need to be connected.

\subsection{Markov Chains} \label{subsec:markov}

A thorough treatment of Markov chains can be found in the standard textbook~\cite{Levin2009}. 
A {\it Markov chain} is a memoryless random process on a state space $\Omega$; for our purposes, $\Omega$ is finite and discrete.
We focus on discrete time Markov chains, where one transition occurs per {\it iteration} (or {\it step}) of the Markov chain.
Because of its stochasticity, we can completely describe a Markov chain by its transition matrix $M$, which is an $|\Omega| \times |\Omega|$ matrix indexed by the states of $\Omega$, where for $x,y \in \Omega$, $M(x,y)$ is the probability, if in state $x$, of transitioning to state $y$ in one step.
The $t$-step transition probability $M^t(x,y)$ is the probability of transitioning from $x$ to $y$ in exactly $t$ steps.

A Markov chain is \emph{irreducible} if for all $x,y \in \Omega$ there is a $t$ such that $M^t(x,y) > 0$.
A Markov chain is \emph{aperiodic} if for all $x \in \Omega$, $\gcd\{t : M^t(x,x) > 0\} = 1$.
A Markov chain is \emph{ergodic} if it is both irreducible and aperiodic.
A {\it stationary distribution} of a Markov chain is a probability distribution $\pi$ over $\Omega$ such that $\pi M = \pi$.
Any finite, ergodic Markov chain converges to a unique stationary distribution given by $\pi(y) = \lim_{t \to \infty} M^t(x,y)$ for any $x,y \in \Omega$; importantly, for such chains this distribution is independent of starting state $x$.
To verify $\pi'$ is the unique stationary distribution of a finite ergodic Markov chain, it suffices to check that $\pi'(x)M(x,y) = \pi'(y)M(y,x)$ for all $x,y \in \Omega$ (the \emph{detailed balance condition}; see, e.g.,~\cite{Feller1968}).

Given a state space $\Omega$, a set of allowable transitions between states, and a desired stationary distribution $\pi$ on $\Omega$, the Metropolis-Hastings algorithm~\cite{Hastings1970} gives a Markov chain on $\Omega$ that uses only allowable transitions and has stationary distribution $\pi$.
To accomplish separation, we consider the state space of particle configurations and allow transitions between states that correspond to one particle move; we pick a stationary distribution $\pi$ over configurations that favors well-separated configurations; and we calculate transition probabilities according to the Metropolis-Hasting algorithm (using a {\it Metropolis filter}).  Importantly, we choose $\pi$ so that these transition probabilities can be calculated by an individual particle using only information in its local neighborhood.


\section{The Separation Algorithm} \label{sec:alg}

We now present our stochastic, local, distributed algorithm for separation.
Our algorithm achieves separation by biasing particles towards moves that both gain them more neighbors overall and more like-colored neighbors.
We use two bias parameters to control this preference: $\lambda > 1$ corresponds to particles favoring having more neighbors, and $\gamma > 1$ corresponds to particles favoring having more neighbors of their own color.

In order to leverage powerful techniques from Markov chain analysis and statistical physics to prove the correctness of our algorithm, we design our algorithm to follow certain invariants.
First, assuming the initial particle system configuration is connected, our algorithm ensures it remains connected; this is necessary because particles have strictly local communication abilities so a disconnected particle is unable to communicate with or even find the rest of the particles.
Second, our algorithm eventually eliminates all holes in the configuration, and no new holes are ever formed.
Third, once all holes have been eliminated, all moves allowed by our algorithm are \emph{reversible}: if a particle moves from node $u$ to an adjacent node $v$ in one step, there is a nonzero probability that it moves back to $u$ in the next step.
Finally, the moves allowed by our algorithm suffice to transform any connected, hole-free configuration into any other connected, hole-free configuration.

Our algorithm uses two locally-checkable properties that ensure particles do not disconnect the system or form a hole when moving (our first two invariants).
We use the following notation.
For a location $\ell$ --- i.e., a node of the triangular lattice $\Gtri$ --- let $N_i(\ell)$ denote the particles of color $c_i$ occupying locations adjacent to $\ell$.
For neighboring locations $\ell$ and $\ell'$, let $N_i(\ell \cup \ell')$ denote the set $N_i(\ell) \cup N_i(\ell')$, excluding particles occupying $\ell$ and $\ell'$.
When ignoring color, let $N(\ell) = \bigcup_i N_i(\ell)$; define $N(\ell \cup \ell')$ analogously.
Let $\mathbb{S} = N(\ell) \cap N(\ell')$ denote the set of particles adjacent to both locations.
A particle can move from location $\ell$ to $\ell'$ if one of the following are satisfied:

\begin{property} \label{prop:1}
$|\mathbb{S}| \in \{1,2\}$ and every particle in $N(\ell \cup \ell')$ is connected to exactly one particle in $\mathbb{S}$ by a path through $N(\ell \cup \ell')$.
\end{property}

\begin{property} \label{prop:2}
$|\mathbb{S}| = 0$, and both $N(\ell) \setminus \{\ell'\}$ and $N(\ell') \setminus \{\ell\}$ are nonempty and connected.
\end{property}

\noindent Note that these properties do not need to be verified for swap moves, since swap moves do not change the set of occupied locations and thus cannot disconnect the system or create a hole.

We now define the Markov chain $\M$ for separation.
The state space $\Omega$ of $\M$ is the set of all connected heterogeneous particle system configurations of $n$ contracted particles, and Algorithm~\ref{alg:markovchainm} defines its transition probabilities.
We note that $\M$, a centralized Markov chain, can be directly translated to a fully distributed, local, asynchronous algorithm $\A$ that can be run by each particle independently and concurrently to achieve the same system behavior.
This translation is much the same as for previous algorithms developed using the stochastic approach to self-organizing particle systems~\cite{Cannon2016,AndresArroyo2018}; we refer the interested reader to those papers for details.
Importantly, this translation is only possible because all probability calculations and property checks in $\M$ use strictly local information available to the particles involved.
Simulations of $\M$ can be found in Section~\ref{subsec:simulations}.

\begin{algorithm}
\caption{Markov Chain $\M$ for Separation and Integration}
\label{alg:markovchainm}
\begin{algorithmic}[1]
\Statex \textbf{Beginning at any connected configuration $\sigma_0$ of $n$ particles, repeat:}
\State Choose a particle $P$ uniformly at random; let $c_i$ be its color and $\ell$ its location. \label{alg:step:mstart}
\State Choose a neighboring location $\ell'$ and $q \in (0,1)$ each uniformly at random.\label{alg:step:mdirprob}
\If {$\ell'$ is unoccupied}
    \State $P$ expands to occupy both $\ell$ and $\ell'$.
    \State Let $e = |N(\ell)|$ (resp., $e_i = |N_i(\ell)|$) be the number of neighbors (resp., of color $c_i$) $P$ had when contracted at location $\ell$, and define $e' = |N(\ell')|$ and $e_i' = |N_i(\ell')|$ analogously.
    \If {$(i)$ $e \neq 5$, $(ii)$ $\ell$ and $\ell'$ satisfy Property~\ref{prop:1} or~\ref{prop:2}, and $(iii)$ $q < \lambda^{e' - e} \cdot \gamma^{e_i' - e_i}$}
        \State $P$ contracts to $\ell'$.\label{alg:step:mcontract}
    \Else {} $P$ contracts back to $\ell$.
    \EndIf
\ElsIf {$\ell'$ is occupied by particle $Q$ of color $c_j$}
    \If {$q < \gamma^{|N_i(\ell') \setminus \{P\}| - |N_i(\ell)| + |N_j(\ell) \setminus \{Q\}| - |N_j(\ell')|}$} $P$ and $Q$ perform a swap move.\label{alg:step:mswap}
    \EndIf
\EndIf
\end{algorithmic}
\end{algorithm}

\subsection{\texorpdfstring{The Stationary Distribution of Markov Chain $\M$}{The Stationary Distribution of Markov Chain M}} \label{subsec:statdist}

In this section, we prove that Markov chain $\M$ maintains the four invariants described previously and then characterize its stationary distribution.

\begin{lem} \label{lem:connectednoholes}
If the particle system is initially connected, it remains connected throughout the execution of $\M$.
Moreover, $\M$ eventually eliminates any holes in the initial configuration, after which no holes are ever introduced again.
\end{lem}
\begin{proof}
This follows directly from the analogous results for the compression algorithm~\cite{Cannon2016}.
Although the separation and compression algorithms assign different probabilities to particle moves, the set of allowed movements is exactly the same, excluding swap moves.
But swap moves do not change the set of occupied nodes of $\Gtri$, so they cannot disconnect the system or introduce a hole.
\end{proof}

\begin{lem}\label{lem:reversible}
	Once all holes have been eliminated, every possible particle move is reversible; that is, if there is a positive probability of moving from configuration $\sigma$ to configuration $\tau$, then there is a positive probability of moving from $\tau$ to $\sigma$.
\end{lem}
\begin{proof}
Suppose, for example, that a particle $P$ moves from location $\ell$ to $\ell'$.
In the next time step, it is possible for $P$ to be chosen again (Step~\ref{alg:step:mstart}) and for $\ell$ to be chosen as the position to explore (Step~\ref{alg:step:mdirprob}).
Because Properties~\ref{prop:1} and~\ref{prop:2} are symmetric with respect to $\ell$ and $\ell'$, whichever was satisfied in the forward move will also be satisfied in this reverse move.
Finally, the probability checked in Condition $(iii)$ of Step~\ref{alg:step:mcontract} is always nonzero, so all together there is a nonzero probability that $P$ moves back to $\ell$ in this reverse move.
Swap moves can be shown to be reversible in a similar~way.
\end{proof}

\begin{lem} \label{lem:ergodic}
Markov chain $\M$ is ergodic on the state space of connected, hole-free configurations.
\end{lem}
\begin{proof}[Proof Sketch]
One can show that $\M$ is irreducible (i.e., the moves of $\M$ suffice to transform any configuration to any other configuration) similarly to the proof of the same fact for compression~\cite{Cannon2016}: it is first shown that any configuration can be reconfigured into a straight line; then, the line can be sorted by the color of the particles; finally, by reversibility (Lemma~\ref{lem:reversible}), the line can be reconfigured into any configuration.
Additionally, it is easy to see that $\M$ is aperiodic: at each iteration of $\M$, there is a nonzero probability that the configuration does not change.
Thus, because $\M$ is irreducible and aperiodic, we conclude it is ergodic.
\end{proof}

Because $\M$ is finite and ergodic, it converges to a unique stationary distribution $\pi$ which we now characterize.
For a configuration $\sigma$, let $h(\sigma)$ be the number of heterogeneous edges in $\sigma$.

\begin{lem} \label{lem:statdist} For $Z = \sum_{\sigma}(\lambda\gamma)^{-p(\sigma)}\cdot \gamma^{-h(\sigma)}$,
the stationary distribution of $\M$ is:
\[\pi(\sigma) = \left\{ \begin{array}{ll}
 (\lambda\gamma)^{-p(\sigma)}\cdot \gamma^{-h(\sigma)} / Z & \text{if $\sigma$ is connected and hole-free}; \\
0 & \text{otherwise.}
\end{array} \right.\]
\end{lem}
\begin{proof}[Proof Sketch]
By Lemma~\ref{lem:connectednoholes}, when $\M$ starts at a connected configuration it eventually reaches and remains in the set of configurations that are connected and hole-free.
Thus, disconnected configurations and configurations with holes have zero weight at stationarity.
In Appendix~\ref{app:detailedbalance}, we show using detailed balance that the unique stationary distribution of $\M$ can be written, for $\sigma$ connected and hole-free, as  $\pi(\sigma) = \lambda^{e(\sigma)}\cdot \gamma^{a(\sigma)} / Z_e$ where $e(\sigma)$ is the number of edges and $a(\sigma)$ is the number of homogeneous edges of $\sigma$ and  $Z_e = \sum_{\sigma}\lambda^{e(\sigma)}\cdot \gamma^{a(\sigma)}$.
This can be rewritten as in the lemma using two facts: $(i)$ since every edge is either homogeneous or heterogeneous, $e(\sigma) = a(\sigma) + h(\sigma)$; and $(ii)$ for any connected, hole-free configuration~$\sigma$, $e(\sigma) = 3n - p(\sigma) - 3$, a result shown in~\cite{Cannon2016}.
\end{proof}

\noindent The remainder of this paper will be spent analyzing this stationary distribution.

\section{Technical Overview} \label{sec:techoverview}

We consider large $\lambda$ and want to know for which values of $\gamma$ separation occurs.
Our proof techniques only apply to compressed configurations, so we must first show that Markov chain $\M$ achieves compression for the values of $\lambda$ and $\gamma$ we are interested in.
Previous proofs of compression in homogeneous particle systems break down for heterogeneous systems, so we introduce the {\it cluster expansion} to overcome this obstacle.
In Section~\ref{sec:cluster}, we give background on the cluster expansion, which comes from statistical physics and allows us to rewrite quantities of interest in more useful forms.
For the cluster expansion to be useful, we need to know the formal power series it involves is convergent.
We give a sufficient condition for the convergence of the cluster expansion (Theorem~\ref{thm:cluster-converge}) and derive the consequence of its convergence that we will use (Theorem~\ref{thm:bdry-volume}).
%
%


We first consider when $\lambda$ and $\gamma$ are large.
To show compression occurs in this case, in Section~\ref{sec:compression-large-gamma} we look at the {\it partition function} $Z_\P$ for different fixed boundaries $\P$, where $Z_\P$ is the sum over all configurations $\sigma$ with boundary $\P$ of their weights $(\lambda\gamma)^{-|\P|} \cdot \gamma^{-h(\sigma)}$.
We cannot analyze $Z_\P$ directly using a cluster expansion, so we instead relate $Z_\P$ to a \emph{polymer partition function} $\Xi_\P^\mL$ (defined in Equation~\ref{eqn:xi-loop}) which does have a cluster expansion.  
Using the sufficient condition given in Section~\ref{sec:cluster}, we show that the cluster expansion for $\Xi_\P^\mL$ --- which is a formal power series for $\ln \Xi_\P^\mL$ --- is convergent when $\gamma > \gsep$.
We then use this expression of $\ln \Xi_\P^\mL$ as a convergent power series to bound $\Xi_\P^{\mathcal{L}}$ in terms of a {\it volume term} depending only on the number of particles $n$ and a {\it surface term} depending only on $|\P|$, the length of boundary $\P$.
Lemma~\ref{lem:xi-bound} states that, provided $\gamma > \gsep$, for constants $c$ and $\psi$ we specify later,
\[e^{(3n-3)\Q - 3c|\P|} \leq \Xi_\P^\mL \leq e^{(3n-3)\Q + 3c|\P|}.\]
This lets us bound the ratio of $\Xi_\P^\mL$ and $\Xi_{\P'}^\mL$ for different boundaries $\P$ and $\P'$ of the same number of particles $n$ by a term that depends only on their lengths $|\P|$ and $|\P'|$, with no dependence on $n$.
We use this to apply a Peierls argument similar to the one used to show compression in~\cite{Cannon2016}.
This argument relates the total weight of undesirable configurations --- those with boundaries longer than $\alpha \cdot p_{min}$ for some constant $\alpha > 1$ --- to the weight of configurations with minimum perimeter, $p_{min}$. 
We show that the former has exponentially small weight in the stationary distribution $\pi$ provided $\lambda$, $\gamma$, and $\alpha$ satisfy the condition given in Theorem~\ref{thm:comp-sep}.
One aspect of this technical condition is summarized in Corollary~\ref{cor:comp-sep-lambdagamma}, which says that for any $\lambda > 1$ and $\gamma > \gsep$ such that $\lambda\gamma > 2(2+\sqrt{2})e^{\fourcsep} \sim \lgsep$, there exists a constant $\alpha$ such that a configuration drawn from the stationary distribution $\pi$ of $\M$ is $\alpha$-compressed with high probability.
(Recall, we say an event $A$ occurs with high probability, or w.h.p., if $\Pr[A] \geq 1 - c^{n^\delta}$, where $0 < c < 1$ and $\delta > 0$ are constants. Unless we explicitly state otherwise, it will always be the case that $\delta = 1/2$.)
Conversely, for any $\alpha > 1$, there exist $\lambda$ and $\gamma$ such that $\M$ with these parameter values achieves $\alpha$-compression at stationarity w.h.p.\ (Corollary~\ref{cor:comp-sep-alpha}).

We next show in Section~\ref{sec:separation}, again when $\lambda$ and $\gamma$ are large enough, that separation provably occurs.
By the work of Section~\ref{sec:compression-large-gamma}, it suffices to show this among compressed configurations.
We use a technique known as {\it bridging} that was developed to analyze molecular mixtures called \emph{colloids}~\cite{Miracle2011}.
Adapting the bridging approach to our setting required several new innovations to overcome obstacles such as the irregular shapes of particle system configurations, the non-self-duality of the triangular lattice, the interchangeability between color classes, and other technicalities related to interfaces between particles of different colors. 
The main result of this section is Theorem~\ref{thm:sep}, which states that if $\P$ is a boundary of an $\alpha$-compressed configuration and if $\alpha$, $\beta$, and $\delta$ satisfy a technical condition, configurations with boundary $\P$ and weights proportional to $\pi$ are $(\beta,\delta)$-separated w.h.p.
Combining this with the results of Section~\ref{sec:compression-large-gamma}, for any $\lambda > 1$ and $\gamma > 4^{5/4} \sim 5.66$ such that $\lambda\gamma > 2(2+\sqrt{2})e^{\fourcsep} \sim \lgsep$, there exists $\beta$ and $\delta$ such that $\M$ achieves $(\beta,\delta)$-separation at stationarity w.h.p.\ (Corollary~\ref{cor:comp+sep-lambdagamma}).
Furthermore, for any $\beta > 2\sqrt{3}$ and any $\delta < 1/2$, there exists $\lambda$ and $\gamma$ large enough so that $\M$ achieves $(\beta,\delta)$-separation at stationarity w.h.p.\ (Corollary~\ref{cor:comp+sep-betadelta}).

In Sections~\ref{sec:compression-small-gamma} and~\ref{sec:integration}, we show that there are some values of $\gamma$ close to one for which separation does not occur.
This counterintuitively includes some $\gamma > 1$, where particles have a preference for being next to particles of the same color.
As we did for large values of $\gamma$,
in Section~\ref{sec:compression-small-gamma} we first show that when $\lambda$ is large and $\gamma$ is close to one, compression provably  occurs.
The polymer partition function $\Xi_\P^\mL$ from above does not have a convergent cluster expansion when $\gamma$ is close to one, so we cannot use it to show compression.
Instead, we carefully relate $Z_\P$ to a different polymer partition function $\Xi_\P^{HT}$ by considering the {\it high temperature expansion}, which rewrites a sum over configurations with a fixed boundary as a sum over even edge sets within that boundary.
We show $\Xi_\P^{HT}$ has a convergent cluster expansion when $\gamma$ is close to one.
We then use the cluster expansion for this high temperature representation, much the same as in Section~\ref{sec:compression-large-gamma}, to show compression provably occurs.
Theorem~\ref{thm:comp-int} gives the condition $\lambda$, $\gamma$, and $\alpha$ must satisfy for $\alpha$-compression to occur.
This theorem implies that for any $\lambda > 1$ and $\gamma \in (\gintlower,\gintupper)$ such that $\lambda(\gamma+1) > 2(2+\sqrt{2}) e^{\threeaint}\sim \lgint$, there exists a constant $\alpha$ such that a configuration drawn from the stationary distribution $\pi$ of $\M$ is $\alpha$-compressed w.h.p.\ (Corollary~\ref{cor:comp-int-lambdagamma}).
Conversely, for any $\alpha > 1$ and any $\gamma \in (\gintlower, \gintupper)$, for large enough $\lambda$ algorithm $\M$ with parameters $\lambda$ and $\gamma$ achieves $\alpha$-compression at stationarity w.h.p.\ (Corollary~\ref{cor:comp-int-alpha}).

Once we have shown that compression occurs for large $\lambda$ and $\gamma$ near one,
in Section~\ref{sec:integration} we show that among these compressed configurations a large amount of separation between color classes is very unlikely.
We prove this with a probabilistic argument in which we find a set of polynomially many events such that if separation occurs, then at least one of these events occurs.
We then show that each event occurs with probability at most $\zeta^{n^{1/2-\varepsilon}}$ for some $\zeta < 1$ and arbitrarily small $\varepsilon > 0$, which via a union bound over the polynomial number of events implies separation is very unlikely.
The main result of this section is Theorem~\ref{thm:int}, which states that if $\P$ is a boundary of an $\alpha$-compressed configuration, then if $\gamma$ and $\delta$ satisfy some technical conditions, configurations with boundary $\P$ and weight proportional to $\pi$ are $(\beta, \delta)$-separated with probability at most $\zeta ^{n^{1/2-\varepsilon}}$ where $\varepsilon > 0$ can be arbitrarily close to zero and $\zeta < 1$.
Combining this with the results of Section~\ref{sec:compression-small-gamma}, we see that for $\lambda> 1$ and $\gamma \in (\gintl,\gintu)$
such that $\lambda(\gamma+1) > 2(2+\sqrt{2}) e^{\threea} \sim \lgint$, there are constants $\beta$ and $\delta$ such that the probability $\M$ with parameters $\lambda$ and $\gamma$ achieves $(\beta, \delta)$-separation at stationarity is at most $\zeta^{n^{1/2-\varepsilon}}$  where $\varepsilon > 0$ and $\zeta < 1$ (Corollary~\ref{cor:comp+int-lambdagamma}).
Conversely, for any $\beta > 0 $ and any $\delta < 1/4$, there exists $\lambda$ and $\gamma$ such that $\M$ with these parameters achieves $(\beta,\delta)$-separation at stationarity with probability at most $\zeta^{n^{1/2-\varepsilon}}$ for $\varepsilon > 0$ and $\zeta < 1$ (Corollary~\ref{cor:comp+int-betadelta}).


All of the results described here require $n$, the number of particles, to be sufficiently large. We do not expect any of the constant values in the bounds we present to be tight.

\subsection{Simulations} \label{subsec:simulations}

We supplement our rigorous results with simulations that show separation occurs for even better values of $\lambda$ and $\gamma$ than our proofs guarantee, indicating that our proven bounds are likely not tight.
We simulated $\M$ on heterogeneous particle systems with two colors, using 50 particles of each color.
Figure~\ref{fig:progress} shows the progression of $\M$ over time with bias parameters $\lambda = 4$ and $\gamma = 4$, the regime in which particles prefer to have more neighbors, especially those of their own color.
Although the simulation ran for nearly 70 million iterations, much of the progress towards a compressed and separated system occurs in the first million iterations.
Separation still occurs even when swap moves are disallowed, but takes much longer to achieve.

\begin{figure}
\centering
\includegraphics[scale=0.14]{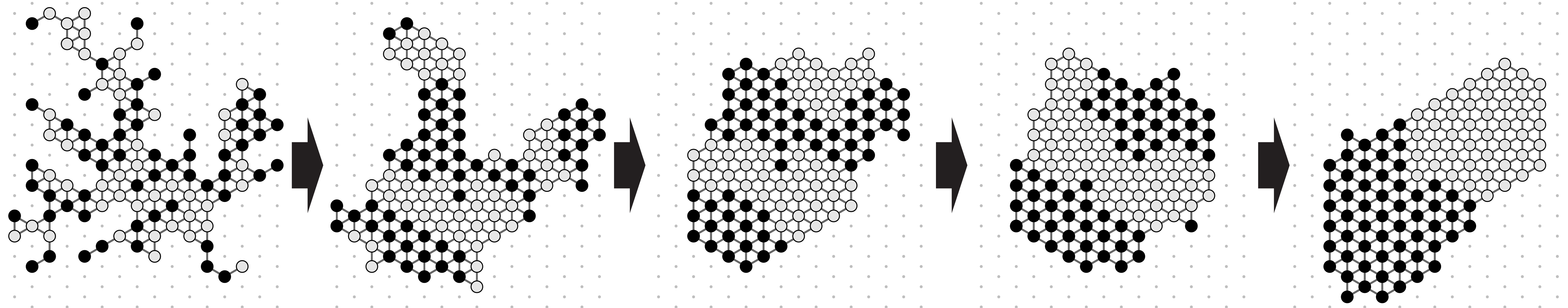}
\caption{A 2-heterogeneous particle system of 100 particles starting from an arbitrary initial configuration after (from left to right) 0; 50,000; 1,050,000; 17,050,000; and 68,250,000; iterations of $\M$ with $\lambda = 4$ and $\gamma = 4$.}
\label{fig:progress}
\end{figure}

Figure~\ref{fig:phases} compares the resulting system configurations after running $\M$ from the same initial configuration for the same number of iterations, varying only the values of $\lambda$ and $\gamma$.
We observe four distinct phases: compressed-separated, compressed-integrated, expanded-separated, and expanded-integrated.
We rigorously verify the compressed-separated behavior (i.e., when $\lambda$ and $\gamma$ are large) in Sections~\ref{sec:compression-large-gamma} and~\ref{sec:separation}, and do the same for the compressed-integrated behavior (i.e., when $\lambda$ is large and $\gamma$ is small) in Sections~\ref{sec:compression-small-gamma} and~\ref{sec:integration}.
We do not give proofs for expanded configurations; in fact, our current definition of separation may not accurately capture what occurs in expanded configurations.

\begin{figure}
\centering
\begin{tabular}{|c|c|c|}
	\hline
	& \small{$\gamma = 5.20$ (Separation)} & \small{$\gamma = 0.58$ (Integration)} \\
	\hline
	\includegraphics[scale = 0.13]{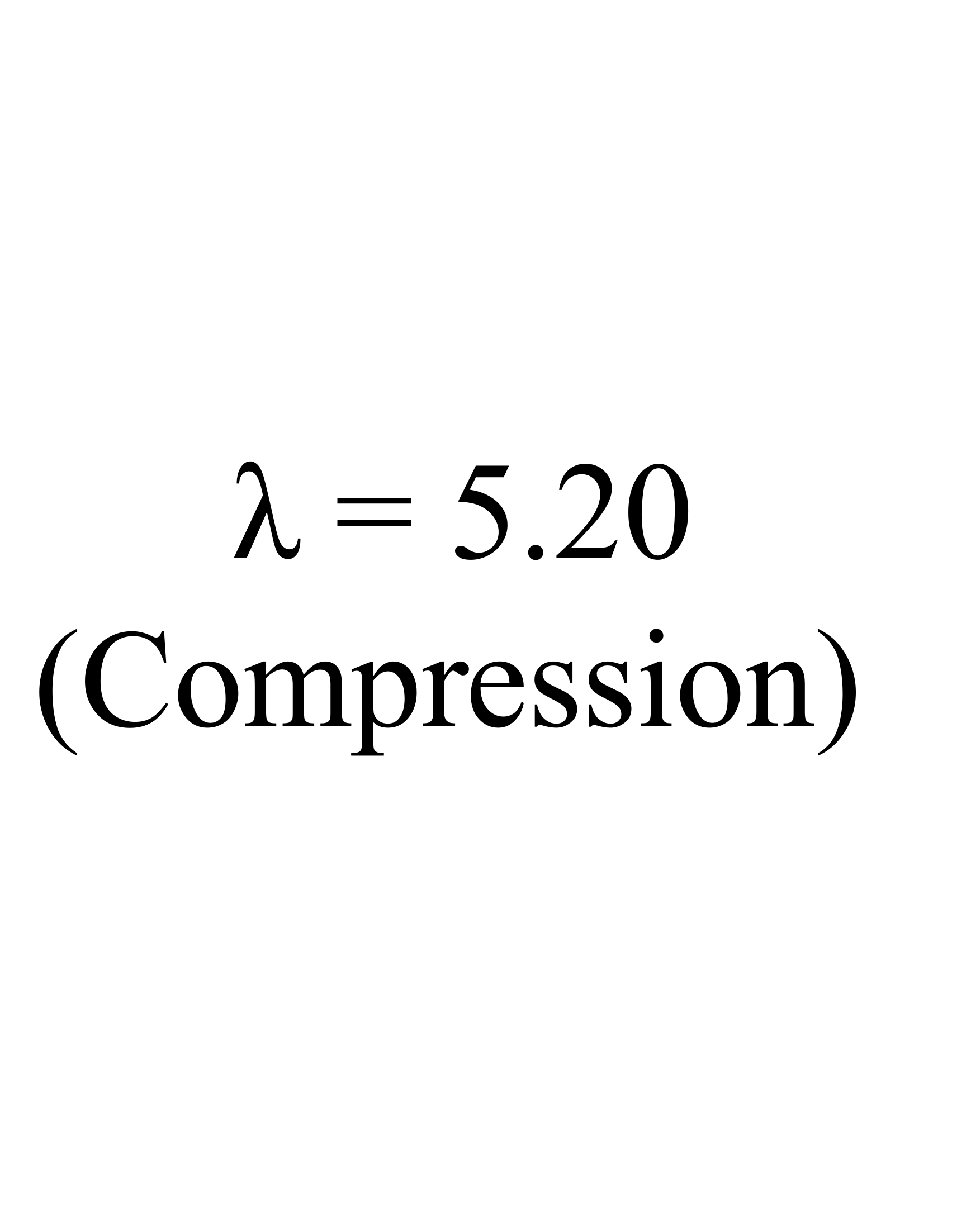} & \includegraphics[scale = 0.11]{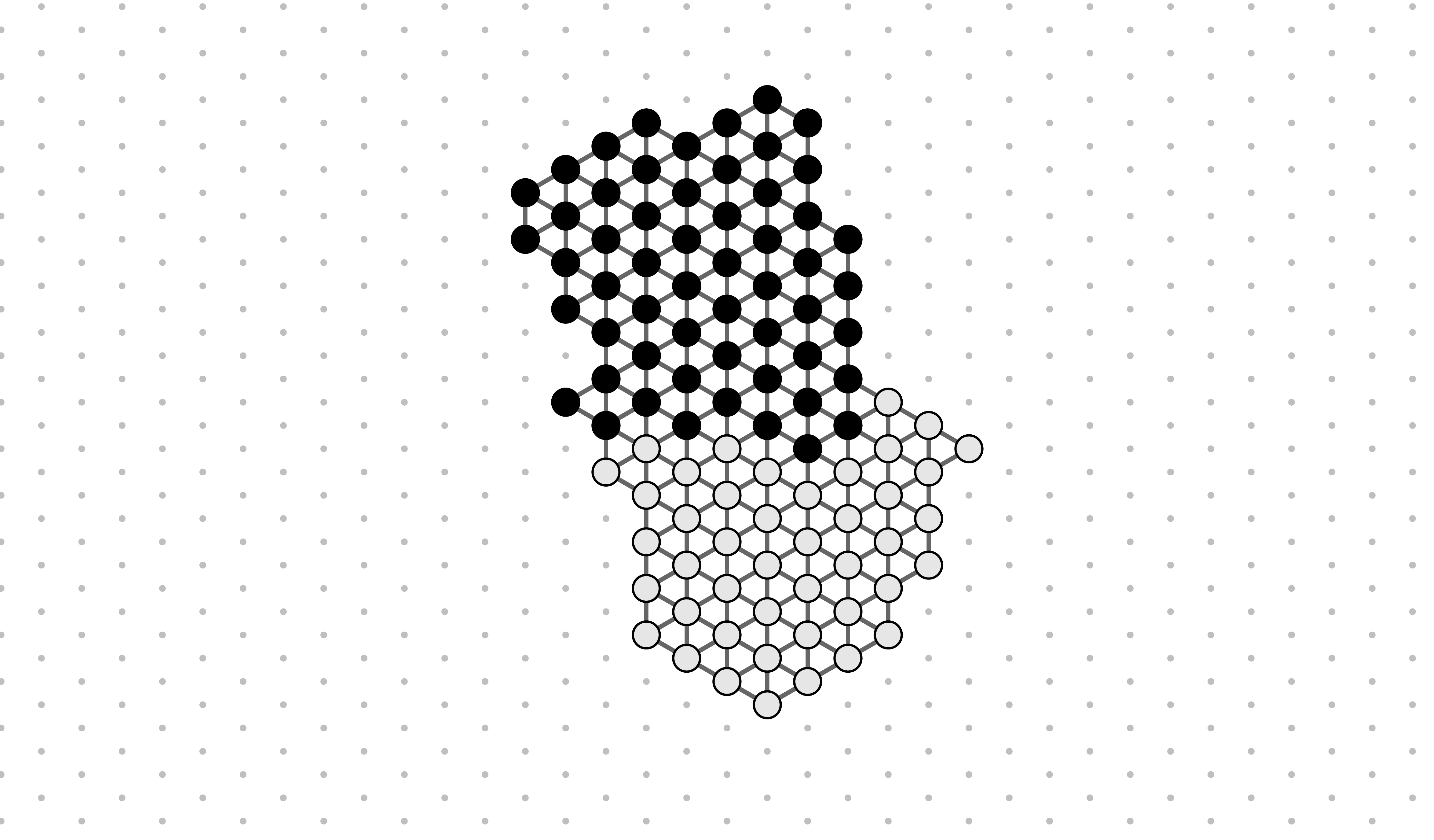} & \includegraphics[scale = 0.11]{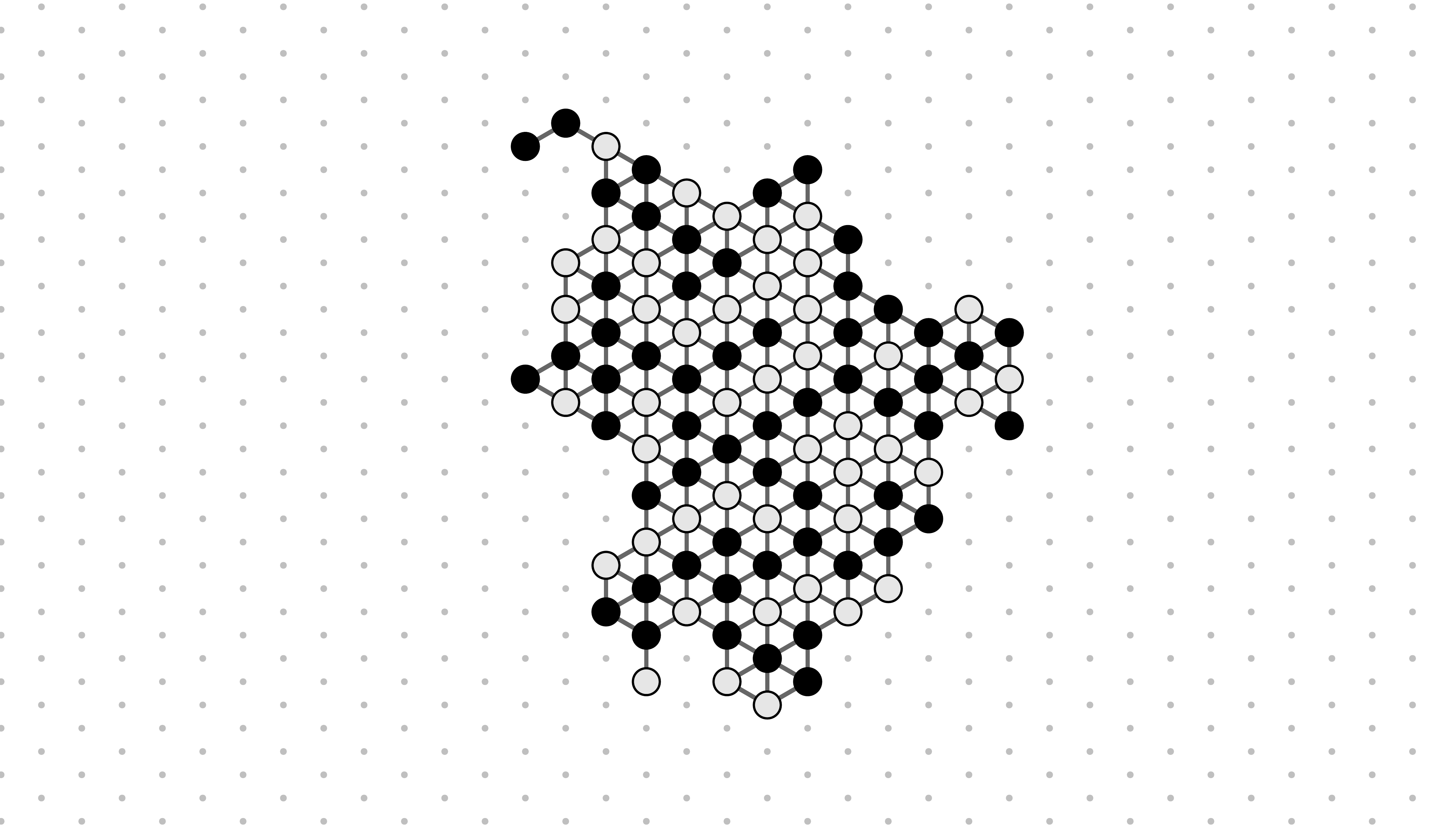} \\
	\hline
	\includegraphics[scale = 0.13]{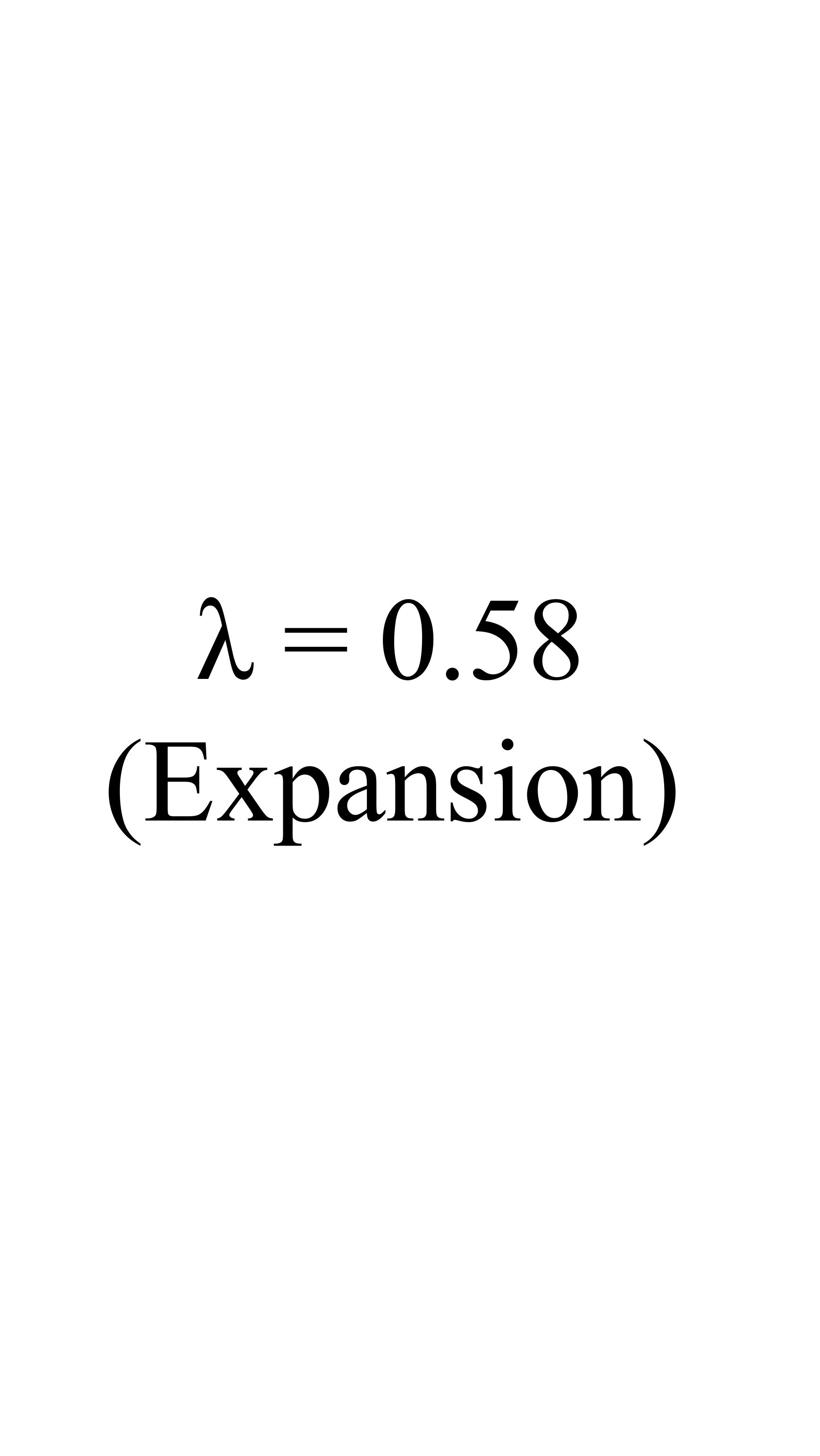} & \includegraphics[scale = 0.11]{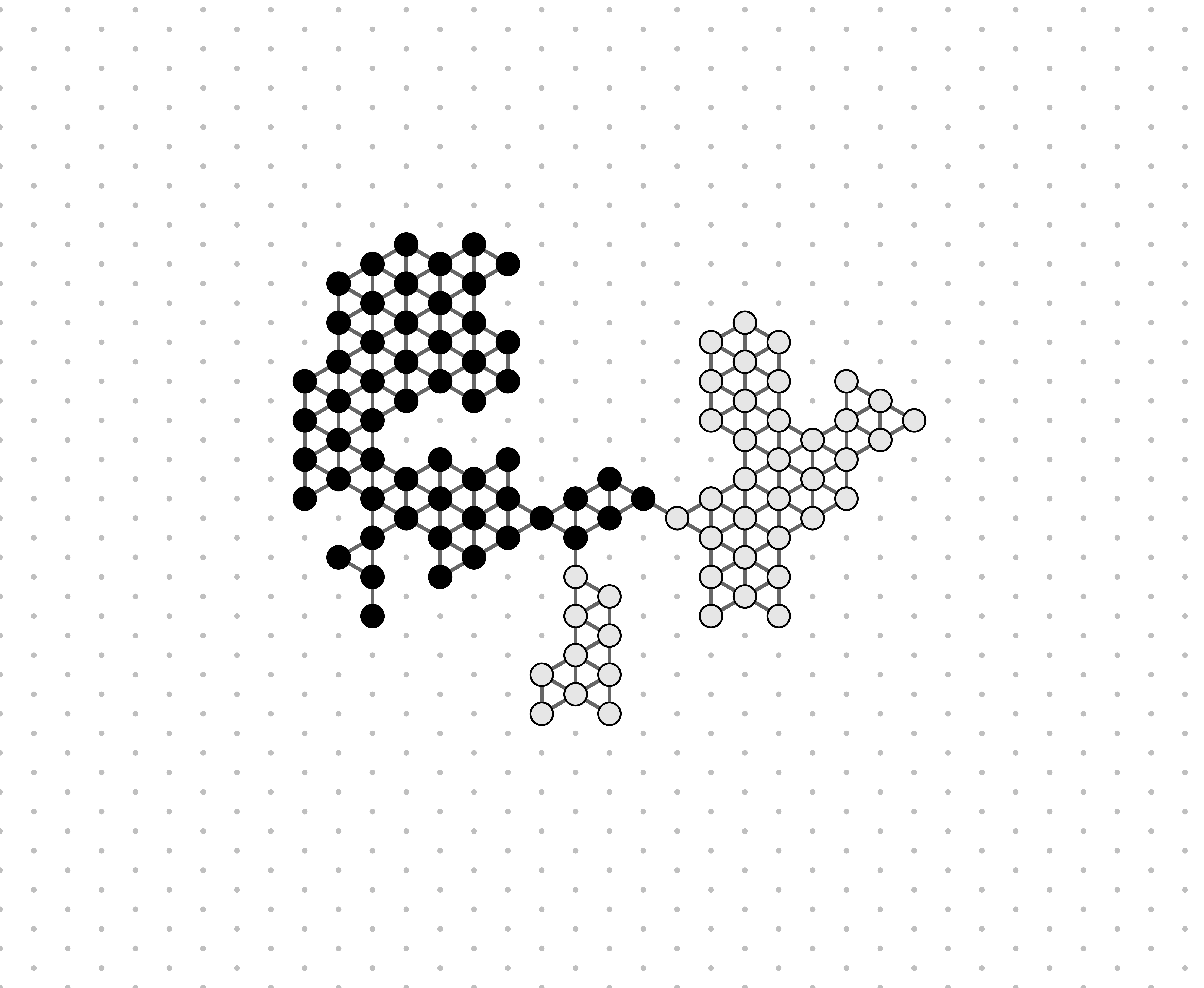} & \includegraphics[scale = 0.11]{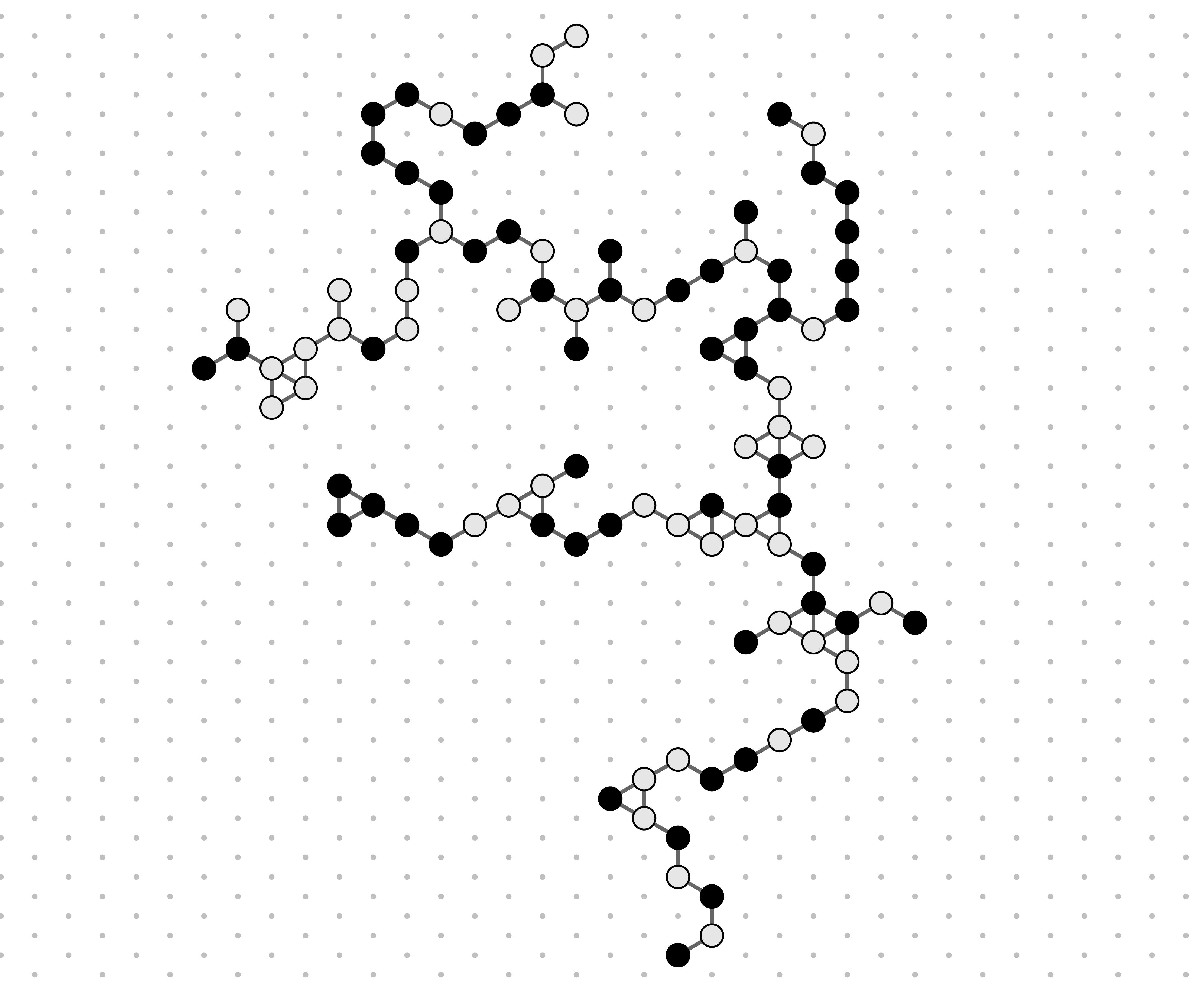} \\
	\hline
\end{tabular}
\caption{A 2-heterogeneous particle system of 100 particles starting in the leftmost configuration of Figure~\ref{fig:progress} after 50,000,000 iterations of $\M$ for various values of the parameters $\lambda$ and $\gamma$.}
\label{fig:phases}
\end{figure}

\section{The Cluster Expansion} \label{sec:cluster}

We highly recommend Chapter 5 of~\cite{Friedli2018} to learn more about the cluster expansion. Here we present only the definitions and results from this chapter that we will need. At a high level, the cluster expansion allows us to write a sum over collections of disjoint objects in terms of a sum over collections of overlapping objects.  This latter sum is often much easier to work with.

In a {\it polymer model}, we consider a finite set $\Gamma$, the elements of which are called {\it polymers}. We will consider polymers that are collections of edges of $\Gtri$ having certain properties; in Section~\ref{sec:compression-large-gamma}, our polymers are minimal cut sets which we call {\it loops}, and in Section~\ref{sec:compression-small-gamma} our polymers are connected edge sets with an even number of edges incident on each vertex. Formally, polymers only need to satisfy:
\begin{itemize}
	\item Each polymer $\xi \in \Gamma$ has a real {\it weight} $w(\xi)$ \footnote{In general $w(\xi)$ can be  complex, but for our purposes it will always be a (positive or negative) real number.}
	\item There is a notion of pairwise {\it compatibility} for polymers.
\end{itemize}
Polymers are typically compatible when they are well-separated in some sense. Our loop polymers will be compatible when they share no edges, and our even polymers will be compatible when they are not incident on any of the same vertices.
We say a collection of polymers $\Gamma' \subseteq \Gamma$ is {\it compatible} if all polymers in $\Gamma'$ are pairwise compatible.

The {\it polymer partition function} is defined as:
\[ \Xi =
\sum_{\substack{\Gamma' \subseteq \Gamma \\ compatible}} \prod_{\xi \in \Gamma'} w(\xi).
\]
Many partition functions of spin systems, such as the Ising model or the hard-core lattice gas model, can be written in this form as polymer partition functions.
Such an abstract sum can sometimes be hard to analyze, but the {\it cluster expansion} gives a way of rewriting this expression in terms of a sum over subsets $\Gamma' \subseteq \Gamma$ where many polymers are incompatible; because incompatible polymers `touch,' we can enumerate such collections more easily and thus such sums are often easier to work with

Formally, consider an ordered multiset $X = \{\xi_1,\xi_2,...,\xi_m\} \subseteq \Gamma$.
Let $H_{X}$ be the {\it incompatibility graph} on vertex set $\{ 1,2,...,m\}$ where $i \sim j$ whenever $\xi_i$ and $\xi_j$ are incompatible.
We say that the ordered multiset $X$ is a {\it cluster} if $H_{X}$ is connected.\footnote{Many sources define clusters to be unordered multisets, necessitating additional combinatorial terms in the cluster expansion; for simplicity, we assume clusters are ordered.}
Let $|X| = m$ denote the number of polymers in cluster $X$ (polymers appearing multiple times are counted with the appropriate multiplicity).

The {\it cluster expansion} is the formal power series for $\ln \Xi$ given in Equation~\ref{eqn:cluster-X}. Often this power series does not converge, but the {\it Kotecky-Preiss condition} guarantees convergence and is often easy to verify~\cite{Kotecky1986}. The following theorem states the Kotecky-Preiss condition (Equation~\ref{eqn:suff-cond}) and the cluster expansion of~$\Xi$.

\begin{thm}[$\!\!$\cite{Friedli2018}, Chapter 5]\label{thm:cluster-converge}
	Let $\Gamma$ be a finite set of polymers $\xi$ with real weights $w(\xi)$ and a notion of pairwise compatibility. If there exists a function $a : \Gamma \rightarrow \mathbb{R}_{>0}$ such that for all $\xi^* \in \Gamma$,
	\begin{align}\label{eqn:suff-cond}
	\sum_{\substack{\xi \in \Gamma: \\ \xi, \xi^* \ incompatible}} |w(\xi)| e^{a(\xi)}  \leq a(\xi^*),
	\end{align}
	then the polymer partition function  $\Xi$ satisfies
	\begin{align}\label{eqn:cluster-X}
	\ln \Xi = \sum_{\substack{ X  : \,{\text cluster}}} \frac{1}{|X|!} \left( \sum_{\substack{G \subseteq H_X:\\ connected, \\ spanning}} (-1)^{|E(G)|} \right) \left( \prod_{\xi \in X} w(\xi) \right),
	\end{align}
	where $G \subseteq H_X$ means $G$ is a subgraph of $H_X$.
\end{thm}

\noindent The cluster expansion is derived and this theorem is proved in Chapter 5 of~\cite{Friedli2018}, for a slightly different (but equivalent) definition of a cluster.
 For simplicity, for any multiset $X$ we define
  \begin{align}\label{eqn:Psi}
  \Psi(X) =  \frac{1}{|X|!} \left( \sum_{\substack{G \subseteq H_X:\\ connected, \\ spanning}} (-1)^{|E(G)|} \right) \left( \prod_{\xi \in X} w(\xi) \right).
  \end{align}
We note that $\Psi(X) = 0$ if $X$ is not a cluster, as in this case $H_X$ has no connected spanning subgraphs.

As previously discussed, we will apply the cluster expansion twice, with two different notions of polymers and compatibility.
In both cases, our polymers will be connected edge sets $\xi \subseteq E(\Gtri)$, and we use that to state a general result here.
 Let $\Gamma$ be an infinite set of such polymers that is invariant under translation and rotation of polymers. Two polymers in $\Gamma $ will be compatible if they are well-separated in the model-dependent sense described above. Polymers are incompatible when they are `too close;'
for a polymer $\xi \in \Gamma$, let $[ \xi ] \subseteq E(\Gtri)$  be the the minimal edge set  such that if $\xi'$ is not compatible with $\xi$, then $\xi'$ must contain an edge of $[\xi]$. We use brackets, consistent with the notaiton of~\cite{Friedli2018} because this is a type of {\it closure} of a polymer.
For our loop polymers, which are compatible if they share no edges, $[\xi] = \xi$.  For our even polymers, which are compatible if they are not incident on any of the same vertices, $[\xi]$ is all edges that share an endpoint with an edges of $\xi$.
 We denote the size of this edge set as $|[\xi]|$.

We will be interested in some finite region $\Lambda \subseteq E(\Gtri)$, and we say $\Gamma_\Lambda \subseteq \Gamma $ is all polymers of $\Gamma$ whose edges are contained in $\Lambda$. Let $\partial \Lambda$ be an edge set such that a cluster containing an edge in $\Lambda$ and an edge not in $\Lambda$ must contain an edge of $\partial \Lambda$.  We will consider loop polymers with edges from $E^{int}_\P$, the set of edges with at least one endpoint strictly inside boundary $\P$, so in this case we will use $\Lambda = E^{int}_\P$ and $\partial\Lambda$ the edges in $\P$. For even polymers, we use $\Lambda = E_\P$, all edges on or inside $\P$, and $\partial \Lambda$ is all edges with one endpoint on boundary $\P$ and the other endpoint outside $\P$.

The following states the key fact about the cluster expansion that we will need. Namely, when a certain mild condition is satisfied, we can use the cluster expansion to give upper and lower bounds on the polymer partition function for $\Lambda$ in terms of a volume term, depending only on $|\Lambda|$, and a surface term, depending only on $|\partial \Lambda|$.

\begin{thm}\label{thm:bdry-volume}
	Let $\Gamma$ be an infinite set of polymers $\xi \subseteq E(\Gtri)$ that is closed under translation and rotation, and let $\Lambda \subseteq E(\Gtri)$ be finite.
	 If there is a constant $c$ such that for any edge $e \in E(\Gtri)$,
	 \begin{align}\label{eqn:suff-cond-edge}
	 \sum_{\substack{\xi \in \Gamma : \\ e \in \xi}} |w(\xi)| e^{c |[\xi]|} \leq c,
	 \end{align}
	 then for any $\Lambda$ the partition function
	\[ \Xi_\Lambda := \sum_{\substack{\Gamma'\subseteq \Gamma_\Lambda \\ compatible}} \prod_{\xi \in \Gamma'} w(\xi) \]
	satisfies 	\[ e^{\psi|\Lambda| - c|\partial \Lambda|} \leq  \Xi_\Lambda \leq e^{\psi|\Lambda| + c|\partial \Lambda|} ,
	\]
	for some constant $\psi\in [-c,c]$ that is independent of $\Lambda$.

	\end{thm}
\begin{proof}
	We follow the same outline as the proof of the same fact for the Ising model in Section 5.7.1 of \cite{Friedli2018}; some details have been omitted here but are carefully considered in Appendix~\ref{app:bdry-volume}.

	Let $\mathcal{X}$ be all clusters comprised of polymers from $\Gamma$, and let $\mathcal{X}_\Lambda$ be all clusters of polymers in $\Gamma_\Lambda$.
	Note that Equation~\ref{eqn:suff-cond-edge} implies the hypothesis of Theorem~\ref{thm:cluster-converge} (Equation~\ref{eqn:suff-cond}) is satisfied, with function $a: \Gamma \rightarrow \mathbb{R}$ given by $a(\xi) = c | [\xi]|$:
	\[
	\sum_{\substack{\xi \in \Gamma: \\ \xi, \xi^* \ incompatible}} |w(\xi)| e^{a(\xi)}
	\leq \sum_{e \in [\xi^*]} \sum_{\substack{\xi \in \Gamma :\\ e \in \xi}} |w(\xi)| e^{c|[\xi]|} \leq c| [\xi^*]|.
	\]
	Because this hypothesis is satisfied for all $\xi^* \in \Gamma$, it certainly holds when we restrict our attention to polymers in $\Gamma_\Lambda$.
	 By Theorem~\ref{thm:cluster-converge}, because $\Gamma_\Lambda$ is a finite set, this means the cluster expansion for $\Xi_\Lambda$ converges:
	\begin{align*}
	\ln \Xi_\Lambda = \sum_{X\in \mathcal{X}_\Lambda} \Psi(X)
	\end{align*}
 Let $\overline{X} = \cup_{\xi \in X} \xi$ be the {\it support} of cluster $X$ and $|\overline{X}|$ the size of this support.
Using Equation~\ref{eqn:suff-cond-edge} and standard techniques (see~\cite{Friedli2018}, the proof of Theorem 5.4 and Equation (5.29); rewritten for our setting and proved as Lemma~\ref{lem:less-than-c}), the translation and rotation invariance of $\Gamma$ imply that for any edge $e \in E(\Gtri)$,
	\begin{align}\label{eqn:Psi<1} \sum_{\substack{X \in \mathcal{X}: \\ e \in \overline{X}}} |\Psi(X)| \leq c.  \end{align}
	The proof of this fact is the reason we need a slightly stronger hypothesis (Equation~\ref{eqn:suff-cond-edge}) than is needed to guarantee the cluster expansion converges (Equation~\ref{eqn:suff-cond}).

	For any cluster $X\in \mathcal{X}_\Lambda$, it trivially holds that $1 = (\sum_{e \in \Lambda} \mathbf{1}_{e \in \overline{X}}) / \overline{X}$. We can use this fact to rewrite the cluster expansion for $\Xi_\Lambda$:
	\begin{align}
	\ln \Xi_\Lambda &= \sum_{X \in \mathcal{X}_\Lambda} \Psi(X)
	 = \sum_{\substack{X \in \mathcal{X}: \\ \overline{X} \subseteq \Lambda}} \Psi(X)
	 = \sum_{e \in \Lambda} \sum_{\substack{ X \in \mathcal{X}: \\ e \in \overline{X}, \\  \overline{X} \subseteq \Lambda}} \frac{1}{|\overline{X}|} \Psi(X)
	\nonumber\\& = \sum_{e \in \Lambda} \left( \sum_{\substack{X \in \mathcal{X}: \\ e \in \overline{X}}} \frac{1}{|\overline{X}|} \Psi(X) -  \sum_{\substack{X \in \mathcal{X}: \\ e \in \overline{X}, \\ \overline{X} \not\subseteq \Lambda}} \frac{1}{|\overline{X}|} \Psi(X) \right)
	\nonumber\\&= \left(\sum_{e \in\Lambda}  \sum_{\substack{X \in \mathcal{X} :\\ e \in \overline{X}}} \frac{1}{|\overline{X}|} \Psi(X) \right) - \left( \sum_{e \in\Lambda} \sum_{\substack{X \in \mathcal{X}: \\ e \in \overline{X}, \\ \overline{X} \not\subseteq \Lambda}}
	\frac{1}{|\overline{X}|} \Psi(X) \right). \label{eqn:differenceofseries}
	\end{align}
	The two infinite sums in parentheses above are absolutely convergent by Equation~\ref{eqn:Psi<1}, so this difference is well-defined.

	To analyze the first term of Equation~\ref{eqn:differenceofseries}, we note that by the translation and rotation invariance of $\Gamma$, the sum
	\begin{align*}
	\psi:= \sum_{\substack{X \in \mathcal{X}: \\ e \in \overline{X}}} \frac{1}{|\overline{X}|} \Psi(X)
	\end{align*}
	is independent of $e$ and of $\Lambda$ and only depends on our particular polymer model; this is the value $\psi$ that appears in the statement of the theorem, and by Equation~\ref{eqn:Psi<1},  $|\psi| \leq c$. We conclude the first term of Equation~\ref{eqn:differenceofseries} is $\psi|\Lambda|$.

	To analyze the second term of Equation~\ref{eqn:differenceofseries}, recall if cluster $X$ satisfies both $e \in \overline{X}$ for some $e \in \Lambda$ and $\overline{X} \not\subseteq \Lambda$, then $\overline{X}$ must contain some edge $f \in \partial \Lambda$. We rewrite the absolute value of this second sum as
	\begin{align*}
\left| 	\sum_{e \in\Lambda} \sum_{\substack{X \in \mathcal{X}: \\ e \in \overline{X}, \\ \overline{X} \not\subseteq \Lambda}} \frac{1}{|\overline{X}|}\Psi(X) \right|
&\leq 	\sum_{e \in\Lambda} \sum_{\substack{X \in \mathcal{X}: \\ e \in \overline{X}, \\ \overline{X} \not\subseteq \Lambda}}
	\frac{1}{|\overline{X}|} \left|\Psi(X)\right|
	\\& \leq \sum_{f \in \partial\Lambda} \sum_{\substack{X \in \mathcal{X}:\\ f \in \overline{X}}} |\overline{X} \cap  \Lambda| \frac{1}{|\overline{X}|} \left|\Psi(X)\right|
	\\& \leq  \sum_{f \in \partial\Lambda} \sum_{\substack{X \in \mathcal{X}:\\ f \in \overline{X}}} \left|\Psi(X)\right| \leq  c \left|\partial \Lambda\right|.
	\end{align*}
	The last inequality above follows from Equation~\ref{eqn:Psi<1} and the translation and rotation invariance of~$\Lambda$.

We conclude that Equation~\ref{eqn:differenceofseries} implies
\begin{align*}
\psi |\Lambda| - c |\partial\Lambda | \leq \ln \Xi_\Lambda \leq \psi |\Lambda| + c |\partial\Lambda |.
\end{align*}
Exponentiation proves the theorem. 
\end{proof}




\section{\texorpdfstring{Proof of Compression for large $\gamma$, $\lambda$}{Proof of Compression for large gamma, lambda}}
\label{sec:compression-large-gamma}

In this section we show that when $\lambda$ and $\gamma$ are large enough, compression occurs with high probability in $\M$'s stationary distribution.  We begin by looking at a fixed boundary $\P$ and examining the total weight of all configurations with this boundary.

Let $\mathcal{P}$ be the boundary of some connected hole-free configuration $\sigma$ with $n$ total particles. Recall $|\mathcal{P}|$ is the total length of walk $\P$, which is equal to the perimeter $p(\sigma)$ of $\sigma$.
Let $\Os \subseteq \Omega$ be the set of all valid particle configurations in $\Omega$ with no holes, boundary $\mathcal{P}$, and the correct number of particles of each color. All configurations in $\Omega_\P$ have the same locations occupied by particles but the colors of these particles vary. We denote by $E_\mathcal{P}$ all edges of $\Gtri$ that have both endpoints on or inside $\P$ (all edges with both endpoints occupied by particles in any $\sigma \in \Omega_\P$). Let $E_\P^{int} = E_\P \setminus \P$ be all edges of $\Gtri$ with at least one endpoint inside $\P$.

 We are interested in bounding
\[ w(\Os)  := \sum_{\sigma \in \Os} (\lambda\gamma)^{-p(\sigma)}\cdot \gamma^{-h(\sigma)} = (\lambda\gamma)^{-|\P|}\sum_{\sigma \in \Os} \gamma^{-h(\sigma)} \]
We will use the cluster expansion examine this last sum.

\subsection{Obtaining Constant Boundary Conditions}

The cluster expansion is most easily applied in settings with constant boundary conditions; in our case, that means all particles on the boundary $\P$ should have the same color. We can relate the weight of configurations with arbitrary colors assigned to particles of $\P$ to configurations with all partices on $\P$ the same color as follows.

We define a map from $\Os$ to a set $\Osc$, consisting of all particle configurations with boundary $\P$ where all particles on this boundary have color $c_1$.  Note that $\Osc$ contains configurations with various numbers of particles of each color, so $\Osc \not\subseteq \Omega_\mathcal{P}$.

We say a {\it face} of a particle configuration $\sigma$ is a maximal simply connected subset $F$ of particles where all particles in $F$ adjacent to a location not in $F$ have the same color, which we call the color of $F$.  An {\it outer face} is a face containing at least one particle of $\P$. See Figure~\ref{fig:facesa}, where there are six faces, three of which are outer.
Note that any edge between two outer faces must be heterogeneous. The heterogeneous edges between two outer faces are linked by the triangular faces of $\Gtri$.  Because any triangular faces of $\Gtri$ can have at most two of its three edges be heterogeneous because of parity, we can trace out an interface of heterogeneous edges between two outer faces by, for a given heterogeneous edge and an incident triangular face, choosing exactly one of the other two edges incident on that fact to be heterogeneous. We will use the fact that heterogeneous interfaces between outer faces always continue in one of two way in the following proof.

\begin{figure}
			\centering
	\begin{subfigure}{0.32\textwidth}
		\centering

		\includegraphics[scale = 0.4, page = 1]{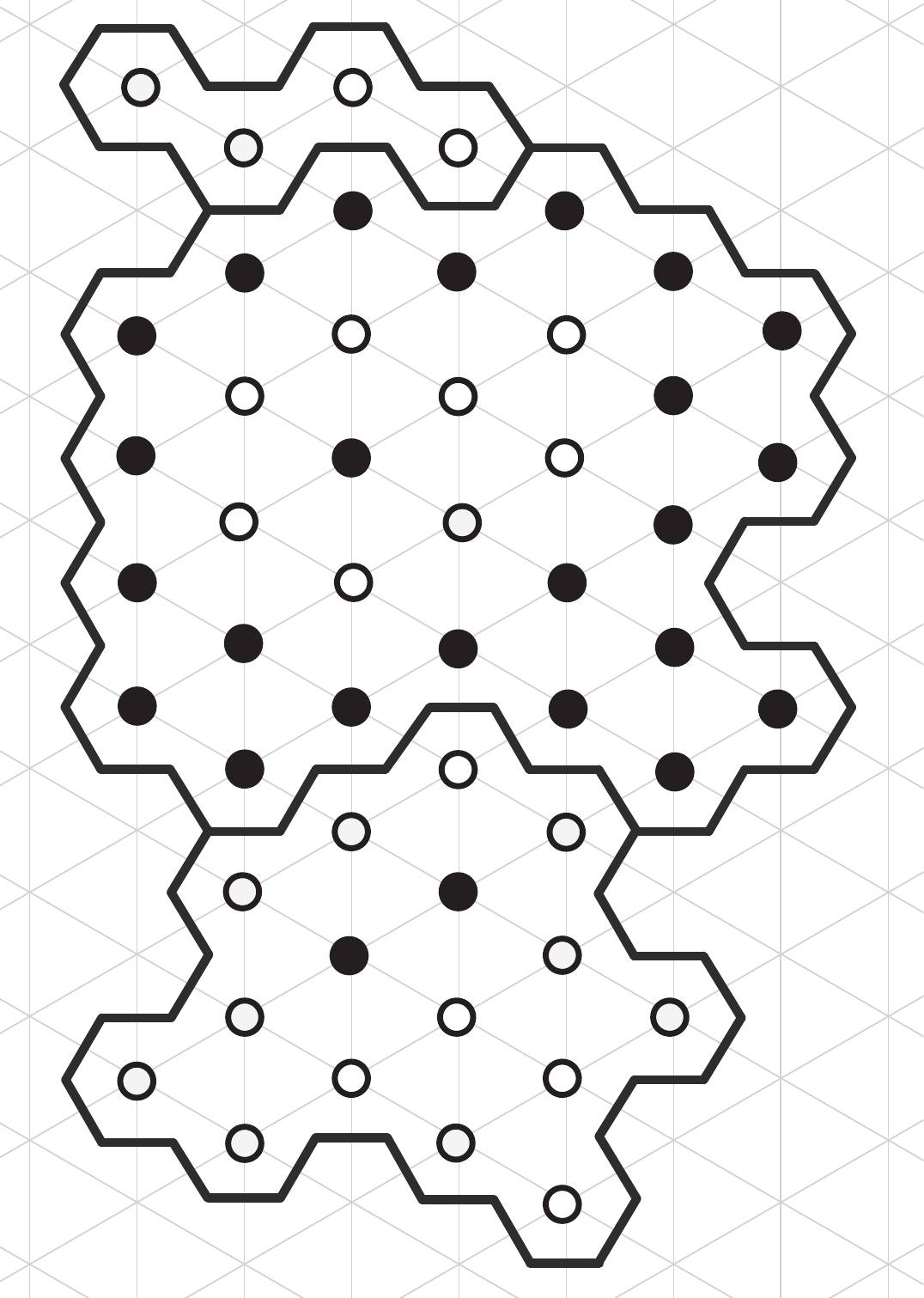}
		\caption{\centering}
		\label{fig:facesa}
	\end{subfigure}%
	\begin{subfigure}{0.32\textwidth}
		\centering
		\includegraphics[scale = 0.4, page = 2]{faces.pdf}
		\caption{\centering}
		\label{fig:facesb}
	\end{subfigure}%
	\begin{subfigure}{0.32\textwidth}
	\centering
		\includegraphics[scale = 0.4, page = 3]{faces.pdf}
	\caption{\centering}
	\label{fig:loops}
\end{subfigure}%
	\caption{(a) A particle configuration $\sigma$ with six faces, three of which are outer faces; the three outer faces are outlined in black. (b) Configuration $f(\sigma)$ obtained in the proof of Lemma~\ref{lem:elim-crossing}, where face $F_2$ has been complemented. (c) The three loops present in $f(\sigma)$ are shown with thick blacks edges, thin black edges, and dashed edges, respectively. }\label{fig:faces}.
	\end{figure}

\begin{lem}\label{lem:elim-crossing} 
	 When $\gamma > 2$,
\[w(\Os) \leq w(\Osc) 2^{|\P|} \frac{2}{\gamma-2}	\]
\end{lem}
\begin{proof}
	We define a map $f: \Os \rightarrow \Osc$ as follows.  For $\sigma \in \Os$, $f(\sigma)$ is obtained by complementing all colors within all outer faces of $\sigma$ that are of color $c_2$; see Figure~\ref{fig:faces}, where color $c_1$ is white and color $c_2$ is black. This guarantees that $f(\sigma) \in \Osc$. We use the notation $w(\sigma) = (\lambda \gamma)^{-p(\sigma)} \gamma^{-h(\sigma)}$ and note that if this process eliminated $x(\sigma)$ heterogeneous edges between outer faces, then $w(\sigma) = \gamma^{-x(\sigma)} w(f(\sigma))$. We also bound, for each $f(\sigma) \in \Osc$, the number of preimages $\sigma \in \Os$ with $x(\sigma) = x$.  There are at most $2^{|\P|}$ ways the particles on $|\P|$ can be colored in any such preimage, and these colors specify starting locations for heterogeneous interfaces between outer faces. From a given starting edge, there are at most $2^t$ choices for what a heterogeneous interface between outer faces of length $t$ looks like. Since a preimage of any particular $f(\sigma) \in \Osc$ can be completely specified by the colors of its boundary particles and the structure of its heterogeneous interfaces between outer faces, $f(\sigma)$ has at most $2^{|\mathcal{P}|} \cdot 2^x$ preimages with $x$ heterogeneous edges between outer faces.

	We conclude
	\begin{align*} w(\Os) = \sum_{\sigma \in \Os} w(\sigma)  &\leq \sum_{\sigma \in \Os}  \gamma^{-x(\sigma)} w(f(\sigma))
	\\&\leq \sum_{\tau \in \Osc} w(\tau) \sum_{x = 1}^{|E_\P|} \gamma^{-x} 2^x 2^{|\P|}
	\leq w(\Osc) 2^{|\P|} \frac{2}{\gamma - 2}.
	\end{align*}
	The last inequality holds because $\gamma > 2$ so the sum over $x$ is a geometric series with ratio less than one.
\end{proof}
 We now focus on bounding $w(\Osc)$, which will be easier to analyze with a cluster expansion because it only includes particle configurations with a monochromatic boundary.

\subsection{A Polymer Model and Convergence of the Cluster Expansion}

For a configuration $\sigma \in \Osc$, let $w_\P(\sigma) = \gamma^{-h(\sigma)}$ be the relative weight of $\sigma$ once boundary $\P$ is fixed (unlike in other sections, where $w(\sigma)$ also includes a term for the perimeter of $\sigma$). That is,
\[
w_\P(\sigma) = \gamma^{-h(\sigma)} = \frac{w(\sigma)}{(\lambda\gamma)^{-|\P|}}.
\]
We want to bound the total weight of all configurations in $\Osc$:
\[
w(\Osc) = (\lambda \gamma)^{-|\P|}  \sum_{\sigma \in \Osc}  \gamma^{-h(\sigma)} 
\]
It is the last sum in this expression that we will analyze using a polymer model.

Note that any configuration in $\Osc$ can be completely characterized by the locations of its heterogeneous edges. Because configurations in $\Osc$ have all boundary particles the same color, all heterogeneous edges must have at least one endpoint inside $\P$; recall the set of all such edges with at least one endpoint inside $\P$ is $E^{int}_\P$. Due to the lattice structure present, all heterogeneous edges form {\it loops}, where a {\it loop} is a minimal cut comprised of edges in $E^{int}_\P$. In other words, a {\it loop} is a set of edges whose removal separates $\Gtri$ into two connected components, and no subset of its edges has this property.  See Figure~\ref{fig:loops}, which depicts a particle configuration with three loops. The word loop comes from looking at the duals of these edge sets in $\Gtri$, which are simple cycles in the hexagon lattice (for more on lattice duality, see Section~\ref{sec:dual}).

Let $\mL$ be the set of all loops in $\Gtri$ of any size, and let $\mL_\P$ be all loops $E$ such that $E \subseteq E^{int}_\P$.
We say two loops are {\it compatible} if they do not share any of the same edges of $\Gtri$ (loops can still be compatible if some of their edges have common endpoints).  We say a collection of loops $\mL'  \subseteq \mL$ is {\it compatible} if all loops $L \in \mL'$ are pairwise compatible.
\begin{lem}\label{lem:bijection-loops}
	There is a bijection between configurations in $\Osc$ and compatible collections of loops $\mL' \subseteq \mL_\P$.
\end{lem}
\begin{proof}
Given a particle configuration $\sigma$ in $\Osc$, the heterogeneous edges form a collection of loops; we show that these loops are compatible and comprised of edges in $E^{int}_\P$. The latter follows from the fact that all particles in $\P$ have the same color but all edges in loops are heterogeneous, so no loops can contain edges along $\P$.  To see the former,
let $L$ be a loop of $\sigma$; because $L$ is minimal, every edge of $L$ must have its endpoints in different components of $\Gtri\setminus L$. If $L$ contains edge $e$ that is incident on triangular face $t$ of $\Gtri$, it must also contain another edge $e'$ incident on $t$, because if this was not the case the endpoints of $e$ would be in the same component of $\Gtri \setminus L$. However, any triangular face of $\Gtri$ can have at most two of its three incident edges heterogeneous by parity because there are only two colors. Thus if loops $L$ and $L'$ in $\sigma$ share an edge $e$ incident on triangular face $t$, then there is exactly one other heterogeneous edge incident on $t$ and that edge must be in both $L$ and $L'$.  Repeating this process, we see that if $L$ and $L'$ have a common edge $e$, then in fact $L = L'$.  We conclude if $L$ and $L'$ are distinct loops in $\sigma$, they share no edges, and thus are {\it compatible}. This implies the loops of $\sigma$ form a compatible collection of loops $\mL' \subseteq \mL_\P$.

Given a collection $\mL'$ of compatible loops, one can construct a valid configuration in $\Osc$ as follows.  For any particle $P$, find a path through occupied edges to a particle on $\P$. Let $k_P$ be the number of edges in a loop of $\mL'$ that are crossed on this path. Because all loops are minimal cut sets, for any such path the partiy of $k_P$ is the same. Assigning each particle $P$  color $c_{1 + k_P (\text{mod } 2)}$ yields a well-defined particle configuration $\sigma$ in $\Osc$.  The loops in $\mL$ correspond exactly to the heterogeneous edges in $\sigma$.

It is straightforward to see that these two maps are inverses of each other, and thus form a bijection.
\end{proof}

We can use this lemma to write $w(\Osc)$ as a polymer partition function, where our polymers are loops.
We say the polymer weight of a loop $L$ is $\gamma^{-|L|}$.  Using this and the previous lemma, it follows immediately that
\begin{align}\label{eqn:wt-xi-loop}
w(\Osc) = (\lambda\gamma)^{-|\P|} \sum_{\sigma \in \Osc} \gamma^{-h(\sigma)} = (\lambda\gamma)^{-|\P|} \sum_{ \substack{\mL' \subseteq \mL_\P \\ compatible}} \prod_{L \in \mL'} \gamma^{-|L|}
\end{align}
Thus to bound $w(\Osc)$ it suffices to examine this last sum, which is a polymer partition function for the loop-based polymer model we've defined, where our polymers are loops inside $\P$.  We define
\begin{align}\label{eqn:xi-loop}
\Xi_\P^{\mL} : =  \sum_{ \substack{\mL' \subseteq \mL_\P \\ compatible}} \prod_{L \in \mL'} \gamma^{-|L|}
\end{align}
Thus to bound $w(\Osc)$ it suffices to examine $\Xi^\mL_\P$, which we will do via its cluster expansion.
Specifically, we want to use Theorem~\ref{thm:bdry-volume} to separate $\Xi^{\mL}_{\P}$ into a surface term and a boundary term. First, we must show that the hypothesis of this theorem holds.
 Recall two loops are compatible if and only if they do not share any edges. That means that for any loop $L$, if $L'$ is incompatible with $L$ then $L'$ must share an edge with $L$; that is, $[L] = L$.

\begin{lem}\label{lem:loop-suff-cond}
For any edge $e \in \Gtri$, whenever $\gamma > \gsep \sim 5.66$ then for $c = \csep$,
\begin{align*}
\sum_{\substack{L \in \mL \\ e \in L}} \gamma^{-|L|} e^{c|L|} \leq c.
\end{align*}
\end{lem}
\begin{proof}
Let $n_k$ be the number of loops $L$ with $|L| = k$ containing a given fixed edge.  In general, $n_k \leq 2^k$, because as we saw in the proof of Lemma~\ref{lem:bijection-loops}, for any edge $e$ in $L$ incident on a triangular face $t$ of $\Gtri$, exactly one of the two other edges incident on $t$ must also be in $L$, so there are two ways to choose this next edge of $L$.  Any loop must have at least 6 edges, and must contain an even number of edges in total (because the dual hexagon lattice is bipartite). We can explicitly compute that $n_6 = 2$, $n_8 = 0$, $n_{10} = 10$, $n_{12} = 8$, $n_{14} = 56$.  We  see that
\begin{align*} \sum_{\substack{L \in \mL \\ e \in L}} \gamma^{-|L|} e^{c|L|}
& \leq \sum_{k = 3}^{\infty} n_{2k} e^{c2k} \gamma^{-2k}
 \\&\leq 2 \left(\frac{e^{c}}{\gamma}\right)^6
 \hspace{-2mm} + 10 \left(\frac{e^{c}}{\gamma}\right)^{10}
\hspace{-2mm} + 8  \left(\frac{e^{c}}{\gamma}\right)^{12}
\hspace{-2mm} + 56  \left(\frac{e^{c}}{\gamma}\right)^{14}
\hspace{-2mm} + \sum_{k = 8}^\infty 2^{2k} e^{c2k} \gamma^{-2k}
 \end{align*}
 Assuming $\gamma > 2e^c$, this geometric series converges, and we get that
 \begin{align*}
 \sum_{\substack{L \in \mL \\ e \in L}} \gamma^{-|L|} e^{c|L|} \leq
 	 2 \left(\frac{e^{c}}{\gamma}\right)^6
 	 \hspace{-2mm}+ 10 \left(\frac{e^{c}}{\gamma}\right)^{10}
 \hspace{-2mm}	+ 8  \left(\frac{e^{c}}{\gamma}\right)^{12}
 \hspace{-2mm}	+ 8  \left(\frac{e^{c}}{\gamma}\right)^{14}
 \hspace{-2mm}	+ \left(\frac{2e^{c}}{\gamma}\right)^{16}\frac{1}{1-(2e^c/\gamma)^2}
 \end{align*}
 We want this expression to be less than $c$.
 We choose a small value of $c$, $c = \csep$, such that one can verify  the desired inequality holds when $\gamma = \gsep$ and $c = \csep$. Noting that the equation above is monotone decreasing in $\gamma$, the desired inequality holds whenever $\gamma \geq \gsep$.
 We note that, indeed, $\gamma \geq \gsep$ implies $\gamma > 2e^c$, so our assumption earlier about the convergence of the geometric series was valid.
\end{proof}
We could improve this lower bound of $\gsep$ on $\gamma$ with a more optimal choice of $c$ and more careful enumeration of self-avoiding walks in the hexagon lattice; for example, the lemma still holds for $c = 0.05$ and $\gamma > 2.71$, though this is surely not optimal. However, in the next section we will need to assume $\gamma > \gsep$ in order to show separation occurs, so we opt not to include those details. We chose $c = \csep$ to be as small as possible such that the theorem holds for all $\gamma > \gsep$ (we rounded $c$ up to the nearest power of ten for simplicity).

 We are interested in polymers whose edges are drawn from set $\Lambda = E^{int}_\P$ consisting of all edges of $\Gtri$ with both endpoints on or inside $\P$ that are not part of $\P$. We can define $\partial \Lambda$, a set of edges, one of which must be in any cluster containing both an edge in $\Lambda$ and an edge not in $\Lambda$, to be $E(\P)$, the edges in boundary $\P$.

\begin{lem}
	When $\gamma > \gsep$, for $c = \csep$,
	there exists a constant $\Q \in [-c,c]$ that depends on $\gamma$ but is independent of $\P$ such that
	\[
	|E^{int}_\mathcal{P}|  \Q
	-  c|E(\mathcal{P})| \leq
	\ln\left(\Xi^\mL_\P\right)
	\leq |E^{int}_\mathcal{P}| \Q
	+  c|E(\mathcal{P})|.
	\]
\end{lem}
\begin{proof}
This follows immediately from Theorem~\ref{thm:bdry-volume}; that the hypothesis of this theorem is satisfied is Lemma~\ref{lem:loop-suff-cond}.
\end{proof}

These upper and lower bounds on $\ln \Xi^\mL_\P$ are what we will use to show that configurations that are not $\alpha$-compressed have exponentially small weight in the stationary distribution. Our arguments will focus on boundaries of configurations, so we now write the above bounds entirely in terms of $|\P|$ plus terms that are the same for every configuration.


\begin{lem} \label{lem:xi-bound}
		When $\gamma > \gsep$, for $c = \csep$,
	there exists a constant $\Q \in [-c,c]$ that depends on $\gamma$ but is independent of $\P$ such that
	\[
	(3n-3) \Q -3c|\P|
	\leq
	\ln\left(\Xi^\mL_\P\right)
	\leq (3n-3) \Q + 3c|\P|
	\]
\end{lem}
\begin{proof}
	We first note that $|E^{int}_\P| = |E_\P| - |E(\P)|$. Because an edge may be traversed at most twice in walk $\P$, we know $|\P|/2 \leq |E(\P)| \leq |\P|$. Using these facts, as well as the previously established relationship $E_\P = 3n - |\P| - 3$, we see that
	\begin{align*}|E^{int}_\mathcal{P}| = |E_\P| - |E(\P)| \geq 3n - |\P| - 3 - |\P| = 3n -3-2|\P| \\
	|E^{int}_\mathcal{P}| = |E_\P| - |E(\P)| \leq 3n - |\P| - 3 - \frac{|\P|}{2} = 3n - 3 - \frac{3}{2} |\P|.
	\end{align*}
	While we don't know if $\Q$ is positive or negative, we can still do case analysis to bound $\log \Xi^\mL_\P$ entirely in terms of $|\P|$ and a term that is identical for all particle configurations.
	If $\Q \geq 0$, then because we also know $\Q \leq c$,
	\begin{align*}
	|E^{int}_\P| \Q + c|E(\P)| &\leq \left(3n-3-\frac{3}{2}|\P|\right)\Q + c |\P| \leq (3n-3)\Q + c|\P| \\
	|E^{int}_\P| \Q -c |E(\P)| &\geq \left(3n-3-2|\P|\right) \Q - c|\P|  \geq (3n-3)\Q - 3c |\P|.\end{align*}
	If $\Q <  0$, then because we also know $\Q  \geq -c$,
	\[
	|E^{int}_\P| \Q + c|E(\P)| \leq \left(3n-3-2|\P|\right)\Q + c|\P| \leq (3n-3)\Q + 3c |\P|
	\]
	\[
	|E^{int}_\P| \Q - c|E(\P)| \geq \left(3n-3 - \frac{3}{2}\right)\Q - c|\P| \geq (3n-3)\Q - c|\P|.\]
	We conclude, by the previous lemma,  that
	\[
	(3n-3) \Q -3c|\P|
	\leq
	\ln\left(\Xi^\mL_\P\right)
	\leq (3n-3) \Q + 3c|\P|.
	\]
\end{proof}
Equivalently,
\[
e^{(3n-3) \Q -3c|\P| }
\leq
\Xi^\mL_\P \leq e^{(3n-3) \Q + 3c|\P|}.
\]

This means, in particular, that the ratios of $\Xi^\mL_\P$ and $\Xi^\mL_{\P'}$ for different boundaries $\P$ and $\P'$ that enclose the same number $n$ of particles can be bounded by an expression that is exponential in the lengths of these boundaries but independent of $n$. This is essential to our compression argument, which will focus on boundaries of various lengths.

\subsection{\texorpdfstring{Bounding the Partition Function of $\pi$}{Bounding the Partition Function of pi}}

Recall our stationary distribution $\pi$ includes a normalizing constant (partition function) $Z$ given by $Z = \sum_{\sigma \in \Omega} \left(\lambda \gamma\right)^{-p(\sigma)} \gamma^{-h(\sigma)}$. In this section we get a lower bound on $Z$. Recall that $w_\P(\sigma) = \gamma^{-h(\sigma)}$.

\begin{lem}\label{lem:complement-comp}
	For a configuration $\sigma \in \Osc$, it is possible to flip the colors of some particles to yield a configuration $f(\sigma) \in \Os$ with $n/2$ particles of each color such that $\gamma^{-|\P|} w(\sigma) \leq  w(f(\sigma))$. Furthermore, for any $\tau \in \Os$, there are at most $n$ different $\sigma \in \Osc$ such that $f(\sigma) = \tau$.
	\end{lem}
\begin{proof}
	If $\sigma \in \Os$, let $f(\sigma) = \sigma$ and the lemma holds.

	For $\sigma \in \Osc \setminus \Os$,
	label the particles of $\sigma$ in order from left to right and, within each column, from top to bottom.
	Flip the colors of particles in this order,
	 until there are the correct number of particles of each color.
	 If $\sigma$ has more than $n/2$ particles of color $c_i$, then after flipping the colors of all particles it has fewer than $n/2$ particles of color $c_i$. At some intermediate step, there must have been exactly $n/2$ particles of color $c_i$ and the configuration is in $\Os$, as desired. We let the first such configuration be $f(\sigma)$.

	 Because we flip all particles in one column before flipping any particles in the next column, all heterogeneous edges introduced by this process are in two adjacent columns.
	 If $h$ is the total height of $\sigma$ - the vertical difference between its lowest and highest particles - then the number of adjacencies between particles whose color was flipped and particles whose color was not flipped is at most $2h$.
	 This is an upper bound on the number of heterogeneous edges introduced by the flips. The height of a particle configuration is less than half its perimeter, so we conclude
	 the number of new heterogeneous edges is at most
	 $ 2h(\sigma) \leq p(\sigma) = |\P|$. Thus $w_\P(\sigma) \leq w_\P(f(\sigma))$.

	Given $\tau \in \Os$ and a number $k\in \{0,1,..., n-1\}$, complementing the colors of the first $k$ elements (according to the canonical ordering from above) of $\tau$ yields a configuration that maps to $\tau$ under~$f$. These $n$ configurations, which may or may not be in $\Osc$, are the only ones that could map to~$\tau$ under $f$.
\end{proof}

\begin{lem}\label{lem:Zbound}
\[ Z \geq 
 \frac{1}{n} e^{(3n-3) \Q} \left(e^{3c} \lambda \gamma^{2} \right)^{-p_{min}}
\]
\end{lem}
\begin{proof}
	We just sum the weights over a subset of the state space to get a lower bound on $Z$.  The subset we consider is those configurations in $\Omega_{\P_{min}}^{c_1}$, complemented so that they have the correct number of particles of each color as in Lemma~\ref{lem:complement-comp}. All such configurations have a contribution to their weight of $(\lambda \gamma)^{-p_{min}}$, and using $w(\sigma) = \gamma^{-h(\sigma)}$, we see that
	\[ Z \geq (\lambda \gamma)^{-p_{min}} \sum_{\sigma \in \Omega_{\P_{min}}^{c_1}} \frac{1}{n} w(f_3(\sigma))
	\geq \frac{(\lambda \gamma)^{-p_{min}}\gamma^{-p_{min}}}{n}  \sum_{\sigma \in \Omega_{\P_{min}}^{c_1}} w(\sigma)
	\]\[= \frac{(\lambda \gamma^{2})^{-p_{min}}}{n}  \Xi^\mL_{\P_{min}}
	\geq \frac{(\lambda \gamma^{2})^{-p_{min}}}{n}  e^{(3n-3) \Q - 3c p_{min}}.
	\]
	Rearranging terms gives the lemma.
\end{proof}

\subsection{Achieving Compression}

Fix $\alpha > 1$, and let $S_\alpha$ be all configurations that are not $\alpha$-compressed.  
We now state and prove our most general condition that implies compression occurs, and then follow this theorem with corollaries explaining what this general condition means for various values of $\alpha$, $\lambda$, and $\gamma$.

\begin{thm} \label{thm:comp-sep}
	Consider algorithm $\M$ when there are $n$ total particles of two different colors.  For $c = \csep$, when constants $\alpha>1$, $\lambda>1$, and $\gamma>\gsep$ satisfy
	\begin{align} \label{eqn:comp-sep-hyp}
	 \frac{2 (2+\sqrt{2}) e^{3c}}{\lambda \gamma}\left(e^{3c} \lambda \gamma^{3/2} \right)^{1/\alpha} <1 ,\end{align}
	when $n$ is sufficiently large then at stationarity for $\M$ with parameters $\lambda $ and $\gamma$, the probability that the configuration is not $\alpha$-compressed is exponentially small:
	\[ \pi(s_\alpha) < \zeta^{\sqrt{n}}
	\]
\end{thm}
\begin{proof}
	This proof builds on the proof of compression in~\cite{Cannon2016}; Lemma~\ref{lem:xi-bound} and its corollary Lemma~\ref{lem:Zbound} are the crucial facts we need to adapt this proof to our setting with particles of multiple colors.
	Pick $\nu$ such that $\nu > 2+\sqrt{2}$ but also
	\[\frac{2 \nu e^{3c}}{\lambda \gamma}\left(e^{3c} \lambda \gamma^{2} \right)^{1/\alpha} <1. \]
	Picking such a $\nu$ is possible because Equation~\ref{eqn:comp-sep-hyp} holds and is a strict inequality.
	We begin by expressing $\pi(S_\alpha)$ in terms of the contributions from configurations with given perimeters.
	\[ \pi(S_\alpha)
	= \frac{w(S_\alpha)}{Z}
	= \frac{\sum_{k = \lceil \alpha p_{min} \rceil}^{p_{max}} \sum_{\mathcal{P}: |\P| = k} w(\Omega_\P)}{Z}.
	\]
	Using Lemma~\ref{lem:Zbound} and Lemma~\ref{lem:elim-crossing} bounding the values of $w(\Omega_\P)$ and $Z$, we see that
	\[
	\pi(S_\alpha) \leq \frac{\sum_{k = \lceil \alpha p_{min} \rceil}^{p_{max}} \sum_{\mathcal{P}: |\P| = k} w(\Osc) 2^k \frac{2}{\gamma-2}}{\frac{1}{n} e^{(3n-3) \Q} \left(e^{3c} \lambda \gamma^{2} \right)^{-p_{min}}}.
	\]
	Recall that we expressed $w(\Osc) = (\lambda \gamma)^{-|\P|} \  \Xi^\mL_\P$ (Equation~\ref{eqn:wt-xi-loop}). Using Lemma~\ref{lem:xi-bound} to bound this polymer partition function $\Xi_\P^\mL$ and Lemma~\ref{lem:nu} to bound the number of configurations with perimeter of a certain length,
	we see that for sufficiently large $n$,
	\begin{align*}
		\pi(S_\alpha) &\leq \frac{\sum_{k = \lceil \alpha p_{min} \rceil}^{p_{max}}  \sum_{\mathcal{P}: |\P| = k} \left(\lambda \gamma\right)^{-k}
			e^{(3n-3)\psi + 3ck}
			 2^k \frac{2}{\gamma-2}}{\frac{1}{n} e^{(3n-3) \Q} \left(e^{3c} \lambda \gamma^{2} \right)^{-p_{min}}}.
		 \\& \leq \frac{\sum_{k = \lceil \alpha p_{min} \rceil}^{p_{max}} \left(\lambda \gamma\right)^{-k} \nu^k
		 	e^{(3n-3)\psi + 3ck}
		 	2^k \frac{2}{\gamma-2}}{\frac{1}{n} e^{(3n-3) \Q} \left(e^{3c} \lambda \gamma^{2} \right)^{-p_{min}}}.
		 \\&= \frac{2}{\gamma-2} \sum_{k = \lceil \alpha p_{min} \rceil}^{p_{max}} n \left(\frac{2 \nu e^{3c}}{\lambda \gamma} \right)^ k \left( e^{3c}\lambda \gamma^{2} \right)^{p_{min}}.
	\end{align*}
	Next we note that as $k \geq \alpha p_{min}$, then $p_{min} \leq k/\alpha$. Furthermore, as $k \geq p_{min} > \sqrt{n}$, then $n < k^2$. As $\lambda, \gamma > 1$ and $c > 0$, we know $e^{3c} \lambda \gamma^{2}>1$, and so
		\begin{align*}
	\pi(S_\alpha) &\leq \frac{2}{\gamma-2} \sum_{k = \lceil \alpha p_{min} \rceil}^{p_{max}} k^2 \left(\frac{2 \nu e^{3c}}{\lambda \gamma} \right)^ k \left( e^{3c} \lambda \gamma^{2} \right)^{k/\alpha}.
	\\& = \frac{2}{\gamma-2} \sum_{k = \lceil \alpha p_{min} \rceil}^{p_{max}} k^2
	\left( \frac{2 \nu e^{3c}}{\lambda \gamma}\left(e^{3c} \lambda \gamma^{2} \right)^{1/\alpha} \right)^k
	\end{align*}
We chose $\nu$ such that the term in parentheses above is less than one.
	For sufficiently large $n$ (i.e., sufficiently large $p_{min}$), this upper bound on $\pi(s_\alpha)$ is exponentially small. That is, for $n$ sufficiently large, there exists a constant $\zeta$ such that $\pi(S_\alpha) < \zeta^{\lceil\alpha p_{min}\rceil} < \zeta^{\sqrt{n}}$.  This proves the theorem.
\end{proof}

We note for any fixed value of $\alpha$ and $\gamma$, there exists a value of $\lambda$ large enough so that the hypothesis of the thoerem is satisfied.  The larger $\alpha$ is, the smaller $\lambda$ and $\gamma$ can be and still have the theorem apply.

\begin{cor}\label{cor:comp-sep-lambdagamma}
Whenever $\lambda >1$ and $\gamma > \gsep$ such that for $c = \csep$,
$\lambda \gamma > 2(2+\sqrt{2})e^{3c} \sim \lgsep$,
 there is an $\alpha$ such that for sufficiently large $n$, the probability that $\M$ with parameters $\lambda$ and $\gamma$ is not $\alpha$-compressed at stationarity is at most $\zeta^{\sqrt{n}}$ for $\zeta < 1$.
\end{cor}
\begin{proof}
Rearranging Equation~\ref{eqn:comp-sep-hyp}, we see that for a given choice of $\lambda$ and $\gamma$, the hypothesis of Theorem~\ref{thm:comp-sep} holds for any $\alpha$ satisfying
\[ \ln{\left(\frac{\lambda \gamma }{2(2+\sqrt{2})e^{3c}}\right)} \alpha >
\ln{(e^{3c}\lambda \gamma^{2})}\]
When $\lambda \gamma > 2(2+\sqrt{2})e^{3c}$, the logarithm on the left hand side above is positive, and we see that this equation holds exactly when
\[\alpha >
\frac{\ln{(e^{3c}\lambda \gamma^{2})}}{\ln{\left(\frac{\lambda \gamma }{2(2+\sqrt{2})e^{3c}}\right)}}\]
We can always find an $\alpha>1$ that is large enough to satisfy this equation. Thus, for any $\lambda$ and $\gamma$ satisfying $\lambda \gamma > 2(2+\sqrt{2})e^{3c}$, there is a constant $\alpha$ such that Equation~\ref{eqn:comp-sep-hyp} is satisfied and  $\alpha$-compression occurs with high probability.
\end{proof}

\begin{cor}\label{cor:comp-sep-alpha}
For any constant $\alpha > 1$, there exists values of $\gamma$ and $\lambda$ such that $\M$, with these parameters, exhibits $\alpha$-compression at stationarity with probability at least $1-\zeta^{\sqrt{n}}$ for $\zeta < 1$.
\end{cor}
\begin{proof}
Let $\lambda$ and $\gamma$ satisfy
\begin{align*}
\lambda^{\alpha -1} \gamma^{\alpha-{2}} >  e^{3c(\alpha+1)} 2^\alpha (2+\sqrt{2})^\alpha.
\end{align*}
Rearranging terms above shows the hypotheses of Theorem~\ref{thm:comp-sep} are satisfied and so $\alpha$-compression occurs with high probability.
\end{proof}

\section{Proof of Separation} \label{sec:separation}

We prove that, among configurations with small perimeter, those that are separated are exponentially more likely than those that are not in the stationary distribution of Markov chain $\M$, which is also the stationary distribution of our distributed algorithm $\A$.

Recall $\Os \subseteq \Omega$ is all configurations with no holes and boundary  $\mathcal{P}$.
Let $\pis$ be the stationary distribution conditioned on being in $\Os$, $\pi_\P(\sigma) = \pi(\sigma) / \pi(\Omega_\P)$.  Because all configurations in $\Os$ have the same perimeter, using the definition of $\pi$ given in Lemma~\ref{lem:statdist} we see that all terms of the form $(\lambda\gamma)^{-p(\sigma)}$ cancel, yielding  $\pis(\sigma) = \gamma^{-h(\sigma)} / \Zs$, where $\Zs = \sum_{\sigma \in \Os} \gamma^{-h(\sigma)}$.

Recall a configuration $\sigma$ is \emph{$\alpha$-compressed} if its perimeter is at most $\alpha \cdot p_{min}$, where $p_{min}$ is the minimum possible perimeter for the particles in $\sigma$.
Our main result in this section is that, for all $\mathcal{P}$ that determine $\alpha$-compressed configurations, non-separated configurations have exponentially small weight according to $\pis$. Because we already saw (Section~\ref{sec:compression-large-gamma}) that non-compressed configurations are exponentially unlikely, this suffices to show the occurrence of separation with high probability in configurations drawn from $\pi$.

We formally define separation in terms of the existence of a monochromatic region $R$ as follows. If $R$ is some subset of the particles in a configuration $\sigma$, then we say that $bd_{int}(R)$ is all edges of $\Gtri$ with both endpoints occupied by particles in $\sigma$ and exactly one endpoint in $R$.
For later use, we also define $bd_{out}(R) \subseteq E(\Gtri)$ as all edges where one point is occupied by a particle in $R$ and the other endpoint is unoccupied and $bd(R) = bd_{int}(R) \cup bd_{out}(R)$.
The following is the definition of separation we will use throughout this section; it is equivalent to Definition~\ref{defn:sepinformal} but stated in a more formal way.
\begin{defn}\label{defn:sep}
	For $\beta > 0$ and $\delta \in (0, 1/2)$, a configuration $\sigma \in \Os$ is \emph{$(\beta,\delta)$-separated} if there is a subset $R$ of particles such that:
	\begin{enumerate}
		\item $|bd_{int}(R)| \leq  \beta \sqrt{n}$;
		\item The density of particles of color $c_1$ in $R$ is at least $1-\delta$; and
		\item The density of particles of color $c_1$ not in $R$ is at most $\delta$.
	\end{enumerate}
\end{defn}
\noindent Here $\delta$ is a tolerance of having particles of the wrong color within the cluster $R$, and $\beta$ is a measure of how  small the boundary between $R$ and $\overline{R}$, the particles  not in $R$, must be. We note that Condition (1) is equivalent to having at most $\beta \sqrt{n}$ edges with one endpoint in $R$ and one endpoint in $\overline{R}$.
We note that this definition is symmetric with respect to the role played by $c_1$ in~$R$ and the role played by $c_2$ in $\overline{R}$.
$R$ does not need to be connected or hole-free.

We let $\mS_{\beta,\delta} \subseteq \Os$ be the configurations in $\Os$ that are $(\beta,\delta)$-separated for $\beta> 0$ and some $\delta < 1/2$. We prove (Theorem~\ref{thm:sep}) that for $\gamma$ sufficiently large, as long as $\mathcal{P}$ is $\alpha$-compressed, $\beta > 2 \alpha \sqrt{3}$, and $\delta< 1/2$, with all but exponentially small probability a sample drawn from $\pis$ is in $\mS_{\beta,\delta}$:
$$\pis(\Os \setminus \mS_{\beta,\delta}) \leq \zeta^{\sqrt{n}},$$
where $\zeta$ is a constant less that one.
In the remainder of this section we prove this result.


\subsection{Lattice Duality and Contours}
\label{sec:duality}\label{sec:dual}

We begin with some background on lattice duality that will simplify our proofs in the remainder of this section.
The dual to the triangular lattice $\Gtri$, obtained by creating a new vertex in every face of $\Gtri$ and connecting two of these vertices if their corresponding triangular faces have a common edge, is the hexagonal lattice $\Ghex$; see Fig.~\ref{fig:dual}.
There is a bijection between edges of $\Gtri$ and edges of $\Ghex$,
associating an edge of $\Gtri$ with the unique edge of $\Ghex$ that crosses it and vice versa.

Throughout, by a {\it contour} we will mean a walk in $\Ghex$ that never visits the same vertex twice, except possibly to start and end at the same place; these are also known as {\it self-avoiding walks} or, when starting and ending at the same place, {\it self-avoiding polygons}.

\begin{figure}
	\centering
	\begin{subfigure}{0.4\textwidth}
		\centering
		\includegraphics[scale = 0.58]{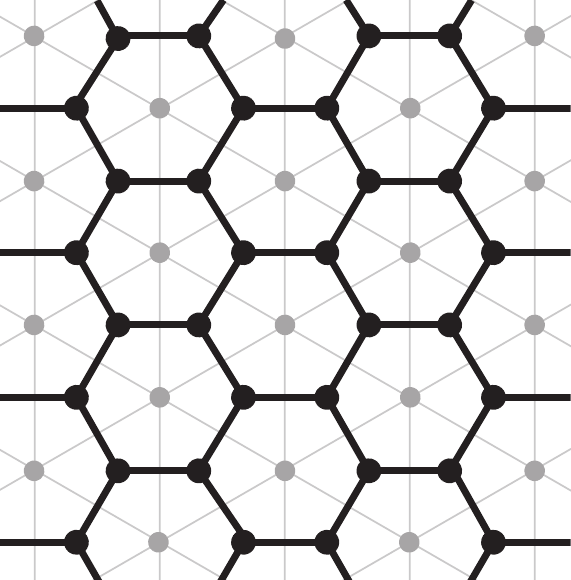}
		\caption{\centering}
		\label{fig:dual}
	\end{subfigure}%
	\begin{subfigure}{0.4\textwidth}
		\centering
		\includegraphics[scale = 0.4]{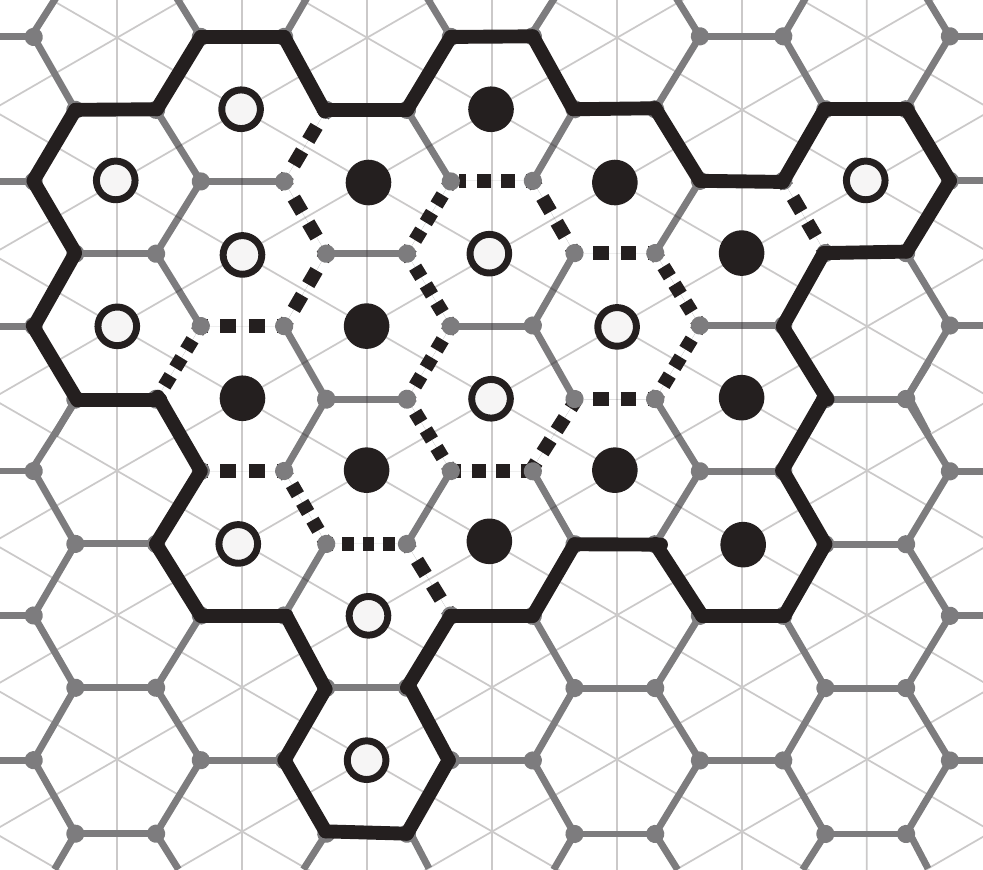}
		\caption{\centering}\label{fig:duality-edges}
	\end{subfigure}
	\caption{
		(a) The duality between the triangular lattice $\Gtri$ (gray) and the hexagonal lattice $\Ghex$ (black).
		(b) A particle configuration $\sigma$ with 11 black particles and 11 white particles. The boundary contour in dual lattice $\Ghex$ is thick and black, while all four heterogeneous contours of $\sigma$ are shown by dashed lines.
}
	\label{fig:SAW-ex}
\end{figure}

Each edge $e \in \Ghex$ crosses a unique edge $f \in \Gtri$, and we  say an $e$ {\it separates} the two locations connected by $f$. 
For a configuration $\sigma$, we say an $e \in \Ghex$ is a {\it boundary edge} if it separates a particle of $\sigma$  from an unoccupied location, and $e$ is a {\it heterogeneous edge} if it separates two particles of different colors.
A contour is a {\it boundary contour} if all of its edges are boundary edges and is a {\it heterogeneous contour} if all of its edges are heterogeneous.
See Fig.~\ref{fig:duality-edges} for an example of a configuration $\sigma$ with particles of two different colors and its boundary contour (black) and heterogeneous contours (dashed).



For a particle configuration $\sigma$ without holes, its boundary $\P$ can be completely described by taking the union of all boundary edges in $\Ghex$, which yields a boundary contour in $\Ghex$ which we will call $\P_{hex}$. 
A result of~\cite{Cannon2016} implies that if $|\P| = k$, then $|\P_{hex}| = 2k+6$.\footnote{In~\cite{Cannon2016} we showed the length of a self-avoiding walk including all but one boundary edge of $\sigma$ in $\Ghex$ had length $2p(\sigma) + 5$; here we consider (closed) boundary contours with all boundary edges, of total length $2p(\sigma) + 6$.}
%
For a particle configuration with no holes and boundary $\P$ or boundary contour $\mathcal{P}_{hex}$, we can completely describe the colors of its particles (up to swapping the colors) by giving all heterogeneous edges.  Because there are only two colors, the heterogeneous edges in $\Ghex$ form non-intersecting contours because by parity every vertex of $\Ghex$ has either zero or two incident heterogeneous edges.
Each (maximal) heterogeneous contour either starts and ends at different places on the boundary contour or is a closed loop; we call the former a {\it crossing contour} and the latter an {\it isolated contour}.
The configuration in Fig.~\ref{fig:duality-edges} has three crossing contours and one isolated contour.
Loop contours are exactly the dual of the loop polymers discussed in Section~\ref{sec:compression-large-gamma}.

The crossing contours of a configuration $\sigma$ separate the particles into simply connected components whose boundary particles all have the same color.
 Recall from Section~\ref{sec:compression-large-gamma} that a {\it face} of a particle configuration $\sigma$ is a maximal simply connected subset $F$ where all particles in $F$ incident on an edge of $bd(F)$ have the same color, which we call the {\it color} of $F$.
For any face $F$, its maximality implies all edges in $bd_{int}(F)$ are heterogeneous in $\sigma$. A face is {\it outer} if it includes at least on particle on $\P$.


\subsection{Bridging Systems}


Let $(B,I)$ be a collection of contours in $\Ghex$ within a face $F$, where $B$ contains {\it bridge contours} connecting each isolated contour in set $I$ (a subset of the isolated contours within $F$) to the boundary of $F$. For a given $(B,I)$, we say particle $P$ is {\it bridged} in face $F$ if there exists a path through particles of the same color as $P$ to $bd(F)$ or to a bridged isolated contour in $I$.  A particle is {\it unbridged} if such a path does not exist.  
We say that $(B,I)$ is a {\it $\delta$-bridge system} for face $F$ if:
\begin{enumerate}
	\item $|B| \leq |I| (1-\delta)/2\delta$, where $|B|$ is the total number of edges in all the bridge contours in $B$ and $|I|$ is the total number of edges in all the bridged isolated contours in $I$.
	\item The number of unbridged particles in $F$ is $\leq \delta |F|$, where $|F|$ is the number of particles in $F$.
\end{enumerate}
Note the $\delta$ in this definition is the same $\delta$ as in the set $\mS_{\beta,\delta}$ that we are trying to show has exponentially small weight.
We now show how to find a $\delta$-bridge system for any face $F$.
\begin{lem}\label{lem:bridge}
	For any face $F$, there exists a $\delta$-bridge system for~$F$.
\end{lem}
\begin{proof}
	Fig.~\ref{fig:bridge} gives one example of a face $F$ and a $\delta$-bridge constructed for $F$.

	Without loss of generality, suppose $F$ is of color $c_1$.
	If $F$ has only one particle, then $(\emptyset, \emptyset)$ is a $\delta$-bridge system for $F$. We now suppose $F$ has more than one particle and there exists a $\delta$-bridge system for all regions with a smaller number of particles than $F$.
	We will iteratively construct a $\delta$-bridge system $(B,I)$ for $F$.  To start, let $(B,I) = (\emptyset, \emptyset)$, which satisfies $|B| \leq |I| (1-\delta)/2\delta$.
	Let $u(F)$ be the unbridged particles for $(B,I)$ in $F$.
	If $|u(F)| \leq \delta |F|$, where $|F|$ is the number of particles in face $F$, then $(B,I)$ is a valid $\delta$-bridge system for $F$.
	If not, we give a procedure for adding to $(B,I)$ that reduces the number of unbridged particles in $F$ and maintains two invariants: (1) $|B| \leq |I| (1-\delta)/2\delta $ and (2) for any $\mathcal{I} \in I$ not surrounded by another contour in $I$, the face $F_{\mathcal{I}}$ consisting of all particles inside $\mathcal{I}$ contains at most $\delta |F_{\mathcal{I}}|$ unbridged particles. Both invariants are true for initial configuration $(\emptyset, \emptyset)$. Repeating this process until $u(F) \leq \delta |F|$ gives a valid $\delta$-bridge for~$F$.

	Suppose we are given a bridge system $(B,I)$ for $F$ that satisfies both invariants but leaves $u > \delta |F|$ unbridged particles.
%
	%
	%
	Let $\Fout$ be the particles in $F$ that are not inside any bridged isolated contours in $(B,I)$. 
	We will consider contours $\mathcal{V}$ in $\Ghex$ that stretch vertically across $F$, from one part of its boundary to another, consisting of alternating down-left and down-right edges.  We call such contours {\it vertical} contours.
	We include in set $\mathbb{V}_F$ all (infinite) vertical contours that contain at least one edge inside $\Fout$; we will only be interested in their intersection with $\Fout$, which need not be contiguous.
	For any $\mV \in \mathbb{V}_F$, let $\mH\cap \Fout$ be all particles in $\Fout$ directly right of $\mH$ and let $\mV \cap u(\Fout)$ be the unbridged ones.
	Because $u(F) > \delta |F|$, applying Invariant (2)
	we conclude that $u(\Fout) > \delta |\Fout|$.
	It follows that there exists $\overline{\mH} \in \mathbb{V}_F$ such that $|\mH \cap u(\Fout)| > \delta |\mH \cap \Fout|$.

%

\begin{figure}
	\centering
	\includegraphics[scale = 0.55]{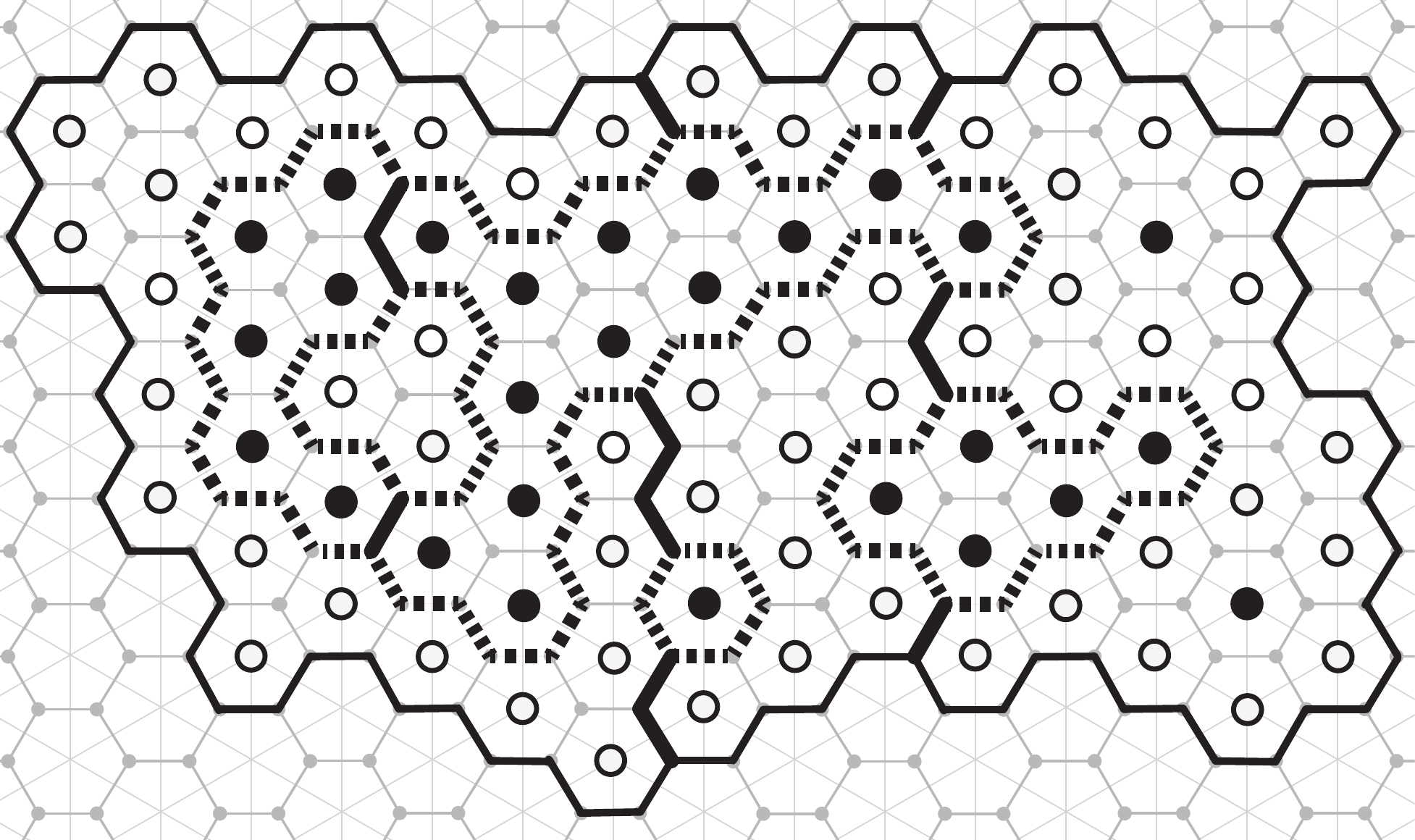}
	\caption{A face $F$ and one potential $\delta$-bridge $(B,I)$ for $F$, where $B$ consists of think black edges and $I$ consists of all dashed edges. The boundary of $F$ in $\Ghex$ is shown as thin black lines.  }\label{fig:bridge}
\end{figure}

	Any particle $P \in \overline\mH \cap u(\Fout)$  must be surrounded by an unbridged isolated contour, as otherwise it would have a monochromatic path to the boundary of $F$; if there are multiple isolated contours surrounding $P$, one must be the {\it outermost}, encircling all the others. 
	Enumerate all outermost isolated contours surrounding particles in $u(\Fout) \cap \overline\mH$ as $\mathcal{I}_j$ for $j = 1,..., k$. Let $F_j$ be the face surrounded by $\mathcal{I}_j$, which is of color $c_2$.
	By our induction hypothesis, because $|F| > |F_j|$ there exists a $\delta$-bridge system $(B_j, I_j)$ for $F_j$.
	We add to bridge system $(B, I)$ for $F$ the set of bridges
	$\bigcup_{j} B_j$ and the set of bridged isolated contours $\bigcup_{j}I_j$.
	Furthermore, we add to $B$ all the segments of $\overline\mH$ that are left of bridged particles in $\overline\mH \cap F_{ext}$, a set we call $B_0$, and we add to $I$ all $\mathcal{I}_j$.
	Because the number of particles that are newly-bridged by this construction is at least $|u(F) \cap \overline\mH|$, we have reduced the number of unbridged particles in $F$. It only remains to show that this new bridge system satisfies the necessary invariants.

	To see that $(B,I)$ satisfies Invariant 2, note that the only new contours $\mathcal{I} \in I$ not surrounded by other contours in $I$ are the $\mathcal{I}_j$.
	All particles that were bridged in any $F_{\mathcal{I}_j} = F_j$ are now bridged in $F$, since both the boundary of $F_j$ and the bridged contours in $I_j$ are now bridged contours in $I$.
	Because $(B_j, I_j)$ is a valid $\delta$-bridge system for $F_{\mathcal{I}_j} = F_j$, $F_j$ contains at most $\delta|F_j|$ unbridged particles, as desired.




	We now check that $(B,I)$ satisfies Invariant 1.
	Because $(B_j, I_j)$ is a $\delta$-bridge for $F_j$,  $|B_j| \leq |I_j|(1-\delta)/2\delta$ for all $j$. Next, we see that
	 $\sum_j |\mathcal{I}_j| \geq 4 \cdot |u(\Fout) \cap \mH|$, as the $\mathcal{I}_j$ collectively contain at least two contour edges left of and two contour edges right of each particle in $u(\Fout) \cap \overline\mH$.
	Because $\overline\mH$ satisfies $|\overline\mH \cap u(\Fout)| > \delta |\overline\mH \cap \Fout|$, then $\sum_j |\mathcal{I}_j| \geq 4\delta |\overline\mH \cap \Fout|$. Bridge $B_0$ added to $B$ contains two contour edges for each bridged particle in $\Fout \cap \overline\mH$ and at most a $1-\delta$ fraction of the particles in $\Fout \cap \overline\mH$ are bridged, so $|B_0|/ 2 \leq (1-\delta) | \overline\mH \cap\Fout|$. Combining the previous two equations,
	\begin{align*} \sum_j |\mathcal{I}_j| \geq 4\delta |\overline\mH \cap\Fout| \geq  4\delta \left( \frac{1}{2(1-\delta)} |B_0|\right) = \frac{2\delta}{1-\delta} |B_0|. \end{align*}
	We conclude that the additions $B_0$ and $B_j$ to $B$ and the additions $\mathcal{I}_j$ and $I_j$ to $I$ satisfy
	\begin{align*} |B_0| + \sum_{j = 1}^k |B_j| \leq \frac{1-\delta}{2\delta} \sum_{j = 1}^k |\mathcal{I}_j| + \frac{1-\delta}{2\delta} \sum_{j = 1}^k |I_j| = \frac{1-\delta}{2\delta}  \sum_{j = 1}^k \left( |\mathcal{I}_j| + |I_j| \right).
	\end{align*}
	Thus Invariant 1 is satisfied.  

	We have added to $(B,I)$ while maintaining both invariants and reducing the number of unbridged particles in $F$. We can continue this process until there are at most $\delta |F|$ unbridged particles in $F$; then, Invariant 1 implies $(B,I)$ is a $\delta$-bridge system for $F$.
	%
\end{proof}

\begin{lem}\label{lem:bridge-overall}
	For each $\sigma \in \Os$ with $n$ particles, there exists a $\delta$-bridge system $(B,I)$ for $\sigma$, where $B$ contains bridge contours connecting each isolated contour in set $I$ (a subset of $\sigma$'s isolated contours) to $\sigma$'s boundary contour or to a crossing contour, such that:
	\begin{itemize}
		\item $|B| \leq |I| (1-\delta)/2\delta$, and
		\item The number of unbridged particles in  $\sigma$ is at most $\delta n$. 
	\end{itemize}
\end{lem}
\begin{proof}
	The crossing contours of $\sigma$ partition $\sigma$ into faces. Construct a $\delta$-bridge system for each of these faces and take their union. 
\end{proof}

\noindent We now connect the notions of $\delta$-bridges and configurations that are $(\beta,\delta)$-separated.


\begin{lem}\label{lem:not-compressed}
	Let $\sigma \in \Os \setminus \mS_{\beta, \delta}$ and let $(B,I)$ be the $\delta$-bridge system  for $\sigma$ constructed in Lemma~\ref{lem:bridge-overall}. Let $x$ be the total length of crossing contours in $\sigma$ and let $y$ be the total length of bridged isolated contours in $I$. Then  $ x + y > \beta \sqrt{n}$.
\end{lem}
\begin{proof}
	Let $\mathcal{F}$ be the set of outermost faces of $\sigma$, that is, those faces of $\sigma$ that contain a particle on $\sigma$'s perimeter.
	For each $F \in \mathcal{F}$ of color $c_i$, if particle $P \in F$ is surrounded by $b$ bridged isolated contours then put $P$ in set $R$ if and only if $i + b \equiv 1 (\text{mod } 2)$.
	Because of how we have carefully defined $R$, inspection shows  $bd_{int}(R) = x + y$. Using the properties of $\delta$-bridge system $(B,I)$, one can show the density of particles of color $c_1$ is at least $1-\delta$ in $R$ and at most $\delta$ outside of $R$. If it were true that $x + y \leq \beta \sqrt{n}$, then $\sigma$ would be $(\beta,\delta)$-separated, a contradiction as $\sigma \notin \mS_{\beta, \delta}$. Thus, it must hold that   $x + y > \beta \sqrt{n}$.
\end{proof}

\subsection{Information Theoretic Argument for Separation}

To show the set $\Os \setminus \mS_{\beta, \delta}$ of configurations with boundary contour $\mathcal{P}$ that are not $(\beta, \delta)$-separated has exponentially small weight under distribution $\pis$,
we will define a map $f = f_3 \circ f_2 \circ f_1$ from this set into $\Os$ and examine how this map changes weights of configurations.
If the number of particles of one color is less than or equal to $\delta n$, then all configurations in $\Os$ are $(\beta, \delta)$-separated with $R = \emptyset$ or $\overline{R} = \emptyset$, so we assume each color class has more than $\delta n$ particles.

For $\sigma \in \Os \setminus \mS_{\beta, \delta}$, let $(B,I)$ be the $\delta$-bridge system constructed for $\sigma$ according to Lemma~\ref{lem:bridge-overall}. Let $f_1(\sigma)$ be the (unique) particle configuration that has the same boundary contour $\mathcal{P}$ as $\sigma$ and particle $P$ that has color $c_i$ in $\sigma$ and is surrounded by $b$ bridged isolated contours in $I$ is given color $c_{(i + b) (\text{mod } 2)}$ in $f_1(\sigma)$.
We let $Im(f_1(\Os\setminus \mS_{\beta, \delta}))$ be the set of configurations that $f_1$ maps to.

We define $f_2$ with domain $Im(f_1(\Os\setminus \mS_{\beta, \delta}))$ to complement all faces of color $c_2$ that touch the boundary of the configuration (i.e. that include particles on $\mathcal{P}$).
The next lemmas explore the composition of these maps $f_1$ and $f_2$ as applied to configurations $\sigma \in \Os\setminus \mS_{\beta, \delta}$.

\begin{lem}\label{lem:f12}
	For any $\sigma \in \Os \setminus \mS_{\beta, \delta}$, $f_2(f_1(\sigma)) $ has boundary contour $\mathcal{P}$, all particles adjacent to $\mathcal{P}$ have color $c_1$, and there are at most $\delta n$ particles of color $c_2$.
\end{lem}
\begin{proof}
	The first two claims follow easily from the definitions of $f_1$ and $f_2$. To see that the last claim holds, we note that any particles of color $c_2$ in $f_2(f_1(\sigma))$ must have been unbridged by the bridge system $(B,I)$ for $\sigma$, and there are at most $\delta n$ such unbridged particles by the definition of a $\delta$-bridge system.
\end{proof}


\begin{lem}\label{lem:preimage_contours}
	Let $\tau \in Im((f_2 \circ f_1)(\Os\setminus \mS_{\beta, \delta}))$. The number of $\sigma\in \Os \setminus \mS_{\beta, \delta}$ with crossing contours of total length $x$ and bridged isolated contours (bridged by a bridging system $(B,I)$ from Lemma~\ref{lem:bridge-overall}) of total length $y$ that have $f_2(f_1(\sigma)) = \tau$ is, for $p = |\mathcal{P}|$ the perimeter of any configuration in $\Os$, at most $3^p 4^{(x+y)\left(\frac{1+3\delta}{4\delta}\right)}$.
\end{lem}
\begin{proof}
	Any configuration $\sigma \in \Os \setminus \mS_{\beta, \delta}$ has boundary $\P$ or length $p$ and boundary contour $\mathcal{P}_{hex}$ of length $2p + 6$. One can verify from first principles that $\mathcal{P}_{hex}$ 
	makes $p$ left turns and $p + 6$ right turns when traversed clockwise. Any bridges or crossing contours that meet $\mathcal{P}$ do so at distinct left turns of $\mathcal{P}$. 
	 We can mark each left turn of $\mathcal{P}$ as the start of a bridge, the start of a crossing contour, or neither; the number of ways to do so is $3^p$.

	Next, we can trace out all crossing contours of $\sigma$, beginning at the starting points marked along~$\mathcal{P}$. In tracing these contours, which do not intersect, at each vertex in $\Ghex$ we make either a left turn or a right turn.  Additionally, each vertex along these contours can either be the beginning of a bridge in $B$, branching in the opposite direction from the contour, or not. Because $x$ is the total length of $\sigma$'s crossing contours, the number of valid ways to do this is at most $2^x \times 2^x = 4^x$.

	Finally, we trace out the bridges and isolated contours of each face of $\sigma$ in a depth-first way, beginning at the starting points marked along $\mathcal{P}$ and the crossing contours. Bridges as constructed in Lemma~\ref{lem:bridge-overall} always move in the vertical direction, so the direction of the next edge of a bridge, if it exists, is known; at each step we only need to know if the bridge continues or if a bridged isolated contour begins.
	When tracing out isolated contours, just like with heterogeneous crossing contours, there are four choices for the next step: the direction in which the contour continues (two choices) and whether or not a bridge branches off (two choices).  Isolated contours end when they reach an already-constructed bridge, and bridges end when they reach a crossing contour, an already-constructed isolated contour, or $\mathcal{P}$.
	The number of possibilities for this depth-first traversal of the bridges and isolated contours of $\sigma$ is at most $2^{|B|} 4^{|I|} \leq 2^{\frac{1-\delta}{2\delta} y} 4^{y}$.

	Altogether, any configuration $\sigma \in \Os \setminus \mS_{\beta, \delta}$ with crossing contours of total length $x$ and bridged isolated contours of total length $y$ that have $f_2(f_1(\sigma)) = \tau$ can be uniquely identified by marking~$\mathcal{P}$, tracing crossing contours, and tracing bridges and bridged isolated contours. The number of valid ways to do this is at most
	\begin{align*} 3^{p} 4^{x} 2^{\frac{1-\delta}{2\delta} y} 4^{y} = 3^{p} 4^{x +  y + \frac{1-\delta}{4\delta} y } \leq 3^p 4^{(x+y)\left(\frac{1+3\delta}{4\delta}\right)}. \end{align*}
	This is an upper bound on the number of preimages of $\tau$ under $f_2 \circ f_1$ with  correct $x$~and~$y$.
\end{proof}

Any $\tau \in Im((f_2 \circ f_1)(\Os \setminus \mS_{\beta, \delta})$ will not be in $\Os$ because it has too few particles of color $c_2$. We will define $f_3$ such that $f_3(\tau)$ is similar to $\tau$ and has the correct number of particles of each color, but we first need the following lemma.


\begin{lem}\label{lem:complement}\label{lem:preimage-nu}
	For a configuration $\tau\in Im((f_2 \circ f_1)(\Os \setminus \mS_{\beta, \delta}))$,
	it is possible to flip the colors of some particles such to yield a configuration $f_3(\tau)$ with the correct number of particles of each color such that at most $|\P|$ additional heterogeneous edges are introduced.  Furthermore, for any $\nu \in \Omega_\P$, there are at most $n$ different $\tau  \in  Im((f_2 \circ f_1)(\Os \setminus \mS_{\beta, \delta}))$ such that $f_3(\tau) = \nu$.
\end{lem}
\begin{proof}
	Note that any configuration in $Im((f_2 \circ f_1)(\Os \setminus \mS_{\beta, \delta}))$ has all particles on $\P$ of color $c_1$. Recall that in Section~\ref{sec:compression-large-gamma}, we called the set of all such configurations $\Osc$. Lemma~\ref{lem:complement-comp} for $\Osc$, a superset of $Im((f_2 \circ f_1)(\Os \setminus \mS_{\beta, \delta}))$, directly implies this result.
\end{proof}

Let $f = f_3 \circ f_2 \circ f_1$ be a map from $\Os \setminus \mS_{\beta, \delta}$ to $\Os$.
For $\sigma \in \Os \setminus \mS_{\beta, \delta}$, let $x$ be the total length of crossing heterogeneous contours in $\sigma$ and $y$ be the total length of all isolated contours in $\sigma$ that are bridged when constructing a $\delta$-bridge system according to the process of Lemma~\ref{lem:bridge}.

\begin{lem}\label{lem:gain}
	For $\sigma \in \Os \setminus \mS_{\beta, \delta}$ where $\P$ is $\alpha$-compressed,  $h(\sigma) - h(f(\sigma)) \geq (x + y) \left( 1 - \frac{2\alpha\sqrt{3} }{\beta}\right).$
\end{lem}
\begin{proof}
	Configuration $f_1(\sigma)$ has $y$ fewer heterogeneous edges than $\sigma$, and configuration $f_2(f_1(\sigma))$ has $x$ fewer heterogeneous edges than $f_1(\sigma)$.  When going from $f_2(f_1(\sigma))$ to $f(\sigma) = f_3(f_2(f_1(\sigma)))$, at most $2\alpha\sqrt{3} \sqrt{n}$ heterogeneous edges are added (Lemma~\ref{lem:complement}). Using Lemma~\ref{lem:not-compressed}, we conclude that
	\begin{equation*}
	h(\sigma) - h(f(\sigma)) \geq x + y - |\P| \geq x + y - 2\alpha\sqrt{3} \sqrt{n}
	\geq  x + y - 2 \alpha \sqrt{3}  \left( \frac{x + y}{\beta}\right)
	\geq (x + y) \left( 1 - \frac{2 \alpha\sqrt{3} }{\beta}\right).
	\end{equation*}
\end{proof}

We are now ready to prove our main result. Recall that for a fixed boundary $\P$, the probability distribution $\pi_\P$ is over colored particle configurations with this boundary where $\pi_\P(\sigma)$ is proportional to $\gamma^{-h(\sigma)}$.

\begin{thm}\label{thm:sep}
	Let $\P$ be the boundary of $n$ particles with $|\P| \leq \alpha p_{min}$. For any $\beta > 2\sqrt{3}\alpha$ and any $\delta < 1/2$, if $\gamma$ is large enough that \[3^{\frac{2\alpha\sqrt{3}}{\beta}}4^{ \frac{1+3\delta}{4\delta}} \gamma^{-1 + \frac{2\alpha\sqrt{3} }{\beta} } < 1\]
	then for sufficiently large $n$ the probability that a configuration drawn from $\pi_\P$ is not $(\beta,\delta)$-separated is exponentially small:
	\[ \pi_{\P}(\Omega_\P\setminus \mS_{\beta,\delta})< \zeta^{\sqrt{n}}\]
	where $\zeta < 1$.
\end{thm}
\begin{proof}
	For any $\nu \in \Os$, we count the number of configurations in $\Os \setminus \mS_{\beta, \delta}$ such that ${f(\sigma) = \nu}$. By Lemma~\ref{lem:preimage_contours},
	the number of such preimages with crossing contours of total length $x$ and bridged isolated contours of total length $y$ is at most
	$ { n 3^p 4^{(x+y)(1+3\delta)/4\delta}}, $
	where $p = |\P|$.
As ${p < \alpha p_{min} < 2\alpha\sqrt{3}\sqrt{n}}$, by
 Lemma~\ref{lem:not-compressed}
${p < 2 \alpha \sqrt{3}(x+y)/\beta} $.
	We can rewrite the number of preimages in $f^{-1}(\nu)$ with given values of $x$ and $y$ as
	\begin{align*} n   3^p 4^{(x+y)\left(\frac{1+3\delta}{4\delta}\right)}
	< n 3^{2 \alpha\sqrt{3} \left(\frac{x+y}{\beta}\right)}4^{(x+y)\left(\frac{1+3\delta}{4\delta}\right)}
	= n \left( 3^{\frac{2\alpha\sqrt{3}}{\beta}}4^{ \frac{1+3\delta}{4\delta}}\right)^{x+y} . \end{align*}

	We now sum over all possible values of $x+y$. For each possible value of $x+y$, there are at most $x+y+1$ ways in which each of $x$ and $y$ could have contributed to this sum.   By Lemma~\ref{lem:not-compressed}, $x+y > \beta \sqrt{n}$, and because the edges counted in $x+y$ are a subset of all edges in the configuration, $x + y < 3n$.
	We conclude, for $z = x+y$,
	\begin{align*}|f^{-1}(\nu)|  \leq  n  \sum_{z = \lceil\beta\sqrt{n}\rceil}^{3n} (z + 1) \left( 3^{\frac{2\alpha\sqrt{3}}{\beta}}4^{ \frac{1+3\delta}{4\delta}}\right)^z .\end{align*}
	Finally, we see that for any $\nu \in \Os$, using Lemma~\ref{lem:gain},
	\begin{align*}
	\frac{\sum_{\sigma \in f^{-1}(\nu)} \pis (\sigma)}{\pi_\P(\nu) }
	&= \sum_{\sigma \in f^{-1}(\nu)} \left(\frac{1}{\gamma}\right)^{h(\sigma) - h(f(\sigma))}
	\\&\leq n \sum_{z = \lceil\beta\sqrt{n}\rceil}^{3n} (z+1) \left(\frac{1}{\gamma}\right)^{z \left( 1 - \frac{2\alpha\sqrt{3} }{\beta}\right)}
	\\&\leq n \sum_{z = \lceil\beta\sqrt{n}\rceil}^{3n} (z+1) \left( 3^{\frac{2\alpha\sqrt{3}}{\beta}}4^{ \frac{1+3\delta}{4\delta}} \gamma^{-1 + \frac{2\alpha\sqrt{3} }{\beta} } \right)^z
	\end{align*}
	This sum is exponentially small whenever the number of particles $n$ is sufficiently large and the base of the exponent satisfies  $3^{\frac{2\alpha\sqrt{3}}{\beta}}4^{ \frac{1+3\delta}{4\delta}} \gamma^{-1 + \frac{2\alpha\sqrt{3} }{\beta} } < 1$.
	Whenever $\beta > 2\alpha\sqrt{3}$, $\delta < 1/2$, and
	$\gamma$ is large enough this is true, so we can find a constant $\zeta < 1$ such that for sufficiently large $n$, \[\frac{\sum_{\sigma \in f^{-1}(\nu)} \pis (\sigma)}{\pi_\P(\nu) } < \zeta^{\lceil \beta \sqrt{n}\rceil} < \zeta^{\sqrt{n}}.\]
	Because each $\sigma \in \Omega_\P \setminus \mS_{\beta,\delta}$ has some image $f(\sigma) \in \Omega_\P$, we use this fact to see that
	\[
	\pi_\P(\Omega_\P \setminus \mS_{\beta,\delta})	 =  \sum_{\sigma \in \Omega_\P \setminus \mS_{\beta,\delta}} \pis(\sigma) \leq \sum_{\nu \in \Omega_\P} \ \  \sum_{\sigma \in f^{-1}(\nu)} \pi_\P(\sigma) \leq \sum_{\nu \in \Omega_\P} \pi_\P(\nu) \zeta^{\sqrt{n}} = \zeta^{\sqrt{n}}.
	\]
	We conclude that when $n$ is sufficiently large, $\beta > 2\alpha\sqrt{3}$, $\delta < 1/2$, and
	$\gamma$ is large enough, the probability a particle configuration drawn from $\pi_\P$ is not $(\beta,\delta)$-separated is exponentially small.
\end{proof}

We now extend this result about the occurrence of separation when fixing an  $\alpha$-compressed boundary $\P$ to a statement about the occurrence of separation among all $\alpha$-compressed boundaries.
Let $\pi_\alpha$ be the probability distribution over all configurations that are $\alpha$-compressed obtained by restricting $\pi$ to this set, so that $\pi_\alpha(\sigma)$ is proportional to $(\lambda\gamma)^{-p(\sigma)} \gamma^{-h(\sigma)}$.
We obtain the following result.
\begin{thm}\label{thm:sep-alpha}
For any $\alpha > 1$, $\beta > 2\sqrt{3}\alpha$, and $\delta < 1/2$, if $\gamma$ is large enough that \[3^{\frac{2\alpha\sqrt{3}}{\beta}}4^{ \frac{1+3\delta}{4\delta}} \gamma^{-1 + \frac{2\alpha\sqrt{3} }{\beta} } < 1\] then for $n$ sufficiently large the probability that a configuration drawn from $\pi_\alpha$ is not $(\beta,\delta)$-separated is exponentially small:
\[ \pi_{\alpha}(\Omega_\alpha\setminus \mS_{\beta,\delta})< \zeta^{\sqrt{n}}\]
where $\zeta < 1$.
\end{thm}
\begin{proof}
This result follows from the previous theorem. Let $\zeta < 1$ be a constant such that for any $\P$ with $|\P| < \alpha p_{min}$, $\pi_\P(\Omega_\P\setminus \mS_{\beta,\delta}) < \zeta^{\sqrt{n}}$. We then see that
\begin{align*}
\pi_\alpha(\Omega_\alpha \setminus \mS_{\beta,\delta} ) = \sum_{\P: |\P| < \alpha p_{min}}  \pi_\alpha(\Omega_\P \setminus \mS_{\beta,\delta})
&= \sum_{\P: |\P| < \alpha p_{min}} \pi_\alpha(\Omega_\P) \pis(\Omega_\P \setminus \mS_{\beta,\delta} )
\\&\leq \sum_{\P: |\P| < \alpha p_{min}} \pi_\alpha(\Omega_\P) \zeta^{\sqrt{n}}
= \zeta^{\sqrt{n}}.
\end{align*}
\end{proof}

\begin{cor}\label{cor:comp+sep-lambdagamma}
For Markov chain $\M$ with parameters $\lambda$ and $\gamma$ satisfying
$\lambda > 1$, $\gamma > 4^{5/4} \sim 5.66 $, and $\lambda \gamma > 2(2+\sqrt{2})e^{\threecsep} \sim \lgsep$, there exist constants $\beta$ and $\delta$ such that for large enough $n$, $\M$ provably accomplishes $(\beta,\delta)$-separation at stationarity with high probability.
\end{cor}
\begin{proof}
	We show there exist constants $\beta$ and $\delta$, depending on $\lambda$ and $\gamma$, such that a configuration drawn from the stationary distribution of $\M$ is $(\beta,\delta)$-separated with probability $1 - \zeta^{\sqrt{n}}$ for some constant $\zeta < 1$.

	Given $\lambda$ and $\gamma$ satisfying the conditions of the theorem, by	Corollary~\ref{cor:comp-sep-lambdagamma} there is a constant $\alpha>1$ and a $\zeta_1 < 1$ such that the stationary probability that the particles are not $\alpha$-compressed is at most $\zeta_1^{\sqrt{n}}$.  If the particles are $\alpha$-compressed, by
	Theorem~\ref{thm:sep-alpha} if $\beta$, $\delta$, and $\gamma$ satisfy
	\[3^{\frac{2\alpha\sqrt{3}}{\beta}}4^{ \frac{1+3\delta}{4\delta}} \gamma^{-1 + \frac{2\alpha\sqrt{3} }{\beta} } < 1\]
	then the probability that the particles are not $(\beta,\delta)$-separated is at most $\zeta_2^{\sqrt{n}}$ for a constant $\zeta_2 < 1$. For $\gamma > 4^{5/4}$, one can always find a $\delta < 1/2$ such that $\gamma > 4^{(1+3\delta)/4\delta}$. For the value $\alpha$ determined by $\lambda$ and $\gamma$ via Corollay~\ref{cor:comp-sep-lambdagamma} above, one can always find a  constant $\beta > 2\sqrt{3} \alpha$ such that the exponent $2\alpha\sqrt{3}/\beta$ is sufficiently close to zero that the above expression is less than one, as desired.  We conclude the probability that $(\beta,\delta)$-separation occurs at stationarity is at least \[\pi(\mS_{\beta,\delta}) \geq 1 - \zeta_1^{\sqrt{n}} - \zeta_2^{\sqrt{n}}  \geq 1 - \zeta^{\sqrt{n}}\]
	for some $\zeta < 1$ provided $n$ is sufficiently large.
	\end{proof}

\begin{cor}\label{cor:comp+sep-betadelta}
	For any $\beta > 2\sqrt{3}$ and $\delta< 1/2$, there are values of $\lambda$ and $\gamma$ such that for large enough $n$, $\M$ provably accomplishes $(\beta,\delta)$-separation with high probability.
	\end{cor}
\begin{proof}
	Because $\beta > 2\sqrt{3}$ and this inequality is strict, we can always find an $\alpha > 1$ such that $\beta > 2\alpha\sqrt{3}$. For this choice of $\alpha$, by Corollary~\ref{cor:comp-sep-alpha} there exists $\lambda$ and $\gamma$ such that $\M$ with these parameters achieves $\alpha$-compression at stationarity with probability $1-\zeta_1^{\sqrt{n}}$ for some $\zeta_1 < 1$. It also holds that if $\gamma$ is large enough then
	\[3^{\frac{2\alpha\sqrt{3}}{\beta}}4^{ \frac{1+3\delta}{4\delta}} \gamma^{-1 + \frac{2\alpha\sqrt{3} }{\beta} } < 1\]
	and by Theorem~\ref{thm:sep-alpha}, if the particles are $\alpha$-compressed then they are $(\beta,\delta)$-separated with probability $1 - \zeta_2^{\sqrt{n}}$.  We conclude that $(\beta,\delta)$-separation occurs in $\M$'s stationary distribution with probability at least $1 - \zeta_1^{\sqrt{n}} - \zeta_2^{\sqrt{n}}  \geq 1 - \zeta^{\sqrt{n}}$ for some $\zeta < 1$ provided $n$ is sufficiently large.
	\end{proof}
This concludes our proofs that $\M$ accomplishes separation.

\section{\texorpdfstring{Proof of Compression when $\gamma$ is close to one}{Proof of Compression when gamma is close to one}}
\label{sec:compression-small-gamma}

Having shown $\M$ provably achieves separation when $\lambda$ and $\gamma$ are large, we now begin to consider the case when $\gamma$ is close to one; we will ultimately show separation {\it does not} occur.

As above, where we showed compression occurs when $\gamma$ and the product $\lambda\gamma$ are large enough, here we show compression occurs when $\gamma$ is close to one and the product $\lambda (\gamma+1)$ is large enough. When $\lambda > \cbound$ and $\gamma = 1$, our algorithm is exactly the compression algorithm of~\cite{Cannon2016}, and so compression provably occurs. Here we extend that result to $\gamma$ in a neighborhood about one, showing compression happens whenever $\gamma \in (\gintlower, \gintupper)$ and  $\lambda(\gamma+1) > 2(2+\sqrt{2}) e^{3a}$ where $a = \aint$, a bound that nearly recovers the compression result when $\gamma = 1$. This will be a crucial step towards proving that for small enough $\gamma$ we do not see separation, as our techniques can only be applied to compressed configurations.

To prove compression for $\gamma$ near one, we again will use the cluster expansion.  We can't look at polymers that are contours separating regions of different colored particles, as we did in Section~\ref{sec:compression-large-gamma}, because that cluster expansion doesn't converge for $\gamma$ close to one.
Instead, we will use a different notion of polymers without a physical interpretation. We reach this new notion of a polymer by considering the {\it high-temperature expansion}, which is well-studied for the Ising model (see, e.g., ~\cite{Friedli2018}, Section 3.7.3) and which we introduce and modify as necessary for our setting here.

As before, let $\P$ be the boundary of a connected hole-free configuration $\sigma$ with $n$ total particles. Recall that $|\mathcal{P}|$ is the length of this walk $\P$  surrounding $\sigma$, so in particular edges that are traversed twice are counted twice in $|\P|$.  Let $\Omega_\P \subseteq \Omega$ be the set of valid all particle configurations in $\Omega$ with no holes and boundary $\P$ that have the correct number of particles of each color.
In this section we will also consider particle configurations with an arbitrary number of particles of each color; let $\overline{\Omega_\P}$ denote all hole-free particle configurations with boundary $\P$ and any number of particles of each color. Note that $\Omega_\P \subsetneq \overline{\Omega_\P}$.
 Recall $E_\P$ is all edges of $\Gtri$ with both endpoints on or inside $\P$ (all edges with both endpoints occupied in any $\sigma \in \overline{\Omega_\P}$).

We first describe particle configurations as an Ising model. For a particle configuration $\sigma \in \overline{\Omega_\P}$ and a vertex $i$ of $\Gtri$ that is occupied in $\sigma$, define $\sigma_i \in \{+1, -1\}$ to be $\sigma_i = +1$ if the particle at $i$ in $\sigma$ is of color $c_1$, and $\sigma_i = -1$ if the particle at $i$ in $\sigma$ is of color $c_2$. For $a(\sigma)$ the number of monochromatic edges of $\sigma$ and $h(\sigma)$ the number of heterogeneous edges of $\sigma$, we have $a(\sigma) + h(\sigma) = e(\sigma) = |E_\P|$, and we note that
\[ \sum_{(i,j) \in E_\P} \sigma_i \sigma_j = a(\sigma) - h(\sigma).
\]
In particular, this means
\begin{align*} h(\sigma) & = \frac{h(\sigma)}{2} + \left(\frac{e(\sigma)}{2}- \frac{a(\sigma)}{2}\right)\\& =\frac{e(\sigma)}{2} +\frac{h(\sigma) - a(\sigma)}{2} \\&= \frac{ |E_\P|}{2} - \frac{1}{2} \sum_{(i,j) \in E_\P} \sigma_i \sigma_j. \end{align*}
Using this, we can rewrite the total weight of configurations in $\overline{\Omega_\P}$ as
\[
w(\overline{\Omega_\P}) :=\sum_{\sigma \in \overline{\Omega_\P}} w(\sigma) =  \sum_{\sigma \in \overline{\Omega_\P}} \left(\lambda \gamma\right)^{-p(\sigma)} \gamma^{-h(\sigma)}
= \left(\lambda \gamma\right)^{-|\P|} \gamma^{-|E_\P|/2} \sum_{\sigma \in \overline{\Omega_\P}} \prod_{(i,j) \in E_\P} \gamma^{\sigma_i\sigma_j/2}.
\]
We set $\beta = \ln\sqrt{\gamma}$; for those familiar with statistical physics, this $\beta$ can be seen as playing the role of inverse temperature. We note that $\gamma$ being in a neighborhood about 1 implies $\beta$ is in a neighborhood about 0, and we can write
\[
w(\overline{\Omega_\P}) =
\left(\lambda \gamma\right)^{-|\P|} \gamma^{-|E_\P|/2} \sum_{\sigma \in \overline{\Omega_\P}} \prod_{(i,j) \in E_\P} e^{\beta \sigma_i\sigma_j}.
\]
We now use the high temperature expansion for the last term above. We say a subset $E \subseteq E_\P$ is {\it even} if $E$ contains an even number of edges incident on each particle. As we show in Appendix~\ref{app:ht}, for $\cosh$ the hyperbolic cosine and $\tanh$ the hyperbolic tangent,
\[
\sum_{\sigma \in \overline{\Omega_\P}} \prod_{(i,j) \in E_\P} e^{\beta \sigma_i\sigma_j} = \cosh(\beta)^{|E_\P|} 2^n \sum_{\substack{E \subseteq E_\mathcal{P}:\\even}} \left(\tanh \beta\right)^{|E|}.
\]
This sum is no longer over objects in one-to-one correspondence with particle configurations. Note that because $\beta = \ln\sqrt{\gamma}$,
\begin{align*}
\cosh \beta &= \frac{e^{2\beta} + 1}{2e^{\beta}} = \frac{\gamma+1}{2\sqrt{\gamma}}
\\\tanh\beta &=  \frac{e^{2\beta} - 1}{e^{2\beta} +1} = \frac{\gamma - 1}{\gamma+1}.
\end{align*}
We conclude that, using the above and $|E_\P| = 3n - 3 - |\P|$,
\begin{align}
w(\overline{\Omega_\P})
& = \left(\lambda \gamma\right)^{-|\P|} \gamma^{-|E_\P|/2} \cosh(\beta)^{|E_\P|} 2^n \sum_{\substack{E \subseteq E_\mathcal{P}:\\even}} \left(\tanh \beta\right)^{|E|}\nonumber
\\&= \left(\lambda \gamma\right)^{-|\P|} \left(\frac{\gamma+ 1}{2\gamma}\right)^{|E_\P|} 2^n \sum_{\substack{E \subseteq E_\mathcal{P}:\\even}} \left(\frac{\gamma-1}{\gamma+1}\right)^{|E|} \nonumber
\\& = \left(\frac{\lambda (\gamma+1)}{2}\right)^{-|\P|} \left(\frac{\gamma+1}{2\gamma}\right)^{3n-3} 2^n \sum_{\substack{E \subseteq E_\mathcal{P}:\\even}} \left(\frac{\gamma-1}{\gamma+1}\right)^{|E|}. \label{eqn:high-temp-rep}
\end{align}
We define the {\it high-temperature partition function} to be
\begin{align}\label{eqn:HT-e}
\Xi_\P^{HT} = \sum_{\substack{E\subseteq E_\P\\ even}} \left(\frac{\gamma-1}{\gamma+1}\right)^{|E|}
\end{align}
This is what we will write as a polymer model and analyze using the cluster expansion. For simplicity, we let $z = \frac{\gamma-1}{\gamma+1}$, and note that because we are interested in $\gamma$ in some neighborhood about 1 then we wish to analyze $z$ in some neighborhood of 0. In particular, we assume $|z| < 1$, true for all $\gamma >0$.

We note that any even subgraph $E \subseteq E_\P$ can be divided into its connected components.
Define a {\it high temperature polymer} in this model to be a subset of edges $C \subseteq E_\P$ that is both connected and even. Define the polymer weight of such a connected even edge set as  $w(C) = z^{|C|} $. We say two connected even edge sets (high temperature polymers) are {\it compatible} if $C_1 \cup C_2$ is not connected. Note this is not the same notion of compatibility that was used for loop polymers earlier; here, if $C_1$ and $C_2$ are disjoint but an edge from $C_1$ shares an endpoint with an edge from $C_2$, then $C_1$ and $C_2$ are considered incompatible.

There is a one-to-one correspondence between even $E \subseteq E_\P$ and pairwise compatible collections of connected even edge sets $C_1, C_2,...,C_m$, where $E$ is associated to its maximal connected components. Note that the weight given to $E$ in $\Xi_\P^{HT}$ is $z^{|E|}  = \prod_{i = 1}^{m} z^{|C_i|}$, where the $C_i$ are the maximal connected components of $C$.
In this way we can write $\Xi_\P^{HT}$ as a polymer partition function.  Let ${\mathcal{C}_\P}$ be a set whose elements are the connected even edges sets $C \subseteq E_\P$. We say $\mathcal{C}' \subseteq \mathcal{C}_\P$ is {\it compatible} if all even edge sets in $\mathcal{C}'$ are pairwise compatible. We can then write
\begin{align*}
\Xi_\P^{HT} = \sum_{\substack{\mathcal{C}' \subseteq \mathcal{C}_\P\\ compatible}} \prod_{C \in \overline{C}'} z^{|C|}.
\end{align*}
Now that $\Xi^{HT}_\P$ has been written as a polymer partition function, we can use the cluster expansion to analyze it, similar to our work on the loop polymer partition function above.

First, we will check that the cluster expansion converges by verifying the hypothesis of Theorem~\ref{thm:bdry-volume}, which implies the hypothesis of Theorem~\ref{thm:cluster-converge} also holds.
The following lemma will be necessary.
\begin{lem}\label{lem:num-even-edge-sets}
	For any vertex $v \in V(\Gtri)$, the number of connected even edge sets with $k$ edges that contain edge $e$ is at most $5^{k-1}$.
\end{lem}
\begin{proof}
Any connected graph where every vertex has even degree has an Eulerian cycle. Such a cycle is a walk in $\Gtri$ starting from a given edge $e$ where there are at most five edges to choose from at each step (because an Eulerian walk visits an edge at most once and $\Gtri$ has degree six). Thus the number of such walks  with $k$ total edges is at most $5^{k-1}$.
\end{proof}

Now, we will use this lemma to verify the hypothesis of Theorem~\ref{thm:bdry-volume}. Let $\mathcal{C}$ be a set whose elements are all connected even edge sets in $\Gtri$; they are not restricted to lie within $\P$.
Note that for a given $C \in \mC$, if $C'$ is incompatible with $C$ then $C'$ must contain an edge that shares an endpoint with an edge of $C$. In fact, $C'$ must contain two distinct edges incident on the common endpoint with $C$ because $C'$ is even. As the number of vertices of $C$ is at most the number of edges of $C$, and picking any five edges incident on each vertex of $C$ to form $[C]$ guarantees $C'$ incompatible with $C$ must contain an edge of $C'$, we see that $|[C]| \leq 5|C|$.

\begin{lem}\label{lem:suff-cond-ht}
	For $z$ satisfying $|z|<1/80 = \zint$ and $a = \aint$, for any edge $e \in E_\P$,
	\begin{align*}
	\sum_{\substack{C \in \mathcal{C} \\ e \in C}} |z|^{|C|} e^{a|[C]|}
	\leq a
	\end{align*}
\end{lem}
\begin{proof}
	Let $n_k$ be the number of connected even edge sets containing a fixed given edge. In general, $n_k \leq 5^{k-1}$ by the previous lemma, but we can explicitly calculate $n_k$ for small $k$ to improve our bounds.
Noting that any even edge set must have at least three edges, we see that $n_1 = n_2 = 0$, $n_3 = 2$, $n_4 = 4$, and $n_5 = 10$. It follows that
\begin{align*}
\sum_{\substack{C \in \mathcal{C} \\e \in C}} |z|^{|C|} e^{a|[C]|}
&\leq  \sum_{\substack{C \in \mathcal{C} \\e \in C}} |z|^{|C|} e^{5a|C|}
\\&\leq \sum_{k = 3}^{\infty} n_k |z|^{k} e^{5ak}
\\& \leq 2\left(|z|e^{5a}\right)^3 + 4\left(|z|e^{5a}\right)^4 + 10\left(|z|e^{5a}\right)^5 + \frac{1}{5} \cdot \left( \frac{ (5|z|e^{5a})^6}{1-5|z|e^{5a}} \right)
\end{align*}
The last inequality holds whenever $5|z| e^{5a} < 1$. We want this expression to be less than $a$.
We choose $a = \aint$ and calculate that for $|z| = \zint$ and $a = \aint$, the necessary condition is satisfied.  As the expression above decreases as $|z|$ approaches 0, the necessary condition is satisfied whenever $|z| < \zint$. We note that, indeed, $|z| < \zint$ implies $5|z|e^{5a} <1$, so our assumption earlier about the convergence of the geometric series was valid.
\end{proof}

We chose to only consider $|z| < \zint$ in this lemma because, as we will see in the next section, we are only able to show integration when $|z| < 0.01265 \ldots \sim \zint$.  We then chose the smallest possible~$a$ --- rounded up to a power of ten --- for which the lemma holds.  This lemma can in fact be shown to hold for a much larger range of $z$: when $|z| < 0.1$, it holds for $a = 0.02$.

Note $|z| < 1/80 = \zint$ occurs whenever $\gamma \in (\gintlower, \gintupper)$. 
We can now use Theorem~\ref{thm:bdry-volume} to rewrite the logarithm of the high-temperature partition function in terms of boundary and surface terms. Note that any connected even edge set that contains both an edge in $E_\P$ and an edge not in $E_\P$ must contain an edge with one endpoint in $\P$ and one endpoint outside $\P$.  Invoking counting results from \cite{Cannon2016}, the number of such edges is $2|\P| + 6$.

\begin{thm}\label{thm:bdry-volume-ht}
	Let $a = \aint$ and $\gamma \in (\gintlower, \gintupper)$.  There is a constant $\psi \in [-a,a]$ such that the high temperature polymer partition function satisfies
	\begin{align*}
	(3n-3) \psi - 3a|\P| -6a \leq \ln \Xi_\P^{HT} \leq (3n-3) \psi + 3a|\P| +6a.
	\end{align*}
\end{thm}
\begin{proof}
	Recall $\mC$ is an infinite set of polymers that is closed under translation and rotation. For finite set $E_\P \subseteq E(\Gtri)$, $\Xi_\P^{HT}$ is a polymer model over polymers with edges from $E_\P$. Because the hypothesis of Theorem~\ref{thm:bdry-volume} is satisfied (Lemma~\ref{lem:suff-cond-ht}), for $\partial E_\P$ all edges with exactly one endpoint on $\P$, of which there are exactly $2|\P| + 6$, we see that
		\begin{align*}
	\psi|E_\P| - a (2|\P| +6) \leq \ln \Xi_\P^{HT} \leq \psi|E_\P| + a (2|\P| +6).
	\end{align*}
	Substituting $|E_\P| = 3n - 3 - |P|$ and using $\psi \in [-a,a]$ we get the desired conclusion.
\end{proof}

Equivalently,
\[
e^{(3n-3)\psi -3a|\P| -6a}
\leq
\Xi^{HT}_\P \leq e^{(3n-3) \psi + 3a|\P|+6a}.
\]

This means, in particular, that the difference in $\Xi^{HT}_\P$ between configurations with different boundary contours of different lengths can be bounded by an expression that is exponential in the lengths of these contours, with no terms corresponding to the number of particles or edges of the configurations. This is essential to our argument that compression occurs for $\gamma$ in a neighborhood about 1 (for $z$ in a neighborhood about zero).

\begin{cor}\label{cor:bdry-volume-ht}
	Let $a = \aint$ and $\gamma \in (\gintlower, \gintupper)$.  There is a constant $\psi \in [-a,a]$ such that
	\begin{align*}
w(\overline{\Omega_\P}) &\geq 2^n e^{(3n-3)\psi}\left(\frac{\gamma+1}{2\gamma}\right)^{3n-3} e^{-6a} \left(\frac{2e^{-3a}}{\lambda(\gamma+1)}\right)^{|\P|}	\\ w(\overline{\Omega_\P}) &\leq  2^n e^{(3n-3)\psi}\left(\frac{\gamma+1}{2\gamma}\right)^{3n-3} e^{6a} \left(\frac{2e^{3a}}{\lambda(\gamma+1)}\right)^{|\P|}
	\end{align*}
	\end{cor}
\begin{proof}
	This follows from Theorem~\ref{thm:bdry-volume-ht} and Equation~\ref{eqn:high-temp-rep} relating $w(\overline{\Omega_\P})$ to the high temperature expansion $\Xi^{HT}_\P$.
\end{proof}

\subsection{\texorpdfstring{Bounding the Partition Function of $\pi$}{Bounding the Partition Function of pi}}

We now use the lower bound of Theorem~\ref{thm:bdry-volume-ht} to get a lower bound on $Z$, the partition function of stationary distribution $\pi$.

\begin{lem}\label{lem:Z-ht} 	Let $a = \aint$ and $\gamma \in (\gintlower, \gintupper)$.  There is a constant $\psi \in [-a,a]$ such that the partition function $Z$ satisfies
	\begin{align*}
	Z \geq \frac{1}{n} e^{-6a} 2^n e^{(3n-3)\psi}\left(\frac{\gamma+1}{2\gamma}\right)^{3n-3}  \left(\frac{2e^{-3a}\left(\gintlowerf\right)}{\lambda(\gamma+1)}\right)^{p_{min}}
	\end{align*}
	\end{lem}
\begin{proof}
	Our proof is similar to Lemma~\ref{lem:Zbound}. We use that $Z \geq w(\Omega_{\P_{min}})$. We note there is a map from $\overline{\Omega_{\P_{min}}}$ to $\Omega_{\P_{min}}$ that complements the color of some number of particles (in a fixed left-to-right, top-to-bottom order) until there are the correct number of particles of each color. Each $\sigma \in \Omega_{\P_{min}}$ is the image of at most $n$ configurations in $\overline{\Omega_{\P_{min}}}$.  This map changes the number of heterogeneous edges in a configuration by at  most $p_{min}$. If $\gamma$ is greater than one, the worst case occurs when $p_{min}$ heterogeneous edges are lost in the map, in which case the weight of a configuration decreases by $\gamma^{-p_{min}} \geq (\gintupper)^{-p_{min}} = (\gintlower)^{p_{min}}$. If $\gamma < 1$, the worst case is when $p_{min}$ heterogeneous edges are added in the map, and the weight of a configuration decreases by $\gamma^{p_{min}} \geq (\gintlower)^{p_{min}}$. We see that
	\begin{align*} Z = w(\Omega_{\P_{min}}) \geq \frac{1}{n} w(\overline{\Omega_{\P_{min}}}) (\gintlower)^{p_{min}}.
	\end{align*}
	Using the lower bound in Corollary~\ref{cor:bdry-volume-ht} for $\mathcal{P} = \P_{min}$, we get the claimed bound.
\end{proof}

\subsection{Achieving Compression}

We now show compression occurs whenever $\gamma \in (\gintlower, \gintupper)$ and $\lambda$ is large enough.  In particular, whenever $\lambda (\gamma+1) > 2(2+\sqrt{2})e^{3a}$ for $a = \aint$, there exists a constant $\alpha$ such that $\alpha$-compression occurs.

We now state and prove our most general condition that implies compression occurs, and then follow this theorem with corollaries explaining what this general condition means for various values of $\alpha$, $\lambda$, and $\gamma$.

\begin{thm} \label{thm:comp-int}
	Consider algorithm $\M$ when there are $n$ total particles of two different colors.  For $a = \aint$, when constants $\alpha>1$, $\lambda>1$, and $\gamma \in (\gintlower, \gintupper)$ satisfy
	\begin{align} \label{eqn:comp-int-hyp}
\frac{2(2+\sqrt{2})  e^{3a} }{\lambda (\gamma+1)} \left(\frac{\lambda (\gamma+1)}{2e^{-3a} \left(\gintlowerf\right)}\right) ^{1/\alpha}< 1
\end{align}
	when $n$ is sufficiently large then at stationarity for $\M$ with parameters $\lambda $ and $\gamma$, the probability that the configuration is not $\alpha$-compressed is exponentially small:
	\[ \pi(S_\alpha) < \zeta^{\sqrt{n}}
	\]
\end{thm}
\begin{proof}
	This proof builds on the proof of compression in~\cite{Cannon2016} and is very similar to the proof of Theorem~\ref{thm:comp-sep}.
	Pick $\nu$ such that $\nu > 2+\sqrt{2}$ but also Equation~\ref{eqn:comp-int-hyp} is satisfied with $\nu$ replacing $2+\sqrt{2}$:
		\begin{align} \label{eqn:comp-int-hyp-nu}
		\frac{2 \nu e^{3a} }{\lambda (\gamma+1)} \left(\frac{\lambda (\gamma+1)}{2e^{-3a} \left(\gintlowerf\right)}\right) ^{1/\alpha}< 1
	\end{align}
	Picking such a $\nu$ is possible because Equation~\ref{eqn:comp-int-hyp} holds and is a strict inequality.
	We begin by expressing $\pi(S_\alpha)$ in terms of the contributions from configurations with given perimeters, where $S_\alpha$ is all particle configurations with perimeter larger than $\alpha p_{min}$.
	\[ \pi(S_\alpha)
	= \frac{w(S_\alpha)}{Z}
	= \frac{\sum_{k = \lceil \alpha p_{min} \rceil}^{p_{max}} \sum_{\mathcal{P}: |\P| = k} w(\Omega_\P)}{Z}.
	\]
	We can use Lemma~\ref{lem:Z-ht} to get a lower bound on the denominator and, noting that $w(\Omega_\P) < w(\overline{\Omega_\P})$, we use Corollary~\ref{cor:bdry-volume-ht} to get an upper bound on the numerator:
	\[
	\pi(S_\alpha) \leq \frac{\sum_{k = \lceil \alpha p_{min} \rceil}^{p_{max}}  \sum_{\mathcal{P}: |\P| = k}
	 2^n e^{(3n-3)\psi}\left(\frac{\gamma+1}{2\gamma}\right)^{3n-3} e^{6a} \left(\frac{2e^{3a}}{\lambda(\gamma+1)}\right)^{k}
}{\frac{1}{n} e^{-6a} 2^n e^{(3n-3)\psi}\left(\frac{\gamma+1}{2\gamma}\right)^{3n-3}  \left(\frac{2e^{-3a}\left(\gintlowerf\right)}{\lambda(\gamma+1)}\right)^{p_{min}}}.
	\]
	Importantly, all volume terms (those involving an exponent of $n$) cancel, and we are left with the expression
	\begin{align*}
	\pi(S_\alpha) &\leq \frac{\sum_{k = \lceil \alpha p_{min} \rceil}^{p_{max}}  \sum_{\mathcal{P}: |\P| = k}
		  e^{6a} \left(\frac{2e^{3a}}{\lambda(\gamma+1)}\right)^{k}
	}{\frac{1}{n} e^{-6a} \left(\frac{2e^{-3a}\left(\gintlowerf\right)}{\lambda(\gamma+1)}\right)^{p_{min}}}
\\&= \sum_{k = \lceil \alpha p_{min} \rceil}^{p_{max}}  n e^{12a} \sum_{\mathcal{P}: |\P| = k}  \left(\frac{2e^{3a}}{\lambda(\gamma+1)}\right)^{k} \left(\frac{\lambda(\gamma+1)}{2e^{-3a}\left(\gintlowerf\right)}\right)^{p_{min}}
	\end{align*}
	Using Lemma~\ref{lem:nu} we can upper bound the number of boundaries $\P$ of length $k$ by $\nu^k$, provided $n$ is sufficiently large.
	\begin{align*}
	\pi(S_\alpha) &\leq  \sum_{k = \lceil \alpha p_{min} \rceil}^{p_{max}} n e^{12a} \nu^k  \left(\frac{2e^{3a}}{\lambda(\gamma+1)}\right)^{k} \left(\frac{\lambda(\gamma+1)}{2e^{-3a}\left(\gintlowerf\right)}\right)^{p_{min}}
	\end{align*}
Next, we note that the values of $k$ considered in the sum satisfy $k \geq \alpha p_{min} > \sqrt{n}$. In particular, this means that $p_{min} \leq k/\alpha$ and $n \leq k^2$.
 One can check that $2e^{-3a} (\gintlower) < 1.96 < \lambda (\gamma+1)$ when $\lambda > 1$ and $\gamma > \gintlower$, the term above with exponent $p_{min}$ is greater than one and we see that
 \begin{align*}
 	\pi(S_\alpha) &\leq
 	\sum_{k = \lceil \alpha p_{min} \rceil}^{p_{max}}  k^2 e^{12a}  \nu^k \left(\frac{2e^{3a}}{\lambda(\gamma+1)}\right)^{k} \left(\frac{\lambda(\gamma+1)}{2e^{-3a}\left(\gintlowerf\right)}\right)^{k/\alpha}
 	\\&\leq
 	\sum_{k = \lceil \alpha p_{min} \rceil}^{p_{max}} k^2 e^{12a}  \left( \frac{2\nu e^{3a}}{ \lambda(\gamma+1)}  \left(\frac{\lambda(\gamma+1)}{2e^{-3a}\left(\gintlowerf\right)}\right)^{1/\alpha}\right)^k.
\end{align*}

	We chose $\nu$ such that the term in parentheses above is less than one (Equation~\ref{eqn:comp-int-hyp-nu}) so the above sum is convergent.
	For sufficiently large $n$ (i.e., sufficiently large $p_{min}$), this upper bound on $\pi(S_\alpha)$ is exponentially small. That is, for $n$ sufficiently large, there exists a constant $\zeta$ such that $\pi(S_\alpha) < \zeta^{\lceil\alpha p_{min}\rceil} < \zeta^{\sqrt{n}}$.  This proves the theorem.
\end{proof}

We note for any fixed values of $\alpha$ and $\gamma$, there exists a value of $\lambda$ large enough so that the hypothesis of the theorem is satisfied.  The larger $\alpha$ is, the smaller $\lambda$ and $\gamma$ can be and still have the theorem apply.

\begin{cor}\label{cor:comp-int-lambdagamma}
	Whenever $\lambda >1$ and $\gamma \in (\gintlower, \gintupper)$ such that for $a = \aint$, if $\lambda$ and $\gamma$ satisfy
	${\lambda(\gamma+1) > 2(2+\sqrt{2})e^{3a} \sim \lgint}$,
	there is an $\alpha$ such that for sufficiently large $n$, $\alpha$-compression occurs with high probability at stationarity.
\end{cor}
\begin{proof}
	Rearranging Equation~\ref{eqn:comp-int-hyp}, we see that for a given choice of $\lambda$ and $\gamma$, the hypothesis of Theorem~\ref{thm:comp-int} holds for any $\alpha$ satisfying
	\[ \ln{\left(\frac{\lambda(\gamma+1)}{2(2+\sqrt{2})e^{3a}}\right)} \alpha >
	\ln{\left(\frac{\lambda (\gamma+1)}{2e^{-3a}\left(\gintlowerf\right)}\right)}.\]
	When $\lambda(\gamma+1) > 2(2+\sqrt{2})e^{3a}$, the logarithm on the left hand side above is positive; because $\lambda > 1$ and $\gamma > \gintlower$, the logarithm on the right hand side is also positive. We see Equation~\ref{eqn:comp-int-hyp} holds for $\alpha$ satisfying
	\[\alpha >
	\frac{\ln{\left(\frac{\lambda (\gamma+1)}{2e^{-3a}(\gintlowerf)}\right)}}{\ln{\left(\frac{\lambda(\gamma+1)}{2(2+\sqrt{2})e^{3a}}\right)}}\]
	We can always find an $\alpha>1$ satisfying this equation. Thus, for any $\lambda$ and $\gamma$ satisfying $\lambda(\gamma+1) > 2(2+\sqrt{2})e^{3a}$, there is a constant $\alpha$ such that Equation~\ref{eqn:comp-sep-hyp} is satisfied and by Theorem~\ref{thm:comp-int} $\alpha$-compression occurs with high probability.
\end{proof}

\noindent We note that this proof nearly recovers our compression result from~\cite{Cannon2016} when $\gamma = 1$; that result says compression occurs whenever $\lambda > \cbound$, while our result occurs for $\lambda > (\cbound) \cdot e^{\threea}$, where $e^{\threea} \sim \ethreea$. This constant $e^{\threea}$ can be pushed arbitrarily close to one as we restrict the range of $\gamma$ for which our proofs apply to be smaller and smaller.

\begin{cor}\label{cor:comp-int-alpha}
	For any constants $\alpha > 1$ and $\gamma \in (\gintlower, \gintupper)$, there exists $\lambda$ such that $\M$, with parameters $\gamma$ and $\lambda$, exhibits $\alpha$-compression at stationarity. In particular, this occurs when $\lambda$ is large enough so that
	\begin{align*}
	\lambda^{\alpha - 1 } > \frac{2^{\alpha - 1} (2+\sqrt{2})^\alpha e^{3a(\alpha - 1)}}{ (\gamma+1)^{\alpha-1} \left(\gintlowerf\right)}
	\end{align*}
\end{cor}
\begin{proof}
	When $\lambda$ and $\gamma$ satisfy the given conditions, rearranging terms shows the hypothesis of Theorem~\ref{thm:comp-sep} are satisfied and so $\alpha$-compression occurs with high probability.
\end{proof}

\section{\texorpdfstring{Proof of Integration when $\gamma$ is close to one, $\lambda$ sufficiently large}{Proof of Integration when gamma is close to one, lambda sufficiently large}}
\label{sec:integration}

We show, for $\gamma$ close to 1 and $\lambda$ large enough such that $\alpha$-compression occurs, there exists $\delta$  and $\beta > 0$ such that when $n$ is sufficiently large $(\beta, \delta)$-separation doesn't occur.

Recall (Definition~\ref{defn:sep}) that a two-color particle configuration is  $(\beta, \delta)$-separated if there is a set $R$ of particles with small boundary ( $|bd_{int}R|< \beta\sqrt{n}$, where $bd_{int}(R)$ is all edges of the configuration with exactly one endpoint in $R$) such that the density of particles of color $c_1$ is at least $1-\delta$ in $R$ and at most $\delta$ outside of $R$. This definition makes sense for $\beta > 0$ and $\delta \in (0,1/2)$.
In this section we assume for the sake of simplicity that there are $n$ total particles with $n/2$ of each color, though we expect our results to generalize with little effort whenever there are a constant fraction of particles of each color. (If there is not a constant fraction of particles of each color, the configuration is always $(\beta,\delta)$-separated for any $\beta$ and $\delta$ with $R$ the set of all particles).

Consider any configuration $\sigma$ that is $\alpha$-compressed and  $(\beta, \delta)$-separated, and let $R$ be a set witnessing this separation: $|bd_{int}(R)| < \beta \sqrt{n}$, at most a $\delta$ fraction of particles in $R$ are color $c_2$, and at most a $\delta$ fraction of the particles not in $R$ are color $c_2$. Recall $\overline{R}$ is all particles not in $R$. We will use the following lemma bounding the size of $R$.
\begin{lem}\label{lem:Rsize}
	For a set $R$ witnessing $(\beta,\delta)$-separation in a particle configuration $\sigma$ of $n$ particles with $n/2$ of each color,
	\[ \frac{1-2\delta}{1-\delta} \cdot \frac{n}{2} < |R| \leq \frac{1}{1-\delta}\cdot \frac{n}{2} \]
	\end{lem}
\begin{proof}
	Suppose $|R| > \frac{n}{2-2\delta}$. By the definition of $(\beta,\delta)$-separation, at most a $\delta$ fraction of these particles can be of color $c_2$, which means at most $\delta n /(2-2\delta)$ of these particles are of color $c_2$.  The remaining particles of $R$ must be of color $c_1$, so the number  of particles of color $c_1$ in $R$ satisfies
	\[|R| - \frac{\delta n}{2-2\delta} > \frac{n}{2-2\delta} - \frac{\delta n}{2-2\delta} = \frac{(1-\delta)n}{2(1-\delta)} = \frac{n}{2} \]
	This says $R$ has strictly more than $n/2$ particles of color $c_1$,  a contradiction as there are only $n/2$ particles of color $c_2$ in the entire configuration.  We conclude it must be that \[|R|\leq \frac{1}{1-\delta}\cdot \frac{n}{2}.\]  By the symmetry between $R$ and $\overline{R}$ in the definition of separation, $|\overline{R}| \leq \frac{1}{1-\delta}\cdot \frac{n}{2}$ and so \[|R| \geq n -\frac{1}{1-\delta}\cdot \frac{n}{2} = \frac{1-2\delta}{1-\delta}\cdot \frac{n}{2}. \]
	\end{proof}

Recall $Bd_{out}(R)$ is all edges of $\Gtri$ with one endpoint in $R$ and the other endpoint unoccupied and $bd(R) = bd_{out}(R) \cup bd_{int}(R)$.
We furthermore note that because we assume our configuration of interest is $\alpha$-compressed,  $|bd R| =|bd_{out} R| +  |bd_{int} R |  \leq \alpha p_{min} + \beta \sqrt{n} \leq (2\sqrt{3} \alpha +\beta) \sqrt{n}$, where we used that $p_{min} \leq 2\sqrt{3} \sqrt{n}$ (Lemma~\ref{lem:pmin}).

To show that $(\beta,\delta)$-separation doesn't occur, we will give polynomially many events such that if $(\beta,\delta)$-separation occurs then at least one of these events occurs; we then show each of these event has a superpolynomially small probability of occuring.
To this end,
 we consider a partition of the lattice $\Gtri$ into diamonds with side length $n^c$ for some $c < 1/4$, where $c$ is chosen so that $n^c$ is an integer; see Figure~\ref{fig:diamonds} for an example of a partition of $\Gtri$ into diamonds with side length six.
 Each diamond in this partition contains $n^{2c}$ vertices of $\Gtri$.
 The events we consider will be that one diamond in the partition is fully occupied by particles and at most $\delta' < 1/2$ of these particles are of color $c_2$.

\begin{figure}\centering
\includegraphics[scale = 0.7]{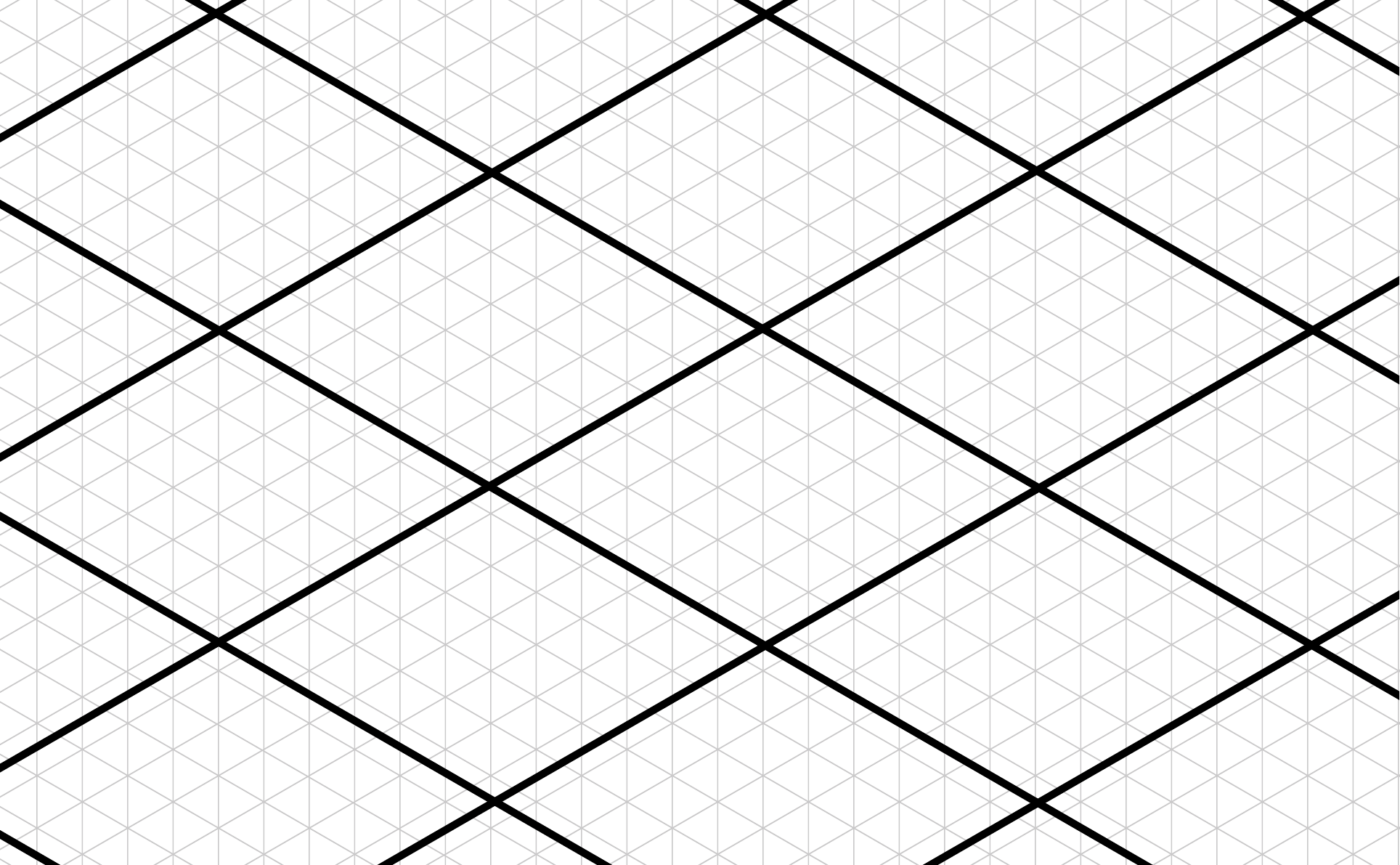}
	\caption{A partition of the lattice $\Gtri$ into diamonds with side length 6 and $6^2 = 36$ total vertices each. }\label{fig:diamonds}
	\end{figure}

\begin{lem}\label{lem:diamondexists}Let $\sigma$ be an $\alpha$-compressed particle configuration.
	If $\sigma$ is $(\beta, \delta)$-separated, then for any $\delta' > \delta/(1-2\delta)$, there exists a diamond in our partition that is fully occupied by particles and has at most $\delta' n^{2c}$ particles of color $c_2$.
	\end{lem}
\begin{proof}
	Let $R$ witness the $(\beta,\delta)$-separation of $\sigma$.
Because each diamond in our partition contains $n^{2c}$ vertices of $\Gtri$, using Lemma~\ref{lem:Rsize} the total number of diamonds intersecting $R$ is at least $|R| / n^{2c} \geq   \frac{n^{1-2c}(1-2\delta)}{1-\delta}.$
	Meanwhile, we see that at most $|bd(R)| \leq (\beta + 2\sqrt{3} \alpha) \sqrt{n}$ diamonds intersect the boundary of $R$.
The number of diamonds in our partition comprised entirely of particles in $R$ (and thus fully occupied by particles) must be at least  $ n^{1-2c}(1-2\delta) /(2-2\delta) - ( \beta + 2\sqrt{3}\alpha) \sqrt{n}$.
We suppose, for the sake of contradiction, that each of these has at least $\delta' > \delta /(1-2\delta)$ particles of color $c_2$. Then the total number $n_2$  of particles of color $c_2$ in $R$ would be at least
\begin{align*}
n_2 &\geq \delta' n^{2c} \left( \frac{(1-2\delta) n^{1-2c}}{2-2\delta} - ( \beta + 2\sqrt{3}\alpha) \sqrt{n}\right)
\\&=  \left( \delta' \frac{1-2\delta}{2-2\delta}n - \delta' (2\sqrt{3}\alpha + \beta) n^{1/2 + 2c}\right)
\end{align*}
Provided $c < 1/4$, asymptotically the first term above dominates and for $n$ sufficiency large we have that $$n_2 > \frac{\delta'(1-2\delta)}{2-2\delta} n =  \left(\frac{\delta}{1-2\delta} \right)  \frac{1-2\delta}{2-2\delta} n  = \delta \left(\frac{1}{2-2\delta} n \right) \geq \delta |R|.$$
This shows that greater than a $\delta$ fraction of the particles in $R$ have color $c_2$, contradiction that $\sigma$ is $(\beta, \delta)$-separated with cluster $R$.
We conclude that there must exist a diamond in the partition that is fully occupied by particles in $R$ with at most a $\delta'$ fraction of particles of color $c_2$.
\end{proof}

We are going to be interested in those values of $\delta$ which give $\delta` < 1/2$, that is, in $\delta < 1/4$.
We now show that the probability a diamond in the partition that is fully occupied by particles has fewer than a $\delta'$ fraction of its particles of color $c_2$ is very small.

Let $\P$ be an $\alpha$-compressed boundary of a particle configuration. We will look at particles configurations with this boundary, which we call $\Omega_\P$. Recall $\pi_\P$ is the stationary distribution conditioned on boundary $\P$. For $\sigma$ a coloring of the particles inside $\P$, we say that the weight of this configuration (once $\P$ is fixed) is  $w_\P(\sigma) = \gamma^{-h(\sigma)}$. We can then write \[\pi_\P(\sigma) = \frac{w_\P(\sigma)}{Z_\P},\]
where $Z_\P = \sum_{\sigma \in \Omega_\P} w_{\P(\sigma)}$.

For $\mathcal{D}$ any $ n^c \times n^c$ diamond with every vertex occupied by a particle on or inside $\P$, we want to look at the $\sigma\in \Omega_\P$ that have fewer than a $\delta'$ fraction of the particles in $\mathcal{D}$ that are color $c_2$. We want to show that the set of all such $\sigma$ has exponentially small weight.

\newcommand{\D}{\mathcal{D}}
\newcommand{\ovD}{\overline{\mathcal{D}}}

For such $\sigma$, we will break the term $w_\P(\sigma)$ up according to contributions within $\D$, contributions between $\D$ and $\ovD$, and contributions within $\ovD$, where $\ovD$ is all particles on or inside $\P$ not in $\D$.
There are at most $8n^c +6$ edges between $\D$ and $\ovD$, so these edges can contribute at most $\max\{1, \gamma^{-8n^c -6}\}$ and at least $\min\{1, \gamma^{-8n^c -6}\}$ to the weight of a configuration, where which values are achieved in this maximum and minimum depend on whether $\gamma > 1$ or $\gamma < 1$. Instead of looking at contributions from $\D$ or $\ovD$ for particular configurations, we look at the sum of contributions within these regions over many possible configurations.

For any set $\Lambda$ of vertices of $\Gtri$, let $\Omega^l_\Lambda$ be all colorings of vertices in $\Lambda$ with exactly $\ell$ particles assigned color $c_2$; we will consider $\Lambda = \D$ and $\Lambda = \ovD$. For $\sigma \in \Omega_\Lambda^\ell$, we say that $h_\Lambda(\sigma)$ is the number of edges of $\Gtri$ where both endpoints are in $\Lambda$ and are assigned different colors in $\sigma$.  We will consider the partition functions
\begin{align*}
Z_\Lambda^{\ell} = \sum_{\sigma \in \Omega_\Lambda^\ell} \gamma^{-h_\Lambda(\sigma)}.
\end{align*}

The following lemma will play an important role. Note that $\Lambda$ need not be connected or hole-free.
\begin{lem}
	For any $\Lambda \subseteq V(\Gtri)$,
	\[
	\frac{Z_{\Lambda}^{\ell}}{Z_\Lambda^{|\Lambda|/2}} \leq \max\{\gamma^{6(\ell - |\Lambda|/2)},\gamma^{ -6(|\Lambda|/2-\ell)} \}
	\]
	\end{lem}
\begin{proof}
We first note that, because of the symmetry between colors, that $Z^\ell_{\Lambda} = Z^{|\Lambda| - \ell}_{\Lambda}$. Without loss of generality, we assume $\ell \geq |\Lambda|/2$.
We define a multimap from configurations in $\Omega_\Lambda^\ell$ to configurations in $\Omega_\Lambda^{|\Lambda|/2}$. For any $\sigma \in\Omega_\Lambda^\ell$, we can map it to a configuration in  $\Omega_\Lambda^{|\Lambda|/2}$ by choosing $\ell - |\Lambda|/2$ of the $\ell$ particles of color $c_2$ and changing their color to $c_1$. There are $\binom{\ell}{\ell-|\Lambda|/2}$ way to do this, and doing so changes the number of heterogeneous edges within $\Lambda$ by at most $6 \cdot\left(\ell - |\Lambda|/2\right)$.  When doing this, each configuration $\tau \in \Omega_\Lambda^{|\Lambda|/2}$ is obtained from $\binom{|\Lambda|/2}{\ell-|\Lambda|/2}$ different $\sigma \in \Omega_\Lambda^{\ell}$. This implies
\begin{align*}
Z^{\ell}_{\Lambda} \cdot \binom{\ell}{\ell-|\Lambda|/2} \cdot  \min\{\gamma^{6 \cdot\left(\ell - |\Lambda|/2\right)}, \gamma^{-6 \cdot\left(\ell - |\Lambda|/2\right)}\} \leq Z^{|\Lambda|/2}_{\Lambda} \cdot \binom{|\Lambda|/2}{\ell-|\Lambda|/2}.
\end{align*}
Because $\ell \geq |\Lambda|/2$, $\binom{\ell}{\ell-|\Lambda|/2} \geq \binom{|\Lambda|/2}{\ell-|\Lambda|/2}$.  Using this and rearranging terms, we obtain the desired result.
	\end{proof}

\begin{lem}\label{lem:diamondbound}
	Let $\delta' < 1/2$ and $\gamma$ close enough to one such that there exists an $\varepsilon \in (\delta',1/2)$ where
	 \begin{align}\label{eqn:eps-gamma} \left(\frac{\varepsilon}{1-\varepsilon}\right)^{(\varepsilon-\delta')/11} < \gamma <  \left(\frac{1-\varepsilon}{\varepsilon}\right)^{(\varepsilon-\delta')/11} . \end{align}
	 Let $\P$ be a boundary of a particle configuration with $n$ particles and let $\D$ be an $n^c \times n^c$ diamond inside $\P$. The probability that a configuration drawn from $\pi_\P$ has at most $\delta'$ particles of color $c_2$ in $\D$ is at most $\zeta^{n^{2c}}$ for some $\zeta < 1$, provided $n$ is sufficiently large.
\end{lem}
\begin{proof}
	Let $k \leq \delta' n$.
First, we note that $Z_\P$ satisfies
\[
Z_\P \geq Z_\D^{n^{2c}/2} Z_{\ovD}^{(n - n^{2c})/2} \min\{ \gamma^{-8n^c-6}, 1\}
\]
For $S_\D^k$ the set of all configurations in $\Omega_\P$ with exactly $k$ particles of color $c_2$ in $\D$,
\begin{align*}
\pi(S_\D^k) \leq \frac{ Z_\D^{k} Z_{\ovD}^{n/2 - k} \max\{ \gamma^{-8n^c-6}, 1\}}{Z_\D^{n^{2c}/2} Z_{\ovD}^{(n - n^{2c})/2} \min\{ \gamma^{-8n^c-6}, 1\}}
 \leq \max\{\gamma^{-8n^c-6}, \gamma^{8n^c+6} \} \frac{ Z_\D^{k} }{Z_\D^{n^{2c}/2}}
\end{align*}
We note that there are fewer than $3n^{2c}$ edges within $\D$, so
\[ \binom{n^{2c}}{k}  \min\{\gamma^{-3n^{2c}} , 1\} \leq Z_\D^k \leq \binom{n^{2c}}{k} \max\{\gamma^{-3n^{2c}}, 1\} . \]
Using this, we see that
\[
\frac{ Z_\D^{k} }{Z_\D^{n^{2c}/2}} \leq \frac{\binom{n^{2c}}{k} \max\{\gamma^{-3n^{2c}}, 1\}}{\binom{n^{2c}}{n^{2c}/2}  \min\{\gamma^{-3n^{2c}} , 1\}}
 = \max\{\gamma^{-3n^{2c}}, \gamma^{3n^{2c}} \} \prod_{i = k+1}^{n^{2c}/2} \frac{i}{n^{2c}-i+1}
\]
Let $\varepsilon$ be a constant satisfying $\delta' < \varepsilon < 1/2$ and Equation~\ref{eqn:eps-gamma}; by hypothesis we know some such $\varepsilon$ exists. Because $k < \delta' n^{2c}$ and each term in the product above is less than one, we see that
\begin{align*}
\frac{ Z_\D^{k} }{Z_\D^{n^{2c}/2}} &\leq \max\{\gamma^{-3n^{2c}}, \gamma^{3n^{2c}} \} \prod_{i = \delta'n^{2c}}^{\varepsilon n^{2c}} \frac{i}{n^{2c}-i+1} \\&\leq \max\{\gamma^{-3n^{2c}}, \gamma^{3n^{2c}} \} \left( \frac{\varepsilon}{1-\varepsilon}\right)^{(\varepsilon - \delta')n^{2c}}
\end{align*}
We see that
\begin{align*}
\pi(S^k_\D) &\leq \max\{\gamma^{-8n^c-6}, \gamma^{8n^c+6} \} \frac{ Z_\D^{k} }{Z_\D^{n^{2c}/2}} \\&\leq \max\{\gamma^{-8n^c-6}, \gamma^{8n^c+6} \} \max\{\gamma^{-3n^{2c}}, \gamma^{3n^{2c}} \} \left( \frac{\varepsilon}{1-\varepsilon}\right)^{(\varepsilon - \delta')n^{2c}}
\\& \leq \max\{\gamma^{-6}, \gamma^{6} \}   \left(\max\{\gamma^{-11}, \gamma^{11}\} \left( \frac{\varepsilon}{1-\varepsilon}\right)^{\varepsilon - \delta'} \right)^{n^{2c}}
\end{align*}
By our hypothesis, the term in parentheses above is strictly less than one, meaning that for sufficiently large $n$ we have that $\pi(S^k_\D) \leq (\zeta_k)^{n^{2c}}$ for some $\zeta_k < 1$.  We now see that the probability that $\mathcal{D}$ has at most $\delta'$ particles of color $c_2$ is
\[
\pi(\cup_{k = 0}^{\delta'n^{2c}} S_\D^k) \leq \sum_{k = 0}^{\delta'n^{2c}} (\zeta_k)^{n^{2c} } \leq \delta' n^{2c} \max_{k} (\zeta_k)^{n^{2c}}.
\]
For sufficiently large $n$, this is at most $\zeta^{n^{2c}}$ for some $\zeta < 1$.
This concludes our proof.
\end{proof}

In our proof we can choose any $\varepsilon$ between $\delta'$ and $1/2$ and obtain some range of $\gamma$ for which the same result holds.  Instead of finding an optimal value of $\varepsilon$ as a function of $\delta'$ to obtain the largest rang, we note that this optimum value is achieved near $\delta'/2 + 1/4$, halfway between $\delta'$ and $1/2$.  Making this assumption allows us to get some concrete bounds on $\gamma$ and $\delta'$, as we do in the corollaries below.  First, we show this result implies the absence of separation.

\begin{thm}\label{thm:int}
	Let $\P$ be any $\alpha$-compressed boundary.
	Let $\delta < 1/4$ and $\gamma$ close enough to one such that there exists an $\varepsilon \in (\delta/(1-2\delta),1/2)$ where
	\begin{align}\label{eqn:eps-gamma-delta} \left(\frac{\varepsilon}{1-\varepsilon}\right)^{(\varepsilon-\delta/(1-2\delta))/11} < \gamma <  \left(\frac{1-\varepsilon}{\varepsilon}\right)^{(\varepsilon-\delta/(1-2\delta))/11} . \end{align}
	For any $\beta$, the probability that a particle configuration drawn at random from $\pi_\P$ is $(\beta, \delta)$-separated is at most $\overline \zeta^{\sqrt{n}}$ for some constant $\overline\zeta < 1$.
\end{thm}
\begin{proof}
	Pick a $\delta' < \delta/(1-2\delta)$ such that  $\delta' < 1/2$, possible because $\delta < 1/4$, and $\delta'$ satisfies
		 \begin{align} \left(\frac{\varepsilon}{1-\varepsilon}\right)^{(\varepsilon-\delta')/11} < \gamma <  \left(\frac{1-\varepsilon}{\varepsilon}\right)^{(\varepsilon-\delta')/11} . \end{align}
		 This second condition is possible because Equation~\ref{eqn:eps-gamma-delta} is satisfied with strict inequalities.

	If a configuration $\sigma \in \Omega_\P$ is $(\beta, \delta)$-separated, then by Lemma~\ref{lem:diamondexists} for the $ \delta'$ we have chosen there is an $n^{c} \times n^c$ diamond $\D$ that contains at most $\delta' n^{2c}$ particles of color $c_2$, where $c < 1/4$ such that $n^c$ is an integer (for larger and larger $n$ we can pick $c$ closer and closer to $1/4$).   The interior of $\P$ can be covered by at most $n$ diamonds, so by a union bound and Lemma~\ref{lem:diamondbound} the probability that $\sigma$ is $(\beta, \delta)$-clustered is less than $n \zeta^{n^{2c}}$.  There exists a constant $\overline\zeta$ such that for sufficiently large $n$, $n \zeta^{n^{2c}} < \overline\zeta^n$.  This proves the theorem.
\end{proof}

We wish to understand the range of $\gamma$ for which there exists an $\varepsilon$ satisfying~\ref{eqn:eps-gamma-delta}; we focus on the upper bound on $\gamma$, as the lower bound is its reciprocal.  If $\gamma < \left(\frac{1-\varepsilon}{\varepsilon}\right)^{\varepsilon/11}$, then we can always choose a $\delta$ small enough so that this equation is satisfied.
Maximizing this expression exactly with respect to $\epsilon$ is challenging to do exactly, so we note that numerically this is achieved when $\varepsilon \sim  0.217812$, corresponding to an upper bound on $\gamma$ of about $1.02564$, which for simplicity we round down to the more explicit bound of $\gamma < 81/79 \sim 1.02532$.

\begin{cor}\label{cor:comp+int-lambdagamma}
	For Markov chain $\M$ with parameters $\lambda$ and $\gamma$ satisfying
	$\lambda > 1$, $\gamma \in (\gintl,\gintu)$, and $\lambda (\gamma+1) > 2(2+\sqrt{2})e^{\threea} \sim \lgint$, there exist constants $\beta$ and $\delta$ such that for large enough $n$, $\M$ accomplishes $(\beta,\delta)$-separation at stationarity with probability at most $\zeta^{2c}$ where $\zeta < 1$ and $c < 1/4$ where $c$ be made arbitrarily close to $1/4$.
\end{cor}
\begin{proof}
Given~$\lambda$ and $\gamma$ satisfying the conditions of the theorem, by Corollary~\ref{cor:comp-int-lambdagamma} there is a constant $\alpha > 1$ and a $\zeta_1 < 1$ such that the stationary probability that the particles are $\alpha$-compressed is at least $1 - \zeta_1^{\sqrt{n}}$.

Suppose the particles are $\alpha$-compressed.  Because $\gamma < \gintu$, then for $\varepsilon = 0.22$ it holds that $\gamma < \left(\frac{1-\varepsilon}{\varepsilon}\right)^{\varepsilon/11}$. Similarly, as $\gamma > \gintl$, for the same $\varepsilon$ it holds that $\gamma > \left(\frac{\varepsilon}{1-\varepsilon}\right)^{\varepsilon/11}$.  We can find a $\delta$ small enough so that Equation~\ref{eqn:eps-gamma-delta}, the hypothesis of Theorem~\ref{thm:int}, holds. Using this theorem, we conclude that for any $\beta$, the probability that the particles are not $(\beta,\delta)$-separated is at least $1-\zeta_2^{n^{2c}}$, where $c < 1/4$ can be made arbitrarily close to $1/4$ for large $n$ and $\zeta_2< 1$.

We conclude by combining these results that the probability the particles are not $(\beta,\delta)$-separated is at least $1 - \zeta_1^{\sqrt{n}} - \zeta_2^{n^{2c}}$, which for some $\zeta < 1$ is at least $1-\zeta^{n^{2c}}$. This implies the claimed result.
\end{proof}

\begin{cor}\label{cor:comp+int-betadelta}
	For any $\beta > 0 $ and any $\delta < 1/4$, there exist values of $\lambda$ and $\gamma$ such that $\M$ with parameters $\lambda$ and $\gamma$ accomplishes $(\beta, \delta)$-separation with probability at most $\zeta^{2c}$ where $\zeta < 1$ and $c < 1/4$ where $c$ be made arbitrarily close to $1/4$.
	\end{cor}
\begin{proof}
For simplicity, we pick $\gamma = 1$ and $\lambda > (2+\sqrt{2}) e^{\threea}$. By Corollary~\ref{cor:comp-int-lambdagamma}, there exists an $\alpha > 1$ such that $\M$ with these parameters achieves $\alpha$-compression at stationarity with probability at least $1-\zeta_1^{\sqrt{n}}$ where $\zeta_1 < 1$. As $\delta < 1/4$, $\delta/(1-2\delta) < 1/2$, and Equation~\ref{eqn:eps-gamma-delta} is satisfied for any $\varepsilon$ between $\delta/(1-2\delta)$ and $1/2$. Thus Theorem~\ref{thm:int} applies, and we conclude that, if $\P$ is $\alpha$-compressed, the probability that the particles are not $(\beta,\delta)$-separated is at least $1-\zeta_2^{n^{2c}}$, where $c < 1/4$ can be made arbitrarily close to $1/4$ for large $n$ and $\zeta_2< 1$. Combining these results, we conclude the probability the particles are not $(\beta,\delta)$-separated is at least $1 - \zeta_1^{\sqrt{n}} - \zeta_2^{n^{2c}}$, which for some $\zeta < 1$ is at least $1-\zeta^{n^{2c}}$. This implies the claimed result.
\end{proof}


\section{Conclusion} \label{sec:conclude}

We have presented a distributed, stochastic algorithm for separation and integration and have rigorously shown that it achieves its goals for various ranges of the bias parameters $\lambda$ and $\gamma$.
Our proofs critically relied on an analysis of the cluster expansion to show that our heterogeneous systems achieve compression; this required techniques from statistical physics literature that are new to computer science.
We then used a nontrivial modification of the \emph{bridging} technique by Miracle, Pascoe, and Randall~\cite{Miracle2011} to show separation occurs among compressed systems when $\lambda$ and $\gamma$ are large; conversely, we used a careful probabilistic argument to show that integration occurs among compressed systems when $\lambda$ is large but $\gamma$ is close to one.
To conclude, we discuss $(i)$ the tightness of the constants that appear in our proofs, $(ii)$ generalizing our proof techniques for heterogeneous systems with more than two colors, $(iii)$ the difficulties in obtaining bounds on the mixing time of our Markov chain, $(iv)$ the generality of our stochastic approach to distributed algorithms, and $(v)$ the utility of the our proof techniques beyond the present work.

We first discuss the constants that appear in our proofs.
For example, Corollary~\ref{cor:comp+sep-lambdagamma} states that for Markov chain $\M$ with parameters $\lambda$ and $\gamma$ large enough, there are constants $\beta$ and $\delta$ such that $(\beta,\delta)$-separation occurs with high probability at stationarity.
However, when we choose explicit values of $\lambda$ and $\gamma$ and calculate for which values of $\beta$ and $\delta$ our proofs guarantee $(\beta,\delta)$-separation, we see $\beta$ and $\delta$ must be quite large. For example, choosing $\lambda = 4$ and $\gamma = 8$ induces compressed-separated behavior both in practice and in our proofs, but only guarantees $\alpha$-compression for $\alpha > 3.60$. Based on this lower bound on $\alpha$, there is no hope to show $(\beta,\delta)$-separation unless $\beta > 2\alpha\sqrt{3} \sim 12.44$ and $\delta$ satisfies $4^{(1+3\delta)/4\delta} < \gamma$, which for $\gamma = 8$ means $\delta > 1/3$. The relationship between $\beta$ and $\delta$ needs to satisfy further conditions as well; for example, when $\delta = 5/12$ we can only achieve $(\beta,\delta)$-separation for $\beta > 22.53$.
Our other corollaries of the the same flavor exhibit similarly bad (or worse) dependencies between these parameters.
We believe these bad dependencies are a result of our proofs and not inherent in the problem itself.
For example, we see clear separation when $\lambda = 4$ and $\gamma = 4$ in Figure~\ref{fig:progress}, but our proofs only guarantee any degree of separation when $\gamma > 4^{5/4} \sim 5.66$; to achieve the amount of separation seen in this figure, $\gamma$ would need to be yet even larger.
Improving these bounds is one direction for future work.
In some sense, it is not particularly surprising that the relationships we get between constants in our proofs are not as good as we observe them to be in practice.
This is because the cluster expansion is often only convergent in very low temperature regimes, i.e., for $\lambda$ and $\gamma$ far away from any critical points where the emergent behavior of $\M$ changes.

We expect our proofs to generalize in a straightforward way for heterogeneous systems with more than two colors, using the same insights that generalize cluster expansion polymers from the Ising model to the Potts model.
For example, when $\gamma$ and $\lambda$ are both large, recall that in Section~\ref{sec:compression-large-gamma} we defined a \emph{loop polymer} to be the set of heterogeneous edges leaving a face.
We proved that for $2$-heterogeneous systems, loop polymers coming from different faces of the same configuration share no edges.
When there are three colors, however, this is no longer true.
Instead, we can define a more general notion of a polymer in a given configuration to be the union of any loop polymers that share edges; now these polymers never share edges because of their maximality.
This is similar to the notion of a \emph{contour} in Pirogov-Sinai theory (see Chapter 7 of~\cite{Friedli2018}).
The rest of our analysis in Section~\ref{sec:compression-large-gamma} should follow with this new notion of a polymer.
Similar modifications in Sections~\ref{sec:separation}, \ref{sec:compression-small-gamma}, and~\ref{sec:integration} should enable our proofs to be adapted --- with little additional insight but a fairly large amount of technical detail --- to the setting with three or more colors.
We expect these generalized proofs for additional colors would yield significantly worse bounds on the constants in the relationships between $\lambda$, $\gamma$, $\beta$, and $\delta$.
However, simulations have shown that compressed-separated behavior occurs for similar values of $\lambda$ and $\gamma$ whether there are two or three colors.
In future work, we can investigate what values of $\lambda$ and $\gamma$ achieve compression and separation as a function of the number of colors.

As in previous papers using the stochastic approach to develop distributed algorithms for programmable matter, we are unable to give any nontrivial bounds on the mixing time of our Markov chain $\M$.
These difficulties in proving polynomial upper bounds on the mixing time are unsurprising, given similarities between $\M$ and a well-studied open problem in statistical physics about the mixing time of Glauber dynamics of the Ising model on $\mathbb{Z}^2$ with plus boundary conditions starting from the all minus state~\cite{Martinelli2010,Lubetzky2013}.
A detailed description of the similarities between this problem and our Markov chain for compression can be found in~\cite{Cannon2016}; introducing heterogeneous particles only further complicates things.
However, the mixing time may not be the right time bound for characterizing when compression and separation occur.
Simulations show that both compression and separation occur fairly quickly (Figure~\ref{fig:progress}), although the algorithm continues to gradually achieve more compression and separation.
This suggests that the stationary distribution isn't reached until well after we first see some degree of separation and compression.

We believe the stochastic approach to self-organizing particle systems, which we used here to develop a distributed algorithm for separation and integration in programmable matter, is in fact much more broadly applicable.
At a high level, this approach can potentially be applied to any objective that can be described by a global energy function (where the desirable configurations have low energy values), provided changes in energy due to particle movements can be calculated with only local information.
Choosing the correct global energy function is the key; translating the energy function into a Markov chain algorithm and then into a distributed algorithm is, by now, fairly routine (see~\cite{Cannon2016,AndresArroyo2018}).
However, proving that the stationary distribution has our desired properties with high probability remains challenging, requiring application-specific proof techniques.

For separation and integration, one such proof technique was the cluster expansion.
The cluster expansion has recently been used in computer science to develop low-temperature approximation and sampling algorithms, and the related Pirogov-Sinai theory has been used to show slow mixing of certain Markov chains.
However, we used a completely different aspect of the cluster expansion: it can be used to separate partition functions into surface and volume terms.
The cluster expansion and Pirogov-Sinai theory have been widely used in the statistical physics literature for a variety of purposes, and we believe there are many more ways in which a thorough understanding of these methods can provide insights for computer science problems.

\bibliographystyle{plainurl}
\bibliography{ref}

\newpage

\renewcommand \thesection{\Alph{section}}
\setcounter{section}{0}

\section{Appendix}

Here we include the proofs of some of our claims that were omitted from the main body of this paper for conciseness and clarity. These proofs were either deemed too elementary or too similar to previous literature to merit inclusion in the main body of the paper, but we include them here for the sake of completeness.

\subsection{Proof of Lemma~\ref{lem:pmin}}
\label{sec:pf-pmin}

Recall that Lemma~\ref{lem:pmin} states that
	for any $n \geq 1$, there is a connected, hole-free particle configuration of $n$ particles with perimeter at most $2\sqrt{3}\sqrt{n}$. That is, $p_{min}(n) \leq 2\sqrt{3}\sqrt{n}$.

\begin{proof}
	The lemma can easily be verified for $n \leq 6$.
	For $n \geq 7$, we begin with the case where $n = 3\ell^2 + 3\ell + 1$ for some integer $\ell \geq 1$. A regular hexagon with side length $\ell$ can be decomposed into six triangles, each with $\ell(\ell+1)/2$ particles, and a single center vertex, for $3\ell^2 + 3\ell + 1$ total particles; see Figure~\ref{fig:pmina}. Such a hexagon has perimeter $6\ell$. We see that
	\[
	p_{min}(3\ell^2 + 3\ell + 1) \leq 6\ell \leq 2\sqrt{3}\sqrt{3 \ell(\ell+1)} \leq 2 \sqrt{3} \sqrt{n-1} \leq 2 \sqrt{3} \sqrt{n}.
	\]
	Now we consider $n = 3\ell^2 + 3\ell + 1 + k$, for integers $\ell$ and $k$, where $k \in [1,6\ell + 6)$. As $(3\ell^2 + 3 \ell + 1) + 6\ell + 6 = 3(\ell+1)^2 + 3(\ell+1) + 1$, this covers all possible values of $n$.
	\begin{figure}
		\centering
		\begin{subfigure}{0.48\textwidth}
			\centering
			\includegraphics[scale = 0.4, page = 1]{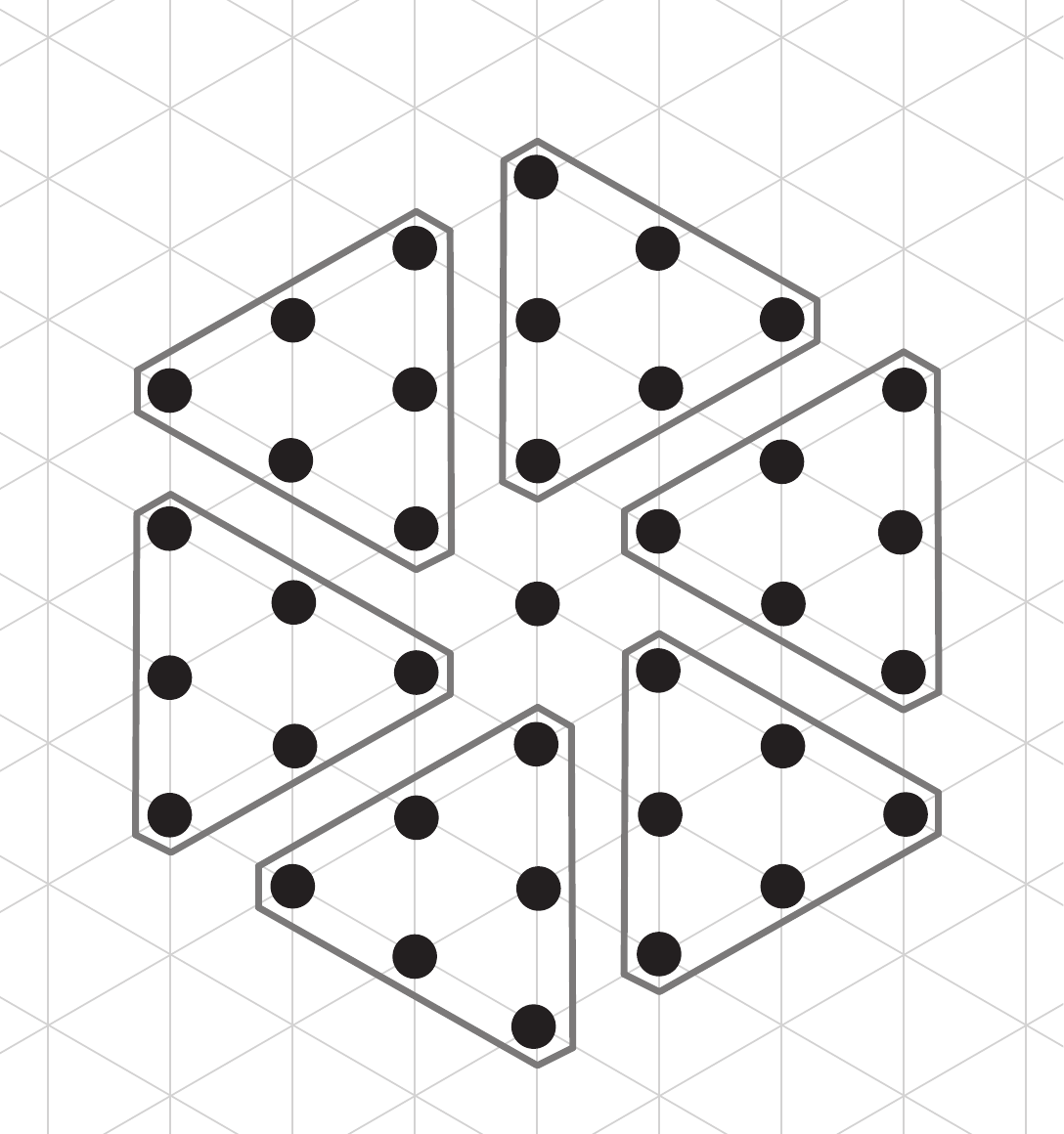}
			\caption{\centering}
			\label{fig:pmina}
		\end{subfigure}%
		\begin{subfigure}{0.48\textwidth}
			\centering
			\includegraphics[scale = 0.4, page = 2]{pmin.pdf}
			\caption{\centering}
			\label{fig:pminb}
		\end{subfigure}%
		\caption{(a) The regular hexagon with side length $\ell = 3$, which has $3 \ell^2 + 3\ell + 1$ total particles. (b) A particle configuration with $n = 3 \ell^2 + 3\ell + 1 + k$ particles for $\ell = 3$ and $k = 6$.  It has perimeter $20 < 2\sqrt{3}\sqrt{n}$.}
	\end{figure}
	We construct a particle configuration on $n = 3 \ell^2 + 3\ell + 1 + k$ particles by first constructing a regular hexagon of side length $\ell$ and then adding the remaining $k$ particles around the outside of this hexagon in a single layer, completing one side before beginning the next; see Figure~\ref{fig:pminb}, where $\ell = 3$ and $k = 6$. For $k \leq \ell$, the perimeter of this configuration is $6\ell + 1$. More generally, the perimeter increases by one when particles begin to be added to a new side of the hexagon, and so for $i = 2,3,4,5,6$, for $(i-1)\ell + (i-2) < k \leq i\ell + (i-1)$ the perimeter of this configuration is $6\ell + i$. We see that (using $i \leq 6$ and $\ell \geq 1$), for any $i = 1,2,3,4,5,6$,
	\begin{align*}
	p_{min}(3\ell^2 + 3\ell + 1 + k) &\leq
	6\ell+i \leq 2\sqrt{3} \sqrt{\left(\sqrt{3}\ell + \frac{i}{2\sqrt{3}}\right)^2}
	= 2\sqrt{3}\sqrt{3\ell^2 + \frac{i^2}{12}+ i}
	\\&\leq 2\sqrt{3}\sqrt{3\ell^2 + 3 + i}
	\\&\leq 2\sqrt{3}\sqrt{3\ell^2 + 3 \ell + 1 + i - 1}
	\\&\leq 2\sqrt{3}\sqrt{3 \ell^2 + 3\ell + 1 + k}
	= 2\sqrt{3}\sqrt{n}.
	\end{align*}
	This concludes our proof.
\end{proof}

\subsection{\texorpdfstring{Detailed Balance Proof that $\pi$ is the Stationary Distribution of $\M$}{Detailed Balance Proof that pi is the Stationary Distribution of M}} \label{app:detailedbalance}

Recall that Lemma~\ref{lem:statdist} states that the stationary distribution of $\M$ is given by $\pi(\sigma) = 0$ if $\sigma$ is disconnected or has holes, and by $\pi(\sigma) = (\lambda\gamma)^{-p(\sigma)} \cdot \gamma^{-h(\sigma)} / Z$ otherwise, where $Z = \sum_{\sigma}(\lambda\gamma)^{-p(\sigma)} \cdot \gamma^{-h(\sigma)}$.
Here, we analyze the necessary cases to verify this with detailed balance.

\begin{proof}
We first verify that $\pi(\sigma) = \lambda^{e(\sigma)} \cdot \gamma^{a(\sigma)} / Z_e$ --- where $e(\sigma)$ is the number of edges of $\sigma$, $a(\sigma)$ is the number of homogeneous edges of $\sigma$, and $Z_e = \sum_{\sigma} \lambda^{e(\sigma)} \cdot \gamma^{a(\sigma)}$ --- is the stationary distribution by detailed balance.
We then show that this form of $\pi$ can be rewritten as in the lemma.

Consider any two connected, hole-free configurations $\sigma, \tau$ that differ by one move of some particle from location $\ell$ in $\sigma$ to a neighboring location $\ell'$ in $\tau$.
By examining $\M$, we see that the probability of transitioning from $\sigma$ to $\tau$ is:
\[M(\sigma,\tau) = \min\left\{1, \lambda^{|N(\ell')| - |N(\ell)|} \cdot \gamma^{|N_i(\ell')| - |N_i(\ell)|}\right\} / 6n.\]
A similar analysis shows:
\[M(\tau, \sigma) = \min\left\{1, \lambda^{|N(\ell)| - |N(\ell')|} \cdot \gamma^{|N_i(\ell)| - |N_i(\ell')|}\right\} / 6n.\]
Without loss of generality, suppose $\lambda^{|N(\ell')| - |N(\ell)|} \cdot \gamma^{|N_i(\ell')| - |N_i(\ell)|} < 1$, meaning $M(\sigma, \tau)$ is this value over $6n$ and $M(\tau, \sigma) = 1/6n$.
Because the only edges that differ in $\sigma$ and $\tau$ are incident to $\ell$ or $\ell'$,
\begin{align*}
\pi(\sigma)M(\sigma, \tau) &= \frac{\lambda^{e(\sigma)} \cdot \gamma^{a(\sigma)}}{Z_e} \cdot \frac{1}{n} \cdot \frac{1}{6} \cdot \lambda^{|N(\ell')| - |N(\ell)|} \cdot \gamma^{|N_i(\ell')| - |N_i(\ell)|}\\
&= \frac{\lambda^{e(\sigma)} \cdot \gamma^{a(\sigma)}}{Z_e} \cdot \frac{1}{n} \cdot \frac{1}{6} \cdot \lambda^{e(\tau) - e(\sigma)} \cdot \gamma^{a(\tau) - a(\sigma)}\\
&= \frac{\lambda^{e(\tau)} \cdot \gamma^{a(\tau)}}{Z_e} \cdot \frac{1}{n} \cdot \frac{1}{6} \cdot 1 = \pi(\tau)M(\tau, \sigma)
\end{align*}
Thus, detailed balance is satisfied for particle moves that are not swaps.

Suppose instead that $\sigma$ and $\tau$ differ by a swap move of particle $P$ with color $c_i$ at location $\ell$ in $\sigma$ and particle $Q$ with color $c_j$ at neighboring location $\ell'$ in $\sigma$.
This move could occur if $P$ or $Q$ is chosen in Step~\ref{alg:step:mstart} of $\M$, so:
\[M(\sigma, \tau) = \min\left\{1, \gamma^{|N_i(\ell') \setminus \{P\}| - |N_i(\ell)| + |N_j(\ell) \setminus \{Q\}| - |N_j(\ell')|}\right\} / 3n.\]
Similarly, because $\tau$ has $P$ at location $\ell'$ and $Q$ at location $\ell$, we have:
\[M(\tau, \sigma) = \min\left\{1, \gamma^{|N_i(\ell) \setminus \{P\}| - |N_i(\ell')| + |N_j(\ell') \setminus \{Q\}| - |N_j(\ell)|}\right\} / 3n.\]
Without loss of generality, suppose that $\gamma^{|N_i(\ell') \setminus \{P\}| - |N_i(\ell)| + |N_j(\ell) \setminus \{Q\}| - |N_j(\ell')|} < 1$, so $M(\sigma, \tau)$ is this value over $3n$ and $M(\tau, \sigma) = 1/3n$.
Then,
\begin{align*}
\pi(\sigma)M(\sigma, \tau) &= \frac{\lambda^{e(\sigma)} \cdot \gamma^{a(\sigma)}}{Z_e} \cdot \frac{2}{n} \cdot \frac{1}{6} \cdot \gamma^{|N_i(\ell') \setminus \{P\}| - |N_i(\ell)| + |N_j(\ell) \setminus \{Q\}| - |N_j(\ell')|}\\
&= \frac{\lambda^{e(\sigma)} \cdot \gamma^{a(\sigma)}}{Z_e} \cdot \frac{2}{n} \cdot \frac{1}{6} \cdot \gamma^{(|N_i(\ell') \setminus \{P\}| + |N_j(\ell) \setminus \{Q\}|) - (|N_i(\ell)| + |N_j(\ell')|)}\\
&= \frac{\lambda^{e(\sigma)} \cdot \gamma^{a(\sigma)}}{Z_e} \cdot \frac{2}{n} \cdot \frac{1}{6} \cdot \gamma^{a(\tau) - a(\sigma)}\\
&= \frac{\lambda^{e(\tau)} \cdot \gamma^{a(\tau)}}{Z_e} \cdot \frac{2}{n} \cdot \frac{1}{6} \cdot 1 = \pi(\tau)M(\tau, \sigma)
\end{align*}
In both cases, detailed balance is satisfied, so we conclude the stationary distribution $\pi$ (which is only non-zero over connected, hole-free configurations) is given by $\pi(\sigma) = \lambda^{e(\sigma)} \cdot \gamma^{a(\sigma)} / Z_e$.

Since every edge of $\sigma$ is either homogeneous or heterogeneous, we have $e(\sigma) = a(\sigma) + h(\sigma)$.
From~\cite{Cannon2016}, we have $e(\sigma) = 3n - p(\sigma) - 3$, where $n$ is the number of particles in the system.
Thus, we can rewrite this unique stationary distribution as follows:
\begin{align*}
\pi(\sigma)
&= \frac{\lambda^{e(\sigma)} \cdot \gamma^{a(\sigma)}}{Z_e}\\
&= \frac{\lambda^{e(\sigma)} \cdot \gamma^{a(\sigma)}}{\sum_{\sigma}\lambda^{e(\sigma)} \cdot \gamma^{a(\sigma)}} \\
&= \frac{(\lambda\gamma)^{-3n + 3} \cdot (\lambda\gamma)^{e(\sigma)} \cdot \gamma^{a(\sigma) - e(\sigma)}}{(\lambda\gamma)^{-3n + 3} \cdot \sum_{\sigma}(\lambda\gamma)^{e(\sigma)} \cdot \gamma^{a(\sigma) - e(\sigma)}} \\
&= \frac{(\lambda\gamma)^{e(\sigma) - 3n + 3} \cdot \gamma^{a(\sigma) - e(\sigma)}}{\sum_{\sigma}(\lambda\gamma)^{e(\sigma) - 3n + 3} \cdot \gamma^{a(\sigma) - e(\sigma)}} \\
&= \frac{(\lambda\gamma)^{-p(\sigma)} \cdot \gamma^{-h(\sigma)}}{\sum_{\sigma}(\lambda\gamma)^{-p(\sigma)} \cdot \gamma^{-h(\sigma)}}. \\
\end{align*}
This concludes our proof.
\end{proof}

\subsection{Lemmas needed for the proof of Theorem~\ref{thm:bdry-volume}} \label{app:bdry-volume}

We first present an intermediate result that follows from the hypothesis of Theorem~\ref{thm:bdry-volume}. This lemma is essentially Theorem 5.4 and Remark 5.5 of~\cite{Friedli2018}, with details included to ensure nothing goes wrong when considering our specific type of infinite $\Gamma$.

Recall our setting.
Our polymers are connected edges sets $\xi \subseteq E(\Gtri)$.  The structural/combinatorial properties a polymer has and the notion of compatibility we will use vary between our different applications of the cluster expansion. For a polymer $\xi$,  $[ \xi ]$ is the the minimal edge set such that if $\xi'$ is not compatible with $\xi$, then $\xi'$ must contain an edge of $[\xi]$; $|[\xi]|$ is the size of this set.
\begin{lem}\label{lem:KP-inf}
	Let $\Gamma$ be an infinite set of polymers $\xi \subseteq E(\Gtri)$ that is closed under translation and rotation. If there is a constant $c$ such that for any edge $e \in E(\Gtri)$,
	\begin{align}
	\sum_{\substack{\xi \in \Gamma : \\ e \in \xi}} w(\xi) e^{c| [\xi]|} \leq c,
	\end{align}
	then for any $\xi_1 \in \Gamma$,
	\begin{align*}
	1 + \sum_{k \geq 2} \sum_{\xi_2 \in \Gamma} .... \sum_{\xi_k \in \Gamma} \frac{1}{(k-1)!} \left|\sum_{\substack{G \subseteq H_{\{\xi_1,\xi_2,...\xi_k\}}:\\ connected, \\ spanning}} (-1)^{|E(G)|} \right| \prod_{i = 2}^k w(\xi_i)  \leq e^{c|[\xi_1]|}.
	\end{align*}
	\end{lem}
\begin{proof}This proof is nearly identical to that of Theorem 5.4 in~\cite{Friedli2018}; we reiterate it here to show that it holds for the infinite sets of polymers $\Gamma$ that we consider as well as the finite sets of polymers considered in~\cite{Friedli2018} (considering infinite sets of polymers is necessary for the proof of the next lemma, which plays a key role in our results).

We will show that for any $\xi_1 \in \Gamma$ and for all $N \geq 2$,
\begin{align} \label{eqn:IH}
1+ \sum_{k=2}^N  \sum_{\xi_2 \in \Gamma} \ldots \sum_{\xi_k \in \Gamma} \frac{1}{(k-1)!} \left|\sum_{\substack{G \subseteq H_{\{\xi_1,\xi_2,...,\xi_k\}} \\ connected \\ spanning}} (-1)^{|E(G)|} \right| \prod_{i = 2}^k w(\xi_i)  \leq e^{c |[\xi_1]|}.
\end{align}

 When we fix a choice of $\xi_1$ and take the limit of the left hand side above as $N$ goes to $\infty$, we obtain the lemma.  We prove Equation~\ref{eqn:IH} by induction on $N$, where our induction hypothesis is that Equation~\ref{eqn:IH} holds for any $\xi_1$ and all smaller values of $N$.

When $N=2$,
$H_{\{\xi_1, \xi_2\} }$ only has connected spanning subgraphs if $H_{\{\xi_1, \xi_2\} }$ itself is connected, that is, if $\xi_2$ is incompatible with $\xi_1$, in which case $H_{\{\xi_1, \xi_2\}} = K_2$ is the complete graph on two vertices.  $K_2$ has exactly one cnnected spanning subgraph, and this spanning subgraph has one edge.  For any choice of $\xi_1 \in \Gamma$, the expression on the left hand side of \ref{eqn:IH} becomes
\begin{align*}
1 +\sum_{\substack{\xi_2 \in \Gamma\\ \xi_1, \xi_2\text{\ incompatible}}}|-1| \ w(\xi_2)
& \leq 1+ \sum_{e \in [\xi_1]} \sum_{\substack{\xi_2 \in \Gamma :\\  e\in \xi_2}} w(\xi_2)
\intertext{By the translation and rotational invariance of $\Gamma$, for any edge $e \in \Gtri$ we can bound this by}
&\leq 1+ |[\xi_1]| \sum_{\substack{\xi_2 \in \Gamma : \\ e \in \xi_2}} w(\xi_2)
\leq 1 + c |[\xi_1]| \leq e^{c|[\xi_1]|}.
\end{align*}
Thus the base case of the induction holds for $N = 2$.

We now suppose that Equation~\ref{eqn:IH} holds for $N$ for all $\xi_1$, and will show that it holds for $N+1$ and any choice of $\xi_1$.
We begin by rewriting the left hand side of \ref{eqn:IH} in a way that will be more useful for us.
\renewcommand{\zeta}{h}
 For $\xi_i, \xi_j \in \Gamma$, let $\zeta(\xi_i, \xi_j)$ be 0 if $\xi_i$ and $\xi_j$ are compatible, and $-1$ if they are incompatible.  Because $H_{\{\xi_1,...,\xi_k\}}$ is the incompatibility graph of $\xi_1$,..., $\xi_k$, we see that
\begin{align}
\label{eqn:zeta}
\sum_{\substack{G \subseteq H_{\{\xi_1,\xi_2,...,\xi_k\}} \\ connected \\ spanning}} (-1)^{|E(G)|} = \sum_{\substack{G \subseteq K_k \\ connected \\ spanning}} \prod_{\substack{(i,j) \in E(G)}} \zeta(\xi_i, \xi_j).
\end{align}
Using this, the statement that we will try to prove is, for $\xi_1 \in \Gamma$ and $N \geq 2$,
\begin{align}\label{eqn:IHz}
1+ \sum_{k=2}^{N+1}  \sum_{\xi_2 \in \Gamma} \ldots \sum_{\xi_k \in \Gamma} \frac{1}{(k-1)!} \left| \sum_{\substack{G \subseteq K_k \\ connected \\ spanning}} \prod_{(i,j) \in E(G)} \zeta(\xi_i, \xi_j) \right| \prod_{i = 2}^k w(\xi_i)  \leq e^{c |[\xi_1]|}.
\end{align}
We assume as our induction hypothesis that the same statement holds for $N$, that is, for any $\xi_1 \in \Gamma$,
\begin{align}\label{eqn:IHN}
1+ \sum_{k=2}^{N}  \sum_{\xi_2 \in \Gamma} \ldots \sum_{\xi_k \in \Gamma} \frac{1}{(k-1)!} \left| \sum_{\substack{G \subseteq K_k \\ connected \\ spanning}} \prod_{(i,j) \in E(G)} \zeta(\xi_i, \xi_j) \right| \prod_{i = 2}^k w(\xi_i)  \leq e^{c |[\xi_1]|}.
\end{align}

The left hand side of Equation~\ref{eqn:IHz} has a summand for each $2 \leq k \leq N+1$.
For each such $k$, consider any connected, spanning subgraph $G \subseteq K_k$ appearing in the sum over all such $G$.
Let $S$ be the set of edges of $G$ that have an endpoint at vertex 1; $S$ is nonempty because $G$ is connected and spans all vertices of $K_k$.
 The graph $G'$, obtained from $G$ by removing vertex 1 and all edges in $S$, splits into components $G_1$, $G_2$,$\ldots$, $G_l$. We can thus see $G$ as obtained by (i) Partitioning the set $\{ 2,3,..., k\}$ into subsets $V_1,...,V_l$ for $l \leq k-1$, (ii) associating with each $V_m$ a connected spanning graph, and (iii) connecting vertex 1 in all possible ways to at least one vertex in each $V_m$.

 Denoting the indices of vertices of subgraphs by $a$ and $b$ to avoid confusion, we see that:
 \begin{align*}
 &\left|\sum_{\substack{G \subseteq K_k \\ connected \\ spanning}} \prod_{(i,j)\in E(G)} \zeta(\xi_i, \xi_j) \right|
 \\& \leq \sum_{l = 1}^{k-1} \frac{1}{l!} \sum_{\substack{V_1,...,V_l:\\
 \text{disjoint},\\\cup_{m} V_m= \{2,...,k\}}} \prod_{m = 1}^l \left( \left| \sum_{\substack{G:V(G) = V_m \\ connected \\ spanning}} \prod_{\substack{(a,b) \in E(G)}} \hspace{-3mm} \zeta(\xi_a, \xi_b) \right| \left| \sum_{\substack{S_m \subseteq V_m \\ S_m \neq \emptyset}} \prod_{a \in S_m} \zeta(\xi_1, \xi_a)\right|\right).
 \end{align*}
 	We examine the last absolute value. First, we note that
 \begin{align}
 \sum_{\substack{S_m \subseteq V_m \\ S_m \neq \emptyset}} \prod_{a \in S_m} \zeta(\xi_1, \xi_a) = \left( \prod_{a \in V_m} \left(1 + \zeta(\xi_1, \xi_a)\right) \right) - 1.
 \end{align}
 Noting that $1+\zeta(\xi_1, \xi_a)$ is zero or one, we see (by looking at two cases: all $\zeta(\xi_1, \xi_a) = -1$ or there is at least one $a$ such that $\zeta(\xi_1, \xi_a) = 0$) that
 	\begin{align}
 \left|\sum_{\substack{S_m \subseteq V_m \\ S_m \neq \emptyset}} \prod_{a \in S_m} \zeta(\xi_1, \xi_a)\right| = \left|\left( \prod_{a \in V_m} \left(1 + \zeta(\xi_1, \xi_a)\right) \right) - 1\right| \leq \sum_{a \in V_m} |\zeta(\xi_1, \xi_a)|.
 \end{align}
	Altogether, this gives us
	\begin{align}
	&\left|\sum_{\substack{G \subseteq K_k \\ connected \\ spanning}} \prod_{(i,j)\in E(G)} \zeta(\xi_i, \xi_j) \right|\nonumber
	\\& \leq \sum_{l = 1}^{k-1} \frac{1}{l!} \sum_{\substack{V_1,...,V_l:\\
			\text{disjoint},\\\cup_{m} V_m= \{2,...,k\}}}  \prod_{m = 1}^l \left(\left| \sum_{\substack{G:V(G) = V_m \\ connected \\ spanning}} \prod_{\substack{(a,b) \in E(G)}} \hspace{-3mm} \zeta(\xi_a, \xi_b) \right| \left( \sum_{a \in V_m} |\zeta(\xi_1, \xi_a)|\right)\right). \label{eqn:abs_val_bound}
	\end{align}

	We can plug this expression into the statement we are trying to prove, and see that
	\begin{align*}
	&\sum_{k=2}^{N+1}  \sum_{\xi_2 \in \Gamma} \ldots \sum_{\xi_k \in \Gamma} \frac{1}{(k-1)!} \left|\sum_{\substack{G \subseteq K_k \\ connected \\ spanning}} \prod_{(i,j)\in E(G)} \zeta(\xi_i, \xi_j) \right| \prod_{i = 2}^k w(\xi_i)
	\\ &\leq
	\sum_{k=2}^{N+1} \sum_{\xi_2 \in \Gamma} \ldots \sum_{\xi_k \in \Gamma} \frac{1}{(k-1)!}
\left(	\sum_{l = 1}^{k-1} \frac{1}{l!} \sum_{\substack{V_1,...,V_l:\\
			\text{disjoint},\\\cup_{m} V_m= \{2,...,k\}}}  \prod_{m = 1}^l \left| \sum_{\substack{G:\\V(G) = V_m \\ connected \\ spanning}} \prod_{\substack{(a,b) \in E(G)}} \hspace{-3mm} \zeta(\xi_a, \xi_b) \right| \sum_{a \in V_m} |\zeta(\xi_1, \xi_a)|
\right)	\prod_{i = 2}^k w(\xi_i).
	\intertext{Rearranging the sums, possible because we keep all infinite sums in the same order and just interchange them with finite ones, we see this is equal to}
	&= \sum_{k=2}^{N+1} \frac{1}{(k-1)!}
	\sum_{l = 1}^{k-1} \frac{1}{l!} \\ &\hspace{5mm} \cdot \left(\sum_{\xi_2 \in \Gamma} \ldots \sum_{\xi_k \in \Gamma} \sum_{\substack{V_1,...,V_l:\\
			\text{disjoint},\\\cup_{m} V_m= \{2,...,k\}}}  \prod_{m = 1}^l \left(\left| \sum_{\substack{G:\\V(G) = V_m\\ connected \\ spanning}} \prod_{\substack{(a,b) \in E(G)}} \hspace{-3mm} \zeta(\xi_a, \xi_b) \right| \left(\sum_{a \in V_m} |\zeta(\xi_1, \xi_a)|\right) \right)
	\prod_{i = 2}^k w(\xi_i) \right).
	\end{align*}
		We now just examine the term in the outermost parentheses in this last expression, for which $k$ and~$l$ have already been fixed. Again because the sum over possible $V_m$ is finite for a fixed value of $k$, we can interchange it with the sums over the $\xi_i$.
We see that
	\newcommand{\m}{n}
	\begin{align*}
	&\sum_{\xi_2 \in \Gamma} \ldots \sum_{\xi_k \in \Gamma} \sum_{\substack{V_1,...,V_l:\\
			\text{disjoint},\\\cup_{m} V_m= \{2,...,k\}}}  \prod_{m = 1}^l \left(\left| \sum_{\substack{G:\\V(G) = V_m \\ connected \\ spanning}} \prod_{\substack{(a,b) \in E(G)}} \hspace{-3mm} \zeta(\xi_a, \xi_b) \right| \left(\sum_{a \in V_m} |\zeta(\xi_1, \xi_a)|\right) \right)
	\prod_{i = 2}^k w(\xi_i)
	\\ &= \sum_{\substack{V_1,...,V_l:\\
			\text{disjoint},\\\cup_{m} V_m= \{2,...,k\}}}  \sum_{\xi_2 \in \Gamma} \ldots \sum_{\xi_k \in \Gamma} \prod_{m = 1}^l \left(\left| \sum_{\substack{G:\\V(G) = V_m \\ connected \\ spanning}} \prod_{\substack{(a,b) \in E(G)}} \hspace{-3mm} \zeta(\xi_a, \xi_b) \right| \left(\sum_{a \in V_m} |\zeta(\xi_1, \xi_a)|\right) \right)
	\prod_{i = 2}^k w(\xi_i).
	\intertext{
		The next step is to specify the number of points in each $V_m$.  Let $|V_m| = \m_m$.
		Note the expression inside the product over $m = 1,...,l$ only depends on the $\xi_j$ such that $j \in V_m$; we can rewrite it to express this as }
	&= \sum_{\substack{V_1,...,V_l:\\
			\text{disjoint},\\\cup_{m} V_m= \{2,...,k\}}}  \prod_{m = 1}^\ell \left[ \sum_{\xi_{1}' \in \Gamma} \ldots \sum_{\xi_{\m_m}' \in \Gamma} \left| \sum_{\substack{G \subseteq K_{\m_m} \\ connected \\ spanning}} \prod_{\substack{(a,b) \in E(G)}} \hspace{-3mm} \zeta(\xi_a', \xi_b') \right|
	\left(  \sum_{a = 1}^{\m_m} |\zeta(\xi_1, \xi_a')|  \right)
	\left( \prod_{j = 1}^{\m_m} w(\xi_j) \right)\right].
	\end{align*}
		Now the term in square brackets has been written so that it does not depend on which vertices in $\{2,\ldots,k\}$ are in the sets $V_i'$, but rather only depends on the sizes  of these sets.  We see that we really only need to sum over the sizes of the $V_1'$,...,$V_l'$, taking into account the number of ways to make such a partition:
	\begin{align*}
	&= \hspace{-3mm} \sum_{\substack{\m_1,...,\m_l\\ \m_1 + ... + \m_l = k-1}}  \frac{(k-1)!}{\m_1! ... \m_l!} \\ & \hspace{10mm} \cdot
	\prod_{m = 1}^\ell \left[ \sum_{\xi_{1}' \in \Gamma} \ldots \sum_{\xi_{\m_m}' \in \Gamma} \left| \sum_{\substack{G \subseteq K_{\m_m} \\ connected \\ spanning}} \prod_{\substack{(a,b) \in E(G)}} \hspace{-3mm} \zeta(\xi_a', \xi_b') \right|
	\left(  \sum_{a = 1}^{\m_m} |\zeta(\xi_1, \xi_a')|  \right)
	\left( \prod_{j = 1}^{\m_m} w(\xi_j') \right)\right]
	\\&=\hspace{-3mm} \sum_{\substack{\m_1,...,\m_l\\ \m_1 + ... + \m_l = k-1}} (k-1)!
	\\ & \hspace{10mm} \cdot \prod_{i = 1}^\ell \left[ \sum_{\xi_{1}' \in \Gamma} \ldots \sum_{\xi_{\m_m}' \in \Gamma} \frac{1}{\m_m!}\left| \sum_{\substack{G \subseteq K_{\m_m} \\ connected \\ spanning}} \prod_{\substack{(a,b) \in E(G)}} \hspace{-3mm} \zeta(\xi_a', \xi_b') \right|
	\left(  \sum_{a = 1}^{\m_m} |\zeta(\xi_1, \xi_a')|  \right)
	\left( \prod_{j = 1}^{\m_m}w(\xi_j') \right)\right].
	\end{align*}

	Calling the term in square brackets above $B_{\m_m}$ for simplicity, noting that it only depends on $\m_m$,  we see that
	\begin{align*}
	\sum_{k=2}^{N+1} \sum_{\xi_2 \in \Gamma} &\ldots \sum_{\xi_k \in \Gamma} \frac{1}{(k-1)!} \left|\sum_{\substack{G \subseteq K_k \\ connected \\ spanning}} \prod_{(i,j)\in E(G)} \zeta(\xi_i, \xi_j)\right| \prod_{i = 2}^k w(\xi_i)
	\\ &\leq  \sum_{k=2}^{N+1} \frac{1}{(k-1)!}\sum_{l = 1}^{k-1} \frac{1}{l!}
	\sum_{\substack{\m_i,...,\m_l\\ \m_1 + ... + \m_l = k-1}} (k-1)!
	\prod_{m = 1}^\ell \left[  B_{\m_m} \right]
	\\&=  \sum_{k=2}^{N+1} \sum_{l = 1}^{k-1}
	\sum_{\substack{\m_1,...,\m_l\\ \m_1 + ... + \m_l = k-1}} \frac{1}{l!}
	\prod_{i = 1}^\ell \left[  B_{\m_m} \right]
	\\&\leq \sum_{l = 1}^N \sum_{\m_1 = 1}^N \sum_{\m_2 = 1}^N \ldots \sum_{\m_l = 1}^N \frac{1}{l!} \prod_{m = 1}^l \left[ B_{\m_m}\right]
	\\&\leq \sum_{l \geq 1} \frac{1}{l!} \prod_{m = 1}^l \left( \sum_{\m_m = 1}^N \left[B_{\m_m}\right] \right)
	\\& \leq \sum_{l \geq 1} \frac{1}{l!} \left(\sum_{m = 1}^N \left[B_{m}\right] \right)^l
	\\& = exp\left( \sum_{m = 1}^N \left[B_{m}\right] \right) - 1.
	\end{align*}
	We have shown that
	\begin{align}\label{eqn:boundbyexp}
	1 + \sum_{k=2}^{N+1} \sum_{\xi_2 \in \Gamma} \ldots \sum_{\xi_k \in \Gamma} \frac{1}{(k-1)!} \left|\sum_{\substack{G \subseteq K_k \\ connected \\ spanning}} \prod_{(i,j)\in E(G)} \zeta(\xi_i, \xi_j)\right| \prod_{i = 2}^k w(\xi_i) \leq exp\left( \sum_{m = 1}^N \left[B_{m}\right] \right).
	\end{align}

	It only remains to show that the term in parentheses above is at most $c|[\xi_i]|$, which we now do. Recall
	\begin{align}\label{eqn:dfnofB}
	\sum_{m = 1}^N \left[B_{m}\right] = \sum_{m=1}^N \sum_{\xi_{1}' \in \Gamma} \ldots \sum_{\xi_{m}' \in \Gamma} \frac{1}{m!}\left| \sum_{\substack{G \subseteq K_{m} \\ connected \\ spanning}} \prod_{\substack{(a,b) \in E(G)}} \hspace{-3mm} \zeta(\xi_a', \xi_b') \right|
	\left(  \sum_{a = 1}^{m} |\zeta(\xi_1, \xi_a')|  \right)
	\left( \prod_{j = 1}^{m} w(\xi_j')\right).
	\end{align}

	We note that $\xi_1$ appears in this expression only in the sum over index $a$.  We will do a change in variables here and replace $\xi_1$ here with $\xi^*$ (so we can avoid using $\xi_1'$, etc., in our proof). We now prove the following:

	\begin{claim}	\label{claim:boundbyc}
	For any $\xi^* \in \Gamma$,
	\begin{align}\label{eqn:boundbyc}
	\sum_{m=1}^N \sum_{\xi_{1} \in \Gamma} \ldots \sum_{\xi_{m} \in \Gamma} \frac{1}{m!}\left| \sum_{\substack{G \subseteq K_{m} \\ connected \\ spanning}} \prod_{\substack{(a,b) \in E(G)}} \hspace{-3mm} \zeta(\xi_a, \xi_b) \right|
	\left(  \sum_{i = 1}^{m} |\xi^*, \xi_i)|  \right)
	\left( \prod_{j = 1}^{m} w(\xi_j) \right) \leq c|[\xi^*]|.
	\end{align}
	\end{claim}
\noindent{\bf Proof of Claim. }
	To prove this, we will begin with the induction hypothesis (Equation~\ref{eqn:IHN}) and apply the same set of operations to both sides, ultimately yielding exactly this expression. We will multiply both sides of Equation~\ref{eqn:IHN} by $|\zeta (\xi^*, \xi_1)| \cdot w(\xi_1)$ and sum over all $\xi_1 \in \Gamma$.  When doing this, the left hand side of Equation~\ref{eqn:IHN} becomes
	\begin{align}
	& \sum_{\xi_1 \in \Gamma} \left[ 1+ \sum_{k=2}^N \sum_{\xi_2 \in \Gamma} \ldots \sum_{\xi_k \in \Gamma} \frac{1}{(k-1)!} \left|\sum_{\substack{G \subseteq K_k \\ connected \\ spanning}} \prod_{(i,j)\in E(G)} \zeta(\xi_i, \xi_j) \right| \prod_{j = 2}^k w(\xi_j) \right] \cdot |\zeta (\xi^*, \xi_1)| \cdot w(\xi_1)\nonumber
	\\ & = \left( \sum_{\xi_1 \in \Gamma} |\zeta (\xi^*, \xi_1)| \cdot w(\xi_1)\right) \nonumber
	\\ & \hspace{10mm} + \sum_{k = 2}^N \sum_{\xi_1 \in \Gamma} \sum_{\xi_2 \in \Gamma} \ldots \sum_{\xi_k \in \Gamma} \frac{1}{(k-1)!} \left|\sum_{\substack{G \subseteq K_k \\ connected \\ spanning}} \prod_{(i,j)\in E(G)} \zeta(\xi_i, \xi_j) \right| \prod_{j = 1}^k  w(\xi_j) \cdot |\zeta (\xi^*, \xi_1)|. \nonumber
	\intertext{Noting that the term in absolute value above is 1 when $k = 1$ because there is only one spanning subgraph of $K_1$ and it has no edges, we can rewrite this as}
	&= \sum_{k = 1}^N  \sum_{\xi_1 \in \Gamma} \sum_{\xi_2 \in \Gamma} \ldots \sum_{\xi_k \in \Gamma} \frac{1}{(k-1)!} \left|\sum_{\substack{G \subseteq K_k \\ connected \\ spanning}} \prod_{(i,j)\in E(G)} \zeta(\xi_i, \xi_j) \right| \prod_{j = 1}^k  w(\xi_j) \cdot |\zeta (\xi^*, \xi_1)|. \label{eqn:LHS-sum}
	\end{align}
	We note that $\xi_1$ does not play any special role in this expression, and by rearranging the summands we see that for any $a = 1,...,k$, it must hold that
	\begin{align*}
	\sum_{\xi_1 \in \Gamma} &\ldots \sum_{\xi_k \in \Gamma} \frac{1}{(k-1)!} \left|\sum_{\substack{G \subseteq K_k \\ connected \\ spanning}} \prod_{(i,j)\in E(G)} \zeta(\xi_i, \xi_j) \right| \prod_{j = 1}^k  w(\xi_j) |\zeta (\xi^*, \xi_1)|
	\\ &= \sum_{\xi_1 \in \Gamma} \ldots \sum_{\xi_k \in \Gamma} \frac{1}{(k-1)!} \left|\sum_{\substack{G \subseteq K_k \\ connected \\ spanning}} \prod_{(i,j)\in E(G)} \zeta(\xi_i, \xi_j) \right| \prod_{j = 1}^k  w(\xi_j) |\zeta (\xi^*, \xi_a)|.
\intertext{Summing over the possible values of $a$ and dividing by $k$, we obtain this is equal to}
&= \sum_{\xi_1 \in \Gamma} \ldots \sum_{\xi_k \in \Gamma} \frac{1}{(k-1)!} \left|\sum_{\substack{G \subseteq K_k \\ connected \\ spanning}} \prod_{(i,j)\in E(G)} \zeta(\xi_i, \xi_j) \right| \prod_{j = 1}^k  w(\xi_j) \left( \frac{1}{k} \sum_{a = 1}^k |\zeta (\xi^*, \xi_a)|\right).
	\end{align*}
	This implies the expression in Equation~\ref{eqn:LHS-sum} is equal to
	\begin{align*}
	\sum_{k = 1}^N \sum_{\xi_1 \in \Gamma} \sum_{\xi_2 \in \Gamma} \ldots \sum_{\xi_k \in \Gamma} \frac{1}{k!} \left|\sum_{\substack{G \subseteq K_k \\ connected \\ spanning}} \prod_{(i,j)\in E(G)} \zeta(\xi_i, \xi_j) \right| \prod_{j = 1}^k  w(\xi_j) \cdot \left( \sum_{a = 1}^k |\zeta (\xi^*, \xi_a)|\right).
	\end{align*}
	This is exactly the left hand side of Equation~\ref{eqn:boundbyc}, which we are trying to bound. When we  perform the same operations on the right hand side of the induction hypothesis (Equation~\ref{eqn:IHN}), multiplying by
	$|\zeta (\xi^*, \xi_1)| \cdot w(\xi_1)$ and summing over all $\xi_1 \in \Gamma$,  we obtain
	\begin{align*}
	\sum_{\xi_1 \in \Gamma} |\zeta (\xi^*, \xi_1)| \cdot w(\xi_1) \cdot e^{c |[\xi_1]|}.
	\end{align*}
	Recalling that $\zeta(\xi^*, \xi_1)$ is $0$ if $\xi^*$ and $\xi_1$ are compatible and -1 if they are incompatible, we see that, for any edge $e \in E(\Gtri)$,
	\begin{align*}
	\sum_{\xi_1 \in \Gamma} |\zeta (\xi^*, \xi_1)| \cdot w(\xi_1) \cdot e^{c |[\xi_1]|} &= \sum_{\substack{\xi_1 \in \Gamma \\ \xi_1, \xi^* \text{\ incompatible}}} w(\xi_1) \cdot e^{c |[\xi_1]|}
	\\&\leq \sum_{e \in |[\xi^*]|} \sum_{\substack{\xi_1 \in \Gamma \\ e \in \xi_1}} w(\xi_1) \cdot e^{c |[\xi_1]|}
	\\&= |[\xi^*]| \sum_{\substack{\xi_1 \in \Gamma \\ e \in \xi_1}} w(\xi_1) \cdot e^{c |[\xi_1]|}
	\end{align*}
	By the hypothesis of the theorem, this is at most $ c |[\xi^*]|$. This concludes our proof of Claim~\ref{claim:boundbyc}.  \qed

	Putting it all together, by Equation~\ref{eqn:zeta}, Equation~\ref{eqn:boundbyexp}, Equation~\ref{eqn:dfnofB}, and Claim~\ref{claim:boundbyc}, we see that
	\begin{align*}
	1+& \sum_{k=2}^{N+1}  \sum_{\xi_2 \in \Gamma} \ldots \sum_{\xi_k \in \Gamma} \frac{1}{(k-1)!} \left|\sum_{\substack{G \subseteq H_{\{\xi_1,\xi_2,...,\xi_k\}} \\ connected \\ spanning}} (-1)^{|E(G)|} \right| \prod_{i = 2}^k w(\xi_i) \\&=
	1+ \sum_{k=2}^{N+1}  \sum_{\xi_2 \in \Gamma} \ldots \sum_{\xi_k \in \Gamma} \frac{1}{(k-1)!} \left| \sum_{\substack{G \subseteq K_k \\ connected \\ spanning}} \prod_{(i,j) \in E(G)} \zeta(\xi_i, \xi_j) \right| \prod_{i = 2}^k w(\xi_i)
	\\&\leq  exp\left( \sum_{m = 1}^N \left[B_{m}\right] \right)
\\	&= exp\left(\sum_{k = 1}^N \sum_{\xi_1 \in \Gamma} \sum_{\xi_2 \in \Gamma} \ldots \sum_{\xi_k \in \Gamma} \frac{1}{k!} \left|\sum_{\substack{G \subseteq K_k \\ connected \\ spanning}} \prod_{(i,j)\in E(G)} \zeta(\xi_i, \xi_j) \right| \prod_{j = 1}^k  w(\xi_j) \cdot \left( \sum_{a = 1}^k |\zeta (\xi^*, \xi_a)|\right)\right)
	\\&\leq exp\left( c |[\xi^*]|\right).
	\end{align*}

	This concludes our  proof by induction. Taking the limit as $N$ goes to $\infty$ completes the proof of the theorem.

\end{proof}

We now prove Lemma~\ref{lem:less-than-c}, which was the key fact we needed for the proof of Theorem~\ref{thm:bdry-volume}.
Recall $\mathcal{X}$ is the set of all clusters comprised of polymers in $\Gamma$.
For a cluster $X \in \mathcal{X}$, recall that $|X|$ is the number of polymers in $X$, the support of $X$ is
\[\overline{X} = \cup_{\xi \in X} \xi , \]
and we defined
\[\Psi(X) = \frac{1}{|X|!} \left( \sum_{\substack{G \subseteq H_X: \\ connected, \\ spanning}} (-1)^{|E(G)|} \right) \left( \prod_{\xi \in X} w(\xi) \right).\]

\begin{lem}\label{lem:less-than-c}
Let $\Gamma$ be an infinite set of polymers $\xi \subseteq E(\Gtri)$ that is closed under translation and rotation.
If there is a constant $c$ such that for any edge $e \in E(\Gtri)$,
\begin{align}\label{eqn:suff-cond-edge-app}
\sum_{\substack{\xi \in \Gamma : \\ e \in \xi}} w(\xi) e^{c| [\xi]|} \leq c,
\end{align}
then for any edge $e\in E(\Gtri)$,
\begin{align*}
\sum_{\substack{X \in \mathcal{X} \\ e \in \overline{X}}} |\Psi(X)| \leq c.
\end{align*}
\end{lem}
\begin{proof}
	We first enumerate all clusters of size $k$ containing edge $e$. Consider cluster $X = (\xi_1,\xi_2,...,\xi_k)$ that contains edge $e$, and recall we defined clusters to be {\it ordered} multisets. Let $i$ be the smallest index such that $e \in \xi_i$. Let $X^i = (\xi_i, \xi_2,..., \xi_{i-1}, \xi_1, \xi_{i+1},...,\xi_k)$ be obtained by swapping $\xi_1$ and $\xi_i$ in $X$.  Note that $X^i$ now has $e$ in its first polymer.  Furthermore, the definition of $\Psi$ is independent of the order of cluster $X$, so $\Psi(X) = \Psi(X^i)$. We can use this fact to rewrite the sum of interest as a sum over clusters $X = (\xi_1,...,\xi_k)$ whose first polymer $\xi_1$ contains $e$.
	\[ \sum_{\substack{X \in \mathcal{X} \\ e \in \overline{X}}} |\Psi(X)| = \sum_{\substack{X \in \mathcal{X} \\ e \in \overline{X}}} |\Psi(X^i)| \leq  \sum_{\substack{X \in \mathcal{X}:\\ e \in \xi_1}} |X| |\Psi(X)| .
	\]
	We do not have equality in the above expression because we are overcounting clusters for which $e$ appears in multiple polymers.
	We can now expand this sum by explicitly summing over the size of the cluster and the polymers in it. Noting that $\Psi(X) = 0$ if $X$ is not a cluster,
	\[\sum_{\substack{X \in \mathcal{X} \\ e \in \overline{X}}} |\Psi(X)| \leq \sum_{k \geq 1} \sum_{\substack{\xi_1 \in \Gamma \\ e \in \xi_1}} \sum_{\xi_2 \in \Gamma} ... \sum_{\xi_k \in \Gamma} k |\Psi(\xi_1,\xi_2,\dots,\xi_k)|. \]
	Recalling the definition of $\Psi$, and noting that if $|X| = 1$ then $H_{X}$ has exactly one connected subgraph - the isolated vertex - and this subgraph has zero edges, we can factor this expression as
	\begin{align*}
	\sum_{k \geq 1} &\sum_{\substack{\xi_1 \in \Gamma \\ e \in \xi_1}} \sum_{\xi_2 \in \Gamma} ... \sum_{\xi_k \in \Gamma} k |\Psi(\xi_1,\xi_2,\dots,\xi_k)|
	\\&= \sum_{\substack{\xi_1 \in \Gamma \\ e \in \xi_1}} w(\xi_1) \left( 1 + \sum_{k \geq 2} \sum_{\xi_2 \in \Gamma}
	... \sum_{\xi_k \in \Gamma} \frac{1}{(k-1)!} \left|\sum_{\substack{G \subseteq H_{\{\xi_1,\xi_2,...,\xi_k\}}\\ connected \\ spanning}} (-1)^{|E(G)|}  \right| \prod_{i = 2}^k w(\xi_i) \right)
	\end{align*}
	By Lemma~\ref{lem:KP-inf}, the term in parentheses above is at most $e^{c|[\xi_1]|}$. We conclude, using our hypothesis of  Equation~\ref{eqn:suff-cond-edge-app}, that
	\begin{align*}
	\sum_{\substack{X \in \mathcal{X} \\ e \in \overline{X}}} |\Psi(X)| \leq \sum_{\substack{\xi_1 \in \Gamma \\ e \in \xi_1}} w(\xi_1) e^{c|[\xi_1]|} \leq c.
	\end{align*}
	This concludes the proof.
\end{proof}

\subsection{High-temperature Expansion}
\label{app:ht}

Recall that for a particle configuration $\sigma$ with boundary $\P$ and a vertex $i$ of $\Gtri$ on or inside $\P$,  we say that $\sigma_i = 1$ if the particle at $i$ in $\sigma$ has color $c_2$ and $\sigma_i = -1$ if this particle has color $c_2$.  For a fixed boundary $\P$, recall $\Omega_\P$ is all particle configurations with boundary $\P$ and the correct number of particles of each color, and $\overline{\Omega_\P}$ is all configurations with boundary $\P$ and any number of particles with each color.
Recall $E_\P$ is all edges with both endpoints on or inside $\P$ (all edges with both endpoints occupied by particles in $\sigma \in \overline{\Omega_\P}$).
For $\beta = \ln \sqrt{\gamma}$, we use the following fact in Section~\ref{sec:compression-small-gamma}:
\[
\sum_{\sigma \in \overline{\Omega_\P}} \prod_{(i,j) \in E_\P} e^{\beta \sigma_i\sigma_j} = \cosh(\beta)^{|E_\P|} 2^n \sum_{\substack{E \subseteq E_\mathcal{P}:\\even}} \left(\tanh \beta\right)^{|E|}.
\]
Here we prove this via the high-temperature expansion for the Ising model (the high-temperature expansion is standard in the statistical physics literature but less well-known outside this field, which is why we include the derivation here for the convenience of the interested reader).

\begin{proof}
	Recall the hyperbolic trigonometric functions are given by
	\begin{align*}
	\cosh \beta = \frac{e^{2\beta} + 1}{2e^{\beta}} ,\hspace{1cm}
	\sinh \beta = \frac{e^{2\beta} - 1}{2e^{\beta}},
	\hspace{1cm}\tanh\beta = \frac{\sinh \beta}{\cosh\beta} =  \frac{e^{2\beta} - 1}{e^{2\beta} +1}.
	\end{align*}
	We note that $\sigma_i \sigma_j$ only takes on two values, $+1$ and $-1$. For $\sigma_i \sigma_j \in \{-1,1\}$, it is true (verifiable by case analysis) that
	\[e^{\beta \sigma_i \sigma_j} = \cosh \beta + \sigma_i \sigma_j \sinh \beta = \cosh \beta \left(1+  \sigma_i \sigma_j \tanh \beta  \right).\]
	Using this, we see that
	\begin{align*}
	\prod_{ (i,j) \in E_\P} e^{\beta \sigma_i \sigma_j}
	&= \prod_{ (i,j) \in E_\P} \cosh \beta \left(1+ \sigma_i \sigma_j \tanh \beta  \right)
	\\&= \left( \cosh \beta\right)^{|E_\P|}  \prod_{ (i,j) \in E_\P} \left(1+ \sigma_i \sigma_j \tanh \beta  \right).
	\end{align*}
	Using the fact that $\prod_{i \in S}  (1+x_i) = \sum_{\overline{S} \subseteq S} \prod_{i \in \overline{S}} x_i$, we rewrite this as
	\begin{align*}
	\prod_{ (i,j) \in E_\P} e^{\beta \sigma_i \sigma_j}
	= \left( \cosh \beta\right)^{|E_\P|}  \sum_{E \subseteq E_\P} \prod_{(i,j)\in E}\sigma_i \sigma_j \tanh \beta.
	\end{align*}
	We now consider the sum over all $\sigma \in \overline{\Omega_\P}$ and rearrange terms.
	\begin{align*}
	\sum_{\sigma \in \overline{\Omega_\P}} \prod_{(i,j) \in E_\P} e^{\beta \sigma_i\sigma_j}
	& = \sum_{\sigma \in \overline{\Omega_\P}} \left( \cosh \beta\right)^{|E_\P|}  \sum_{E \subseteq E_\P} \prod_{(i,j)\in E}\sigma_i \sigma_j \tanh \beta
	\\&= \left( \cosh \beta\right)^{|E_\P|} \sum_{E \subseteq E_\P} (\tanh\beta)^{|E|} \sum_{\sigma \in \overline{\Omega_\P}}\prod_{(i,j)\in E}\sigma_i \sigma_j .
	\end{align*}
	For a given set $E \subseteq E_\P$ and a given vertex $i$ on or inside $\P$, let $I(i,E)$ be the number of edges in $E$ that are incident on $i$; note $0 \leq I(i,E) \leq 6$. For a  given $E$, we see that if $V_\P$ is all vertices of $\Gtri$ on or inside $\P$,
	\newcommand{\VP}{V_\P}
	\begin{align*}
	\prod_{(i,j)\in E} \sigma_i \sigma_j = \prod_{i \in \VP} \sigma_i^{I(i,E)},
	\end{align*}
	and it follows that
	\begin{align*}
	\sum_{\sigma \in \overline{\Omega_\P}} \prod_{(i,j) \in E_\P} e^{\beta \sigma_i\sigma_j} = \left( \cosh \beta\right)^{|E_\P|} \sum_{E \subseteq E_\P} (\tanh\beta)^{|E|} \sum_{\sigma \in \overline{\Omega_\P}} \prod_{i \in \VP} \sigma_i^{I(i,E)}
	\end{align*}
	Consider the summand corresponding to a given set $E \subseteq E_\P$.
	Suppose there  is a vertex $j \in \VP$ such that $I(j,E)$ is odd.  The sum over configurations $\sigma \in \overline{\Omega_\P}$ with $\sigma_j = +1$ is the opposite of the sum over configurations $\sigma$ with $\sigma_j = -1$, and thus such configurations contribute zero to this sum. That is, for such an $E$ with an odd number of edges incident on vertex $j \in \VP$,
	\begin{align*}
	(\tanh\beta)^{|E|} \sum_{\sigma \in \overline{\Omega_\P}} \prod_{i \in \VP} \sigma_i^{I(i,E)}
	& = (\tanh\beta)^{|E|}\left(
	 \sum_{\substack{\sigma \in \overline{\Omega_\P}: \\ \sigma_j = 1}} \prod_{i \in \VP} \sigma_i^{I(i,E)} +
	 	 \sum_{\substack{\sigma \in \overline{\Omega_\P}: \\ \sigma_j = -1}} \prod_{i \in \VP} \sigma_i^{I(i,E)}
	 \right)
	 \\&=  (\tanh\beta)^{|E|}\left(
	 (1)^{I(j,E)}\sum_{\substack{\sigma \in \overline{\Omega_\P}: \\ \sigma_j = 1}} \prod_{\substack{i \in \VP:\\i \neq j}} \sigma_i^{I(i,E)} + (-1)^{I(j,E)}
	 \sum_{\substack{\sigma \in \overline{\Omega_\P}: \\ \sigma_j = -1}} \prod_{i \in \VP} \sigma_i^{I(i,E)}
	 \right)
	 \\& = (\tanh\beta)^{|E|}\left(
	 \sum_{\substack{\sigma \in \overline{\Omega_\P}: \\ \sigma_j = 1}} \prod_{\substack{i \in \VP:\\i \neq j}} \sigma_i^{I(i,E)} -
	 \sum_{\substack{\sigma \in \overline{\Omega_\P}: \\ \sigma_j = -1}} \prod_{\substack{i \in \VP:\\i \neq j}} \sigma_i^{I(i,E)}
	 \right)
	 \\&= 0.
	\end{align*}
	  Recall we say a subset $E \subseteq E_\P$ is {\it even} if it has an even number of edges incident on every vertex of $\Gtri$, that is, if  $I(i,E)$ is even for all $i \in \VP$ ($I(i,E) = 0$ whenever $i \notin \VP$). These are the only subsets $E$ with nonzero contributions to the above sum.
	 We see that
	 \begin{align*}
	 	\sum_{\sigma \in \overline{\Omega_\P}} \prod_{(i,j) \in E_\P} e^{\beta \sigma_i\sigma_j} = \left( \cosh \beta\right)^{|E_\P|} \sum_{\substack{E \subseteq E_\P: \\ E \text{ even}}} (\tanh\beta)^{|E|} \sum_{\sigma \in \overline{\Omega_\P}} \prod_{i \in \VP} \sigma_i^{I(i,E)}
	 \end{align*}
	 When $E$ is even, for any $\sigma\in \overline{\Omega_\P}$ we have $\prod_{i \in \VP} \sigma_i^{I(i,E)} = 1$. As the number of configurations in $\overline{\Omega_\P}$ is $2^n$ where $n$ is the number of vertices of $\Gtri$ on or inside $\P$, we see that
	 \begin{align*}
	\sum_{\sigma \in \overline{\Omega_\P}} \prod_{(i,j) \in E_\P} e^{\beta \sigma_i\sigma_j}
	&=  \left( \cosh \beta\right)^{|E_\P|} \sum_{\substack{E \subseteq E_\P: \\ E \text{ even}}} (\tanh\beta)^{|E|} \sum_{\sigma \in \overline{\Omega_\P}} 1
	\\&= \left( \cosh \beta\right)^{|E_\P|} \sum_{\substack{E \subseteq E_\P: \\ E \text{ even}}} (\tanh\beta)^{|E|} 2^n.
	\end{align*}
	Rearranging terms give the result we are trying to prove.
\end{proof}

\end{document}